\pdfoutput=1
\documentclass[11pt]{article}
\usepackage[T1]{fontenc}
\usepackage[utf8]{inputenc}
\usepackage{lmodern}
\usepackage{microtype}
\usepackage[margin=1in]{geometry}
\usepackage{amsmath,amssymb,amsthm,mathtools}
\usepackage{booktabs}
\usepackage{array}
\usepackage{enumitem}
\usepackage{graphicx}
\usepackage[hidelinks]{hyperref}

\newtheorem{theorem}{Theorem}[section]
\newtheorem{lemma}[theorem]{Lemma}
\newtheorem{proposition}[theorem]{Proposition}
\newtheorem{corollary}[theorem]{Corollary}
\newtheorem{fact}[theorem]{Fact}
\theoremstyle{definition}
\newtheorem{definition}[theorem]{Definition}
\newtheorem{openproblem}[theorem]{Open Problem}
\theoremstyle{remark}
\newtheorem{remark}[theorem]{Remark}

\newcommand{\R}{\mathbb{R}}
\newcommand{\E}{\mathbb{E}}
\newcommand{\Prob}{\mathbb{P}}
\newcommand{\dk}{d_k}
\newcommand{\dv}{d_v}
\newcommand{\vrho}{\varrho}
\newcommand{\eps}{\varepsilon}
\newcommand{\norm}[1]{\lVert #1\rVert}
\newcommand{\inner}[2]{\langle #1,#2\rangle}
\newcommand{\Attn}{\operatorname{Attn}}
\newcommand{\rad}{\operatorname{rad}}
\newcommand{\Dud}{\operatorname{Dud}}
\newcommand{\diam}{\operatorname{diam}}
\newcommand{\gam}{\gamma}
\newcommand{\disc}{\operatorname{disc}}
\newcommand{\INDEX}{\mathrm{INDEX}}
\newcommand{\tbd}{\textendash}
\newcommand{\vEaLBands}{0--5, 6--11, 12--18, 19--24, 25--31}
\newcommand{\vEaLGroup}{4}
\newcommand{\vEaLKvHeads}{256}
\newcommand{\vEaLLayers}{32}
\newcommand{\vEaLPromptsMeasured}{212}
\newcommand{\vEaLRows}{54{,}272}
\newcommand{\vEaLSkipped}{38}
\newcommand{\vEaLSkippedShort}{2}
\newcommand{\vEaLerhoMax}{$8.05\!\times\!10^{17}$}
\newcommand{\vEaLerhoPnn}{$1.69\!\times\!10^{17}$}
\newcommand{\vEaLfracHi}{100\%}
\newcommand{\vEaLfracLo}{0\%}
\newcommand{\vEaLfracSampling}{0\%}
\newcommand{\vEaLgqaHi}{1.677}
\newcommand{\vEaLgqaLo}{1.388}
\newcommand{\vEaLgqaRatio}{1.2}
\newcommand{\vEaLkappaMed}{24.6}
\newcommand{\vEaLqA}{39.586}
\newcommand{\vEaLqB}{41.153}
\newcommand{\vEaLqC}{42.744}
\newcommand{\vEaLqD}{39.774}
\newcommand{\vEaLqE}{42.109}
\newcommand{\vEaLrhoControl}{37.2}
\newcommand{\vEaLrhoMax}{41.230}
\newcommand{\vEaLrhoMaxHead}{72.0}
\newcommand{\vEaLrhoMinHead}{30.9}
\newcommand{\vEaLrhoNatural}{42.0}
\newcommand{\vEaLrhoPninetyHead}{49.6}
\newcommand{\vEaLrhoPnn}{39.667}
\newcommand{\vEaLrhoPtenHead}{36.1}
\newcommand{\vEaLrhoRound}{41}
\newcommand{\vEaLtrimFour}{0.968}
\newcommand{\vEaLtrimThirtytwo}{0.928}
\newcommand{\vEaLuErhoMax}{$1.43\!\times\!10^{23}$}
\newcommand{\vEaLuRhoMax}{53.310}
\newcommand{\vEaLvrhoMed}{1.68}
\newcommand{\vEaQBands}{0--4, 5--10, 11--15, 16--21, 22--27}
\newcommand{\vEaQBiasGapLarge}{6\%}
\newcommand{\vEaQBiasHeadsLarge}{7}
\newcommand{\vEaQBiasMax}{920.337}
\newcommand{\vEaQGroup}{7}
\newcommand{\vEaQHeadDim}{128}
\newcommand{\vEaQKvHeads}{112}
\newcommand{\vEaQLayers}{28}
\newcommand{\vEaQPromptsMeasured}{223}
\newcommand{\vEaQRows}{24{,}976}
\newcommand{\vEaQSkipped}{27}
\newcommand{\vEaQerhoMax}{$2.88\!\times\!10^{23}$}
\newcommand{\vEaQerhoPnn}{$9.22\!\times\!10^{21}$}
\newcommand{\vEaQfracHi}{100\%}
\newcommand{\vEaQfracLo}{0\%}
\newcommand{\vEaQfracSampling}{0\%}
\newcommand{\vEaQgqaHi}{2.415}
\newcommand{\vEaQgqaLo}{1.722}
\newcommand{\vEaQgqaRatio}{1.39}
\newcommand{\vEaQkappaMed}{22.1}
\newcommand{\vEaQqA}{44.716}
\newcommand{\vEaQqB}{53.935}
\newcommand{\vEaQqC}{60.708}
\newcommand{\vEaQqD}{57.344}
\newcommand{\vEaQqE}{53.062}
\newcommand{\vEaQrhoControl}{46.5}
\newcommand{\vEaQrhoMax}{54.015}
\newcommand{\vEaQrhoMaxHead}{562.5}
\newcommand{\vEaQrhoMinHead}{30.3}
\newcommand{\vEaQrhoNatural}{55.6}
\newcommand{\vEaQrhoPninetyHead}{73.5}
\newcommand{\vEaQrhoPnn}{50.575}
\newcommand{\vEaQrhoPtenHead}{41.9}
\newcommand{\vEaQrhoRound}{54}
\newcommand{\vEaQtrimFour}{0.970}
\newcommand{\vEaQtrimThirtytwo}{0.929}
\newcommand{\vEaQuErhoMax}{$1.91\!\times\!10^{30}$}
\newcommand{\vEaQuRhoMax}{69.719}
\newcommand{\vEaQvrhoMed}{2.41}
\newcommand{\vEbLBeatsRealAll}{340}
\newcommand{\vEbLBeatsRealClean}{102}
\newcommand{\vEbLBkvOversizeMax}{1.35}
\newcommand{\vEbLBkvRandA}{0.718}
\newcommand{\vEbLBkvRandB}{0.634}
\newcommand{\vEbLBkvRandC}{0.569}
\newcommand{\vEbLBkvRealA}{0.372}
\newcommand{\vEbLBkvRealB}{0.293}
\newcommand{\vEbLBkvRealC}{0.210}
\newcommand{\vEbLCapDumps}{18}
\newcommand{\vEbLCapHeads}{9}
\newcommand{\vEbLCellsRealAll}{540}
\newcommand{\vEbLCellsRealClean}{180}
\newcommand{\vEbLCleanDumps}{18}
\newcommand{\vEbLDumps}{54}
\newcommand{\vEbLHardMaxDumps}{27}
\newcommand{\vEbLHardMaxHeads}{0}
\newcommand{\vEbLHhRandA}{0.741}
\newcommand{\vEbLHhRandB}{0.725}
\newcommand{\vEbLHhRandC}{0.634}
\newcommand{\vEbLHhRealA}{0.174}
\newcommand{\vEbLHhRealB}{0.104}
\newcommand{\vEbLHhRealC}{0.047}
\newcommand{\vEbLKeptDumps}{27}
\newcommand{\vEbLKeptHeads}{9}
\newcommand{\vEbLLengthsMax}{32768}
\newcommand{\vEbLLengthsMin}{3957}
\newcommand{\vEbLOhCells}{135}
\newcommand{\vEbLOhRandA}{0.712}
\newcommand{\vEbLOhRandB}{0.635}
\newcommand{\vEbLOhRandC}{0.568}
\newcommand{\vEbLOhRealA}{0.374}
\newcommand{\vEbLOhRealB}{0.288}
\newcommand{\vEbLOhRealC}{0.202}
\newcommand{\vEbLOhWins}{73}
\newcommand{\vEbLOsbCells}{135}
\newcommand{\vEbLOsbRandA}{0.719}
\newcommand{\vEbLOsbRandB}{0.636}
\newcommand{\vEbLOsbRandC}{0.584}
\newcommand{\vEbLOsbRealA}{0.374}
\newcommand{\vEbLOsbRealB}{0.285}
\newcommand{\vEbLOsbRealC}{0.186}
\newcommand{\vEbLOsbWins}{79}
\newcommand{\vEbLOursOversizeMax}{1.10}
\newcommand{\vEbLPqRandA}{0.742}
\newcommand{\vEbLPqRandB}{0.633}
\newcommand{\vEbLPqRandC}{0.514}
\newcommand{\vEbLPqRealA}{0.391}
\newcommand{\vEbLPqRealB}{0.273}
\newcommand{\vEbLPqRealC}{0.171}
\newcommand{\vEbLRatioMed}{0.989}
\newcommand{\vEbLRatioQhi}{1.105}
\newcommand{\vEbLRatioQlo}{0.842}
\newcommand{\vEbLRealQueriesMax}{512}
\newcommand{\vEbLTailBeats}{49}
\newcommand{\vEbLTailCells}{80}
\newcommand{\vEbLTailDumps}{18}
\newcommand{\vEbLTailKeptDumps}{8}
\newcommand{\vEbLTailRhoMax}{38.8}
\newcommand{\vEbLTercRatioA}{0.965}
\newcommand{\vEbLTercRatioB}{0.991}
\newcommand{\vEbLTercRatioC}{0.993}
\newcommand{\vEbLThinRandA}{0.724}
\newcommand{\vEbLThinRandB}{0.640}
\newcommand{\vEbLThinRandC}{0.566}
\newcommand{\vEbLThinRealA}{0.397}
\newcommand{\vEbLThinRealB}{0.305}
\newcommand{\vEbLThinRealC}{0.199}
\newcommand{\vEbLUnifRandA}{0.719}
\newcommand{\vEbLUnifRandB}{0.640}
\newcommand{\vEbLUnifRandC}{0.558}
\newcommand{\vEbLUnifRealA}{0.382}
\newcommand{\vEbLUnifRealB}{0.288}
\newcommand{\vEbLUnifRealC}{0.190}
\newcommand{\vEbQBeatsRealAll}{273}
\newcommand{\vEbQBeatsRealClean}{35}
\newcommand{\vEbQBkvOversizeMax}{1.30}
\newcommand{\vEbQBkvRandA}{1.050}
\newcommand{\vEbQBkvRandB}{0.989}
\newcommand{\vEbQBkvRandC}{0.902}
\newcommand{\vEbQBkvRealA}{0.575}
\newcommand{\vEbQBkvRealB}{0.453}
\newcommand{\vEbQBkvRealC}{0.300}
\newcommand{\vEbQCapDumps}{18}
\newcommand{\vEbQCapHeads}{9}
\newcommand{\vEbQCellsRealAll}{540}
\newcommand{\vEbQCellsRealClean}{90}
\newcommand{\vEbQCleanDumps}{9}
\newcommand{\vEbQDumps}{54}
\newcommand{\vEbQHardMaxDumps}{26}
\newcommand{\vEbQHardMaxHeads}{3}
\newcommand{\vEbQHhRandA}{1.055}
\newcommand{\vEbQHhRandB}{1.014}
\newcommand{\vEbQHhRandC}{0.976}
\newcommand{\vEbQHhRealA}{0.362}
\newcommand{\vEbQHhRealB}{0.190}
\newcommand{\vEbQHhRealC}{0.046}
\newcommand{\vEbQKeptDumps}{28}
\newcommand{\vEbQKeptHeads}{6}
\newcommand{\vEbQLengthsMax}{32768}
\newcommand{\vEbQLengthsMin}{4096}
\newcommand{\vEbQOhCells}{140}
\newcommand{\vEbQOhRandA}{1.053}
\newcommand{\vEbQOhRandB}{0.996}
\newcommand{\vEbQOhRandC}{0.897}
\newcommand{\vEbQOhRealA}{0.573}
\newcommand{\vEbQOhRealB}{0.449}
\newcommand{\vEbQOhRealC}{0.311}
\newcommand{\vEbQOhWins}{64}
\newcommand{\vEbQOsbCells}{140}
\newcommand{\vEbQOsbRandA}{1.034}
\newcommand{\vEbQOsbRandB}{0.985}
\newcommand{\vEbQOsbRandC}{0.900}
\newcommand{\vEbQOsbRealA}{0.574}
\newcommand{\vEbQOsbRealB}{0.444}
\newcommand{\vEbQOsbRealC}{0.305}
\newcommand{\vEbQOsbWins}{70}
\newcommand{\vEbQOursOversizeMax}{1.03}
\newcommand{\vEbQPqRandA}{1.069}
\newcommand{\vEbQPqRandB}{0.995}
\newcommand{\vEbQPqRandC}{0.952}
\newcommand{\vEbQPqRealA}{0.636}
\newcommand{\vEbQPqRealB}{0.458}
\newcommand{\vEbQPqRealC}{0.283}
\newcommand{\vEbQRatioMed}{1.002}
\newcommand{\vEbQRatioQhi}{1.081}
\newcommand{\vEbQRatioQlo}{0.917}
\newcommand{\vEbQRealQueriesMax}{896}
\newcommand{\vEbQSeeds}{10}
\newcommand{\vEbQTailBeats}{43}
\newcommand{\vEbQTailCells}{90}
\newcommand{\vEbQTailDumps}{18}
\newcommand{\vEbQTailKeptDumps}{9}
\newcommand{\vEbQTailRhoMax}{50.6}
\newcommand{\vEbQTercRatioA}{1.002}
\newcommand{\vEbQTercRatioB}{1.019}
\newcommand{\vEbQTercRatioC}{0.996}
\newcommand{\vEbQThinRandA}{1.049}
\newcommand{\vEbQThinRandB}{0.991}
\newcommand{\vEbQThinRandC}{0.889}
\newcommand{\vEbQThinRealA}{0.578}
\newcommand{\vEbQThinRealB}{0.454}
\newcommand{\vEbQThinRealC}{0.291}
\newcommand{\vEbQUnifRandA}{1.038}
\newcommand{\vEbQUnifRandB}{0.981}
\newcommand{\vEbQUnifRandC}{0.919}
\newcommand{\vEbQUnifRealA}{0.550}
\newcommand{\vEbQUnifRealB}{0.490}
\newcommand{\vEbQUnifRealC}{0.291}
\newcommand{\vEcLBkvPostDiv}{1.000}
\newcommand{\vEcLBkvRatio}{1.000}
\newcommand{\vEcLBody}{Uniform sampling & $n/16$ & 0.405 & 0.415 & 1.000 & 54 of 54 & 8 & 0.08 \\
BalanceKV & $n/16$ & 0.427 & 0.426 & 1.000 & 45 of 54 & 8 & 0.10 \\
\textbf{Halving, re-centred subsample} & $n/16$ & 0.410 & 0.413 & 1.000 & 54 of 54 & 4 & 0.05 \\
Heavy hitters (oracle) & $n/16$ & 0.107 & 0.135 & 1.056 & 45 of 54 & 22 & 0.18 \\
Uniform sampling & $n/4$ & 0.290 & 0.300 & 1.000 & 54 of 54 & 6 & 0.09 \\
BalanceKV & $n/4$ & 0.323 & 0.341 & 1.000 & 45 of 54 & 8 & 0.07 \\
\textbf{Halving, re-centred subsample} & $n/4$ & 0.291 & 0.265 & 1.000 & 45 of 54 & 18 & 0.11 \\
Heavy hitters (oracle) & $n/4$ & 0.050 & 0.054 & 1.083 & 27 of 54 & 147 & 0.58 \\}
\newcommand{\vEcLCacheMax}{8192}
\newcommand{\vEcLCacheMin}{3957}
\newcommand{\vEcLHhBEq}{4}
\newcommand{\vEcLHhBGt}{27}
\newcommand{\vEcLHhBHeads}{9}
\newcommand{\vEcLHhBHeadsAbove}{5}
\newcommand{\vEcLHhBLt}{14}
\newcommand{\vEcLHhCEq}{3}
\newcommand{\vEcLHhCGt}{18}
\newcommand{\vEcLHhCHeads}{9}
\newcommand{\vEcLHhCHeadsAbove}{6}
\newcommand{\vEcLHhCLt}{6}
\newcommand{\vEcLHhPostDiv}{1.028}
\newcommand{\vEcLHhRatio}{1.075}
\newcommand{\vEcLHhRatioAll}{1.000}
\newcommand{\vEcLOsbBEq}{11}
\newcommand{\vEcLOsbBGt}{23}
\newcommand{\vEcLOsbBLt}{20}
\newcommand{\vEcLOsbPostDiv}{0.995}
\newcommand{\vEcLOsbRatio}{1.000}
\newcommand{\vEcLUnifPostDiv}{1.000}
\newcommand{\vEcLUnifRatio}{1.000}
\newcommand{\vEcQBkvPostDiv}{1.005}
\newcommand{\vEcQBkvRatio}{1.000}
\newcommand{\vEcQBody}{Uniform sampling & $n/16$ & 0.582 & 0.567 & 0.984 & 54 of 54 & 4 & 0.09 \\
BalanceKV & $n/16$ & 0.645 & 0.634 & 1.006 & 54 of 54 & 4 & 0.08 \\
\textbf{Halving, re-centred subsample} & $n/16$ & 0.651 & 0.680 & 0.993 & 54 of 54 & 4 & 0.07 \\
Heavy hitters (oracle) & $n/16$ & 0.212 & 0.254 & 1.110 & 36 of 54 & 98 & 0.44 \\
Uniform sampling & $n/4$ & 0.506 & 0.494 & 1.000 & 45 of 54 & 38 & 0.20 \\
BalanceKV & $n/4$ & 0.508 & 0.528 & 1.000 & 54 of 54 & 10 & 0.11 \\
\textbf{Halving, re-centred subsample} & $n/4$ & 0.503 & 0.514 & 1.000 & 54 of 54 & 5 & 0.09 \\
Heavy hitters (oracle) & $n/4$ & 0.067 & 0.087 & 1.125 & 18 of 54 & 256 & 1.00 \\}
\newcommand{\vEcQCacheMax}{8192}
\newcommand{\vEcQCacheMin}{4114}
\newcommand{\vEcQHhBEq}{10}
\newcommand{\vEcQHhBGt}{23}
\newcommand{\vEcQHhBHeads}{9}
\newcommand{\vEcQHhBHeadsAbove}{8}
\newcommand{\vEcQHhBLt}{3}
\newcommand{\vEcQHhCEq}{1}
\newcommand{\vEcQHhCGt}{16}
\newcommand{\vEcQHhCHeads}{9}
\newcommand{\vEcQHhCHeadsAbove}{8}
\newcommand{\vEcQHhCLt}{1}
\newcommand{\vEcQHhPostDiv}{1.071}
\newcommand{\vEcQHhRatio}{1.121}
\newcommand{\vEcQHhRatioAll}{1.000}
\newcommand{\vEcQOsbBEq}{3}
\newcommand{\vEcQOsbBGt}{22}
\newcommand{\vEcQOsbBLt}{29}
\newcommand{\vEcQOsbOpenB}{0.651}
\newcommand{\vEcQOsbPostDiv}{1.010}
\newcommand{\vEcQOsbRatio}{1.000}
\newcommand{\vEcQRows}{54}
\newcommand{\vEcQSteps}{256}
\newcommand{\vEcQTokenLength}{8192}
\newcommand{\vEcQUnifOpenB}{0.582}
\newcommand{\vEcQUnifPostDiv}{1.004}
\newcommand{\vEcQUnifRatio}{1.000}
\newcommand{\vEdCensoredFits}{0}
\newcommand{\vEdConstructionCells}{12}
\newcommand{\vEdFits}{60}
\newcommand{\vEdFullBody}{Sphere, $\dv=1$ & Uniform sampling & 0.54 [0.40, 0.60] & 0.81 [0.76, 0.93] & 0.57 [0.40, 0.72] \\
 & BalanceKV & 0.48 [0.37, 0.59] & \textbf{0.97$^{\dagger}$} [0.79, 1.19] & 0.68 [0.44, 1.00] \\
 & Kernel halving & 0.42 [0.33, 0.53] & 0.90 [0.83, 1.07] & 0.58 [0.48, 0.81] \\
 & \textbf{Halving, re-centred subsample} & 0.64 [0.44, 0.84] & \textbf{1.08} [0.91, 1.14] & \textbf{0.98} [0.72, 1.15] \\
 & \textbf{Halving} & 0.58 [0.32, 0.86] & \textbf{1.10} [0.98, 1.26] & \textbf{1.03} [0.82, 1.22] \\
Sphere, $\dv=8$ & Uniform sampling & 0.34 [0.31, 0.38] & 0.58 [0.54, 0.63] & 0.27 [0.25, 0.27] \\
 & BalanceKV & 0.36 [0.25, 0.47] & 0.65 [0.60, 0.71] & 0.35 [0.20, 0.39] \\
 & Kernel halving & 0.34 [0.30, 0.39] & 0.65 [0.57, 0.69] & 0.33 [0.31, 0.40] \\
 & \textbf{Halving, re-centred subsample} & 0.42 [0.34, 0.51] & 0.90 [0.84, 0.93] & 0.60 [0.50, 0.74] \\
 & \textbf{Halving} & 0.41 [0.35, 0.45] & 0.90 [0.86, 0.96] & 0.66 [0.57, 0.75] \\
Clustered, $\dv=1$ & Uniform sampling & 0.42 [0.31, 0.59] & 0.57 [0.45, 0.72] & 0.44 [0.32, 0.73] \\
 & BalanceKV & 0.69 [0.25, 1.19] & 0.60 [0.44, 0.92] & 0.37 [0.21, 0.73] \\
 & Kernel halving & 0.42 [0.25, 0.79] & 0.52 [0.41, 0.67] & 0.50 [0.34, 0.63] \\
 & \textbf{Halving, re-centred subsample} & 0.72$^{\dagger}$ [0.17, 1.49] & 0.69$^{\dagger}$ [0.52, 0.86] & 0.70$^{\dagger}$ [0.54, 1.01] \\
 & \textbf{Halving} & 0.83$^{\dagger}$ [0.40, 1.47] & 0.77$^{\dagger}$ [0.54, 1.17] & 0.63$^{\dagger}$ [0.45, 1.06] \\
Clustered, $\dv=8$ & Uniform sampling & 0.33 [0.29, 0.43] & 0.41 [0.37, 0.49] & 0.27 [0.23, 0.32] \\
 & BalanceKV & 0.33 [0.22, 0.45] & 0.55 [0.26, 0.66] & 0.28 [0.03, 0.43] \\
 & Kernel halving & 0.37 [0.29, 0.41] & 0.43 [0.34, 0.54] & 0.29 [0.22, 0.38] \\
 & \textbf{Halving, re-centred subsample} & 0.51$^{\dagger}$ [0.44, 0.70] & 0.57 [0.52, 0.68] & 0.51 [0.34, 0.65] \\
 & \textbf{Halving} & 0.52 [0.41, 0.78] & 0.60 [0.49, 0.73] & 0.49 [0.36, 0.66] \\}
\newcommand{\vEdFullDaggered}{8}
\newcommand{\vEdFullWithin}{5}
\newcommand{\vEdGainHi}{5.3}
\newcommand{\vEdGainLo}{3.2}
\newcommand{\vEdMaxSlope}{1.10}
\newcommand{\vEdMonotone}{31\%}
\newcommand{\vEdOffGridFrac}{2\%}
\newcommand{\vEdOffGridTotal}{1800}
\newcommand{\vEdOhWithin}{2}
\newcommand{\vEdOsbRatioRa}{0.18}
\newcommand{\vEdOsbRatioRb}{0.20}
\newcommand{\vEdOsbRatioRc}{0.28}
\newcommand{\vEdOsbWithin}{2}
\newcommand{\vEdOursHighest}{12 of 12}
\newcommand{\vEdRatioBody}{Sphere, $\dk=2$, $\dv=1$ & 0.11 & 0.13 & 0.15 \\
Sphere, $\dk=2$, $\dv=8$ & 0.19 & 0.21 & 0.24 \\
Sphere, $\dk=8$, $\dv=1$ & 0.14 & 0.21 & 0.33 \\
Sphere, $\dk=8$, $\dv=8$ & 0.25 & 0.43 & 0.73 \\
Sphere, $\dk=64$, $\dv=1$ & 0.26 & 0.43 & 1.02 \\
Sphere, $\dk=64$, $\dv=8$ & 0.31 & 0.48 & 0.98 \\
Clustered, $\dk=2$, $\dv=1$ & \tbd & 0.12 & 0.26 \\
Clustered, $\dk=2$, $\dv=8$ & 0.14 & 0.17 & 0.29 \\
Clustered, $\dk=8$, $\dv=1$ & 0.16 & 0.15 & 0.20 \\
Clustered, $\dk=8$, $\dv=8$ & 0.15 & 0.16 & 0.27 \\
Clustered, $\dk=64$, $\dv=1$ & 0.31 & 0.28 & 0.47 \\
Clustered, $\dk=64$, $\dv=8$ & 0.26 & 0.36 & 0.51 \\}
\newcommand{\vEdRatioCellsRa}{11}
\newcommand{\vEdRatioCellsRb}{12}
\newcommand{\vEdRatioCellsRc}{12}
\newcommand{\vEdRatioRa}{0.19}
\newcommand{\vEdRatioRb}{0.21}
\newcommand{\vEdRatioRc}{0.31}
\newcommand{\vEdRatioWinsRa}{11}
\newcommand{\vEdRatioWinsRb}{12}
\newcommand{\vEdRatioWinsRc}{11}
\newcommand{\vEdSeeds}{10}
\newcommand{\vEdSlopeBody}{Sphere, $\dv=1$ & Uniform sampling & 0.54 [0.40, 0.60] & 0.81 [0.76, 0.93] & 0.57 [0.40, 0.72] \\
Sphere, $\dv=1$ & \textbf{Halving} & 0.58 [0.32, 0.86] & \textbf{1.10} [0.98, 1.26] & \textbf{1.03} [0.82, 1.22] \\
Sphere, $\dv=8$ & Uniform sampling & 0.34 [0.31, 0.38] & 0.58 [0.54, 0.63] & 0.27 [0.25, 0.27] \\
Sphere, $\dv=8$ & \textbf{Halving} & 0.41 [0.35, 0.45] & 0.90 [0.86, 0.96] & 0.66 [0.57, 0.75] \\
Clustered, $\dv=1$ & Uniform sampling & 0.42 [0.31, 0.59] & 0.57 [0.45, 0.72] & 0.44 [0.32, 0.73] \\
Clustered, $\dv=1$ & \textbf{Halving} & 0.83$^{\dagger}$ [0.40, 1.47] & 0.77$^{\dagger}$ [0.54, 1.17] & 0.63$^{\dagger}$ [0.45, 1.06] \\
Clustered, $\dv=8$ & Uniform sampling & 0.33 [0.29, 0.43] & 0.41 [0.37, 0.49] & 0.27 [0.23, 0.32] \\
Clustered, $\dv=8$ & \textbf{Halving} & 0.52 [0.41, 0.78] & 0.60 [0.49, 0.73] & 0.49 [0.36, 0.66] \\}
\newcommand{\vEdViolation}{0.006}
\newcommand{\vEfBkvRa}{0.025}
\newcommand{\vEfBkvRb}{0.034}
\newcommand{\vEfBkvRc}{0.051}
\newcommand{\vEfBkvRd}{0.069}
\newcommand{\vEfN}{65536}
\newcommand{\vEfOhRa}{0.011}
\newcommand{\vEfOhRb}{0.018}
\newcommand{\vEfOhRc}{0.026}
\newcommand{\vEfOhRd}{0.038}
\newcommand{\vEfOsbRa}{0.011}
\newcommand{\vEfOsbRb}{0.016}
\newcommand{\vEfOsbRc}{0.024}
\newcommand{\vEfOsbRd}{0.038}
\newcommand{\vEfSize}{1024}
\newcommand{\vEfThinRa}{0.024}
\newcommand{\vEfThinRb}{0.035}
\newcommand{\vEfThinRc}{0.053}
\newcommand{\vEfThinRd}{0.068}
\newcommand{\vEfUnifRa}{0.031}
\newcommand{\vEfUnifRb}{0.040}
\newcommand{\vEfUnifRc}{0.058}
\newcommand{\vEfUnifRd}{0.074}
\newcommand{\vKernelBody}{Qwen2.5-7B-Instruct & $[0,50)$ & 14 & 4 & 256.0 & 14/14 & 0.972 & 0.179 & 0 \\
 & $[50,100)$ & 28 & 7 & 256.0 & 28/28 & 0.962 & 0.039 & 0 \\
 & $[100,200)$ & 4 & 1 & 158.0 & 0/4 & 0.951 & $3.20\!\times\!10^{-6}$ & 1 \\
 & $\ge 200$ & 8 & 2 & 133.5 & 0/8 & 0.967 & $3.31\!\times\!10^{-11}$ & 5 \\
Llama-3.1-8B-Instruct & $[0,50)$ & 43 & 9 & 256.0 & 43/43 & 0.954 & 0.064 & 0 \\
 & $[50,100)$ & 11 & 3 & 256.0 & 11/11 & 0.931 & 0.002 & 0 \\}
\newcommand{\vKnLAbHi}{0.064}
\newcommand{\vKnLAbLo}{0.002}
\newcommand{\vKnLHeadsHi}{0}
\newcommand{\vKnLOperSat}{54/54}
\newcommand{\vKnLResid}{0.036}
\newcommand{\vKnQAbHi}{0.179}
\newcommand{\vKnQAbLo}{0.039}
\newcommand{\vKnQCap}{256}
\newcommand{\vKnQCollapse}{100}
\newcommand{\vKnQDegenHi}{6}
\newcommand{\vKnQDumpsHi}{12}
\newcommand{\vKnQHeads}{9}
\newcommand{\vKnQHeadsHi}{2}
\newcommand{\vKnQOperSat}{42/42}
\newcommand{\vKnQResid}{0.015}
\newcommand{\vMeLBkvFrac}{16\%}
\newcommand{\vMeLBkvRatio}{1.000}
\newcommand{\vMeLDumps}{27}
\newcommand{\vMeLHeads}{9}
\newcommand{\vMeLHhFrac}{93\%}
\newcommand{\vMeLHhRatio}{0.250}
\newcommand{\vMeLMaxFrac}{26\%}
\newcommand{\vMeLOhFrac}{21\%}
\newcommand{\vMeLOhRatio}{1.000}
\newcommand{\vMeLOsbFrac}{26\%}
\newcommand{\vMeLOsbRatio}{1.000}
\newcommand{\vMeLPqFrac}{39\%}
\newcommand{\vMeLPqRatio}{1.000}
\newcommand{\vMeLThinFrac}{7\%}
\newcommand{\vMeLThinRatio}{1.998}
\newcommand{\vMeQBkvFrac}{19\%}
\newcommand{\vMeQBkvRatio}{1.000}
\newcommand{\vMeQDumps}{28}
\newcommand{\vMeQHeads}{6}
\newcommand{\vMeQHhFrac}{93\%}
\newcommand{\vMeQHhRatio}{0.250}
\newcommand{\vMeQMaxFrac}{25\%}
\newcommand{\vMeQOhFrac}{25\%}
\newcommand{\vMeQOhRatio}{1.000}
\newcommand{\vMeQOsbFrac}{24\%}
\newcommand{\vMeQOsbRatio}{1.000}
\newcommand{\vMeQPqFrac}{46\%}
\newcommand{\vMeQPqRatio}{1.000}
\newcommand{\vMeQThinFrac}{21\%}
\newcommand{\vMeQThinRatio}{1.000}

\newcommand{\vRunEaLGpu}{RTX 5090}
\newcommand{\vRunEaLHours}{0.7}
\newcommand{\vRunEaQGpu}{RTX 5090}
\newcommand{\vRunEaQHours}{0.4}
\newcommand{\vRunEbLGpu}{RTX 5090}
\newcommand{\vRunEbLHours}{2.9}
\newcommand{\vRunEbQGpu}{RTX 3090}
\newcommand{\vRunEbQHours}{5.2}
\newcommand{\vRunEcLGpu}{RTX 5090}
\newcommand{\vRunEcLHours}{0.5}
\newcommand{\vRunEcQGpu}{RTX 5090}
\newcommand{\vRunEcQHours}{0.4}
\newcommand{\vRunEdGpu}{RTX 3090/RTX 5090}
\newcommand{\vRunEdHours}{10.8}
\newcommand{\vScLExp}{0.054}
\newcommand{\vScLExpHi}{0.059}
\newcommand{\vScLExpKappa}{0.039}
\newcommand{\vScLExpLo}{0.048}
\newcommand{\vScLExpSe}{0.004}
\newcommand{\vScLExpVrho}{0.013}
\newcommand{\vScLFillFrac}{10\%}
\newcommand{\vScLRhoMinHead}{26.775}
\newcommand{\vScQExp}{0.073}
\newcommand{\vScQExpHi}{0.078}
\newcommand{\vScQExpKappa}{0.057}
\newcommand{\vScQExpLo}{0.062}
\newcommand{\vScQExpSe}{0.006}
\newcommand{\vScQExpVrho}{0.014}
\newcommand{\vScQFillFrac}{8\%}
\newcommand{\vScQLever}{64}
\newcommand{\vScQRhoMinHead}{24.100}
\newcommand{\vScQShortest}{512}
\newcommand{\vScaleBody}{Qwen2.5-7B-Instruct & 512 & 112 & 250 & 18.287 & 2.288 & 41.954 \\
 & 1024 & 112 & 250 & 19.143 & 2.319 & 44.721 \\
 & 2048 & 112 & 250 & 20.047 & 2.339 & 47.942 \\
 & 4096 & 112 & 36 & 20.502 & 2.398 & 49.685 \\
 & 8192 & 112 & 63 & 21.474 & 2.420 & 52.368 \\
 & 16384 & 112 & 62 & 22.517 & 2.445 & 56.454 \\
 & 32768 & 112 & 62 & 23.259 & 2.396 & 55.946 \\
Llama-3.1-8B-Instruct & 512 & 256 & 250 & 21.762 & 1.589 & 34.195 \\
 & 1024 & 256 & 250 & 22.279 & 1.625 & 35.832 \\
 & 2048 & 256 & 250 & 22.789 & 1.650 & 37.511 \\
 & 4096 & 256 & 27 & 22.762 & 1.664 & 37.627 \\
 & 8192 & 256 & 63 & 23.978 & 1.659 & 39.606 \\
 & 16384 & 256 & 60 & 24.973 & 1.697 & 42.271 \\
 & 32768 & 256 & 62 & 25.471 & 1.677 & 42.595 \\}
\newcommand{\vSepKaBound}{5.886}
\newcommand{\vSepKaPq}{0.878}
\newcommand{\vSepKaWitness}{16.010}
\newcommand{\vSepKbBound}{11.772}
\newcommand{\vSepKbPq}{0.813}
\newcommand{\vSepKbWitness}{32.005}
\newcommand{\vSepKcBound}{23.544}
\newcommand{\vSepKcPq}{1.000}
\newcommand{\vSepKcWitness}{64.003}
\newcommand{\vSepSeeds}{8}

\title{Query-Oblivious Coresets for Softmax Attention:\\
Improved Bounds and Efficient Constructions}
\author{Ofek I. Cohen}
\date{September 2026}

\begin{document}
\maketitle
{\let\thefootnote\relax\footnotetext{Tel Aviv University. \texttt{ofeki@mail.tau.ac.il}}}

\begin{abstract}
A query-oblivious coreset for a softmax-attention head is a subset of the
key--value pairs whose attention output is within $\varepsilon$ of the full
one for every query in a ball. With $\rho$ the query radius times the centred
key radius, we build such coresets of size
$O(e^{\rho}(\sqrt{d_v}+\sqrt{d_k\log(1+\rho)})/\varepsilon)$ in randomised
polynomial time, improving the existential
$O(\sqrt d\,e^{\rho+\frac12\log\rho+o(\log\log\rho)}/\varepsilon)$ of Liberty,
Andoni and Kleiner to within $\sqrt{\log(1+\rho)}$ of their
$\Omega(\sqrt d\,e^{\rho}/\varepsilon)$ lower bound: a spherical lift of both
balls into one exponential-kernel instance admits the Bozzai--Rothvoss
chaining bound, and Chevet's inequality splits key from value dimension. A
sampling cap $O(e^{2\rho}/\varepsilon^{2})$ completes the envelope, and in
fixed dimension Tai's diameter-free bound removes the logarithm. A census of
every head of Qwen2.5-7B-Instruct and Llama-3.1-8B-Instruct finds
$\rho\ge\vScQRhoMinHead$, so every whole-ball bound with a factor
$e^{\rho}/\varepsilon$ prescribes a coreset larger than the cache: the
algorithmic contribution is asymptotic on these models.
\end{abstract}

% =====================================================================
\section{Introduction}
\label{sec:intro}

Fix an attention head with keys $k_1,\dots,k_n\in\R^{\dk}$ and values
$v_1,\dots,v_n\in\R^{\dv}$, and write
\begin{equation}
\label{eq:attn}
  A(q)=\sum_{i=1}^{n}e^{q^{\!\top}k_i}v_i,\qquad
  B(q)=\sum_{i=1}^{n}e^{q^{\!\top}k_i},\qquad
  \Attn(q)=\frac{A(q)}{B(q)} .
\end{equation}
For $S\subseteq[n]$ write $\Attn_S$ for the same expression over $S$. An
unweighted \emph{query-oblivious $\eps$-coreset} is a subset $S$ with
\begin{equation}
\label{eq:coreset}
  \sup_{\norm{q}\le\vrho}\;\norm{\Attn_S(q)-\Attn(q)}_2\;\le\;\eps .
\end{equation}
Compressing the key--value cache of a transformer to such a subset keeps
attention correct for every query the model may later produce, without
reweighting and without knowing the queries. The governing parameter is
\begin{equation}
\label{eq:rho}
  \rho:=\vrho\,\kappa,\qquad
  \kappa\ge\max_i\norm{k_i-\bar k},\quad \bar k=\tfrac1n\sum_ik_i ,
\end{equation}
the product of the query radius and the \emph{centred} key radius. Translating
all keys by $-\bar k$ multiplies numerator and denominator in~\eqref{eq:attn}
by the same factor $e^{-q^{\!\top}\bar k}$, for the full set and for every
subset alike, so~\eqref{eq:coreset} is invariant under centring; the
centred radius can be twice the uncentred one, and a theorem phrased with
$\max_i\norm{k_i}$ silently permits $e^{2\rho}$. Throughout, after
centring and rescaling once, we may and do assume the normalised model
\begin{equation}
\label{eq:norm}
  \tfrac1n\sum_ik_i=0,\qquad \norm{k_i}\le1,\qquad \norm{q}\le\rho,\qquad \norm{v_i}\le1 .
\end{equation}
Centring and the joint rescaling $k\mapsto(k-\bar k)/\kappa$, $q\mapsto\kappa q$
leave~\eqref{eq:coreset} unchanged, as in~\cite[App.~B.5]{KSHZK25},
\cite[\S6.1.2]{CFIKKP26} and~\cite{SM26}; dividing every $v_i$ by $\max_i\norm{v_i}$
scales its left side by the same factor, so $\eps$ is measured in units of $\max_i\norm{v_i}$.
Two facts about~\eqref{eq:norm} are used repeatedly: every weight
$e^{q^{\!\top}k_i}$ lies in $[e^{-\rho},e^{\rho}]$, and by Jensen's inequality
$B(q)\ge n\exp(q^{\!\top}\bar k)=n$ for every $q$.

Liberty, Andoni and Kleiner~\cite{LAK26} proved that an unweighted coreset of
size $O(\sqrt d\,e^{\zeta}/\eps)$ exists, $\zeta=\rho+\tfrac12\log\rho+o(\log\log\rho)$,
$d=\dk=\dv$, against a lower bound $\Omega(\sqrt d\,e^{\rho}/\eps)$ under
conditions on $d$ and $\eps$ relative to $\rho$; the argument is existential.
This paper narrows the residual $\rho^{1/2+o(1)}$ (a $\sqrt{\log n}$ in the
worst case) to $\sqrt{\log(1+\rho)}$ by two reductions to known
kernel-discrepancy theorems and gives the first polynomial-time construction
with the uniform guarantee~\eqref{eq:coreset} at this size. Whether $e^{\rho}$ is moderate on
transformer caches, and whether the construction of Theorem~\ref{thm:main-alg}
beats sampling there, is measured in Section~\ref{sec:exp}: it is not, and
it does not, and the same section locates the
regime, $\rho\le4$ on synthetic instances, in which the construction needs
fewer pairs than sampling.

\subsection{Results}

\begin{theorem}[Upper envelope]
\label{thm:main}
Let $\eps\in(0,1]$, $\dk,\dv\ge1$, and let the instance satisfy~\eqref{eq:norm}.
There is an unweighted subset $S$ satisfying~\eqref{eq:coreset} with
\begin{equation}
\label{eq:envelope}
  |S|
  \;=\;O\!\left(\min\left\{
    n,\;
    \frac{e^{\rho}\bigl(\sqrt{\dv}+\sqrt{\dk\log(1+\rho)}\bigr)}{\eps},\;
    \frac{e^{2\rho}}{\eps^{2}},\;
    \frac{F(\dk,\dv)\,e^{\rho}}{\eps}
  \right\}\right),
\end{equation}
where $F(\dk,\dv)=7\sqrt{\dv}\,e^{\dv/4}(f_{\mathrm T}(\dk+1+2\dv)+1)$ and
$f_{\mathrm T}(D)$ is Tai's diameter-free Gaussian-kernel discrepancy constant
in $\R^{D}$, finite for every $D$.
\end{theorem}

The first branch is proved in Section~\ref{sec:upper} and
Appendix~\ref{app:coreset}, with constants explicit up to Dudley's constant
and the Gaussian $\psi_2$ constant; the second, with explicit constant $32$,
is Theorem~\ref{thm:sampling}; the third is Tai's theorem transported in
Section~\ref{sec:nolog}, whose constant $F$ is exponential in $\dv$ and, as
far as is known, in $\dk$ (Open Problem~\ref{prob:g}).

\begin{theorem}[Polynomial time]
\label{thm:main-alg}
Under~\eqref{eq:norm} and $\eps\in(0,1]$, there are randomised algorithms
producing a set $S$ satisfying~\eqref{eq:coreset}
\begin{enumerate}[label=(\alph*),nosep]
\item with probability at least $\tfrac12$, of size
  $O(e^{\rho}(\sqrt{\dv}+\sqrt{\dk\log(1+\rho)})/\eps)$ always, in expected
  time $O(n^{\omega+1}+n^{2}(\dk+\dv))$, $\omega<2.372$ the
  matrix-multiplication exponent;
\item for $n\ge2$ and $\delta\in(0,\tfrac12]$, with probability at least $1-\delta$, of size
  $O(e^{\rho}(\sqrt{\dv}+\sqrt{\dk\log(1+\rho)}+\sqrt{\log(\log n/\delta)})/\eps)$
  always, in the same time;
\item with probability at least $\tfrac14$, of the size in~(a), in expected time
  $O(n\dk)+\mathrm{poly}(e^{\rho}/\eps,\dk,\dv)$, by sampling
  $O(e^{2\rho}/\eps^{2})$ pairs and balancing the sample.
\end{enumerate}
\end{theorem}

The Gram--Schmidt walk~\cite{BDGL18} replaces Banaszczyk's theorem inside
the halving, as in~\cite{PT18,BR23}; the rejection scheme on checkable
constraints is Tai's~\cite[Alg.~1]{Tai22}. BalanceKV~\cite{KSHZK25}
and~\cite{CFIKKP26} run discrepancy walks on the same tensor features for
attention with a guarantee \emph{per query}, and \cite{HKV26} in mean squared
error over a query distribution; \cite{LAK26} is uniform over the ball but
existential. Theorem~\ref{thm:main-alg} is the first polynomial-time
construction whose guarantee is one signing uniform over the query ball at
the existential size up to $\sqrt{\log(1+\rho)}$: the per-query walks
of~\cite{KSHZK25} become uniform through a union bound over a net of the
ball, at a polylogarithmic cost and with $e^{2\rho}$ in place of $e^{\rho}$,
which is how~\cite{LAK26} restate them. Proposition~\ref{prop:sep} shows that
the uniform quantity is a different one, not a different proof.

\begin{theorem}[The logarithm is not minimax in fixed dimension]
\label{thm:main-nolog}
For every $\dk$ there is a finite $g(\dk)\le f_{\mathrm T}(\dk+1)+1$ such that
for every $\rho>0$, every $n$ and all keys in the unit ball of $\R^{\dk}$ there
is a signing $\sigma\in\{\pm1\}^{n}$ with $|\sum_i\sigma_i|\le1$ and
\[
  \sup_{\norm{q}\le\rho}\Bigl|\sum_i\sigma_ie^{q^{\!\top}k_i}\Bigr|\le g(\dk)\,e^{\rho}.
\]
Consequently, for fixed $\dk$ there is no $c>0$ such that for all large $\rho$
some configuration forces
$\min_\sigma\sup_{\norm{q}\le\rho}|\sum_i\sigma_ie^{q^{\!\top}k_i}|\ge c\,e^{\rho}\sqrt{\dk\log(1+\rho)}$;
and the third branch of~\eqref{eq:envelope} holds.
\end{theorem}

Theorem~\ref{thm:main-nolog} is a fixed-dimension statement; whether the
logarithm can be removed when $\dk$ grows with $\rho$, the regime of
Theorem~\ref{thm:main-lb}(a), is open.

\begin{theorem}[Lower bounds]
\label{thm:main-lb}
Under~\eqref{eq:norm}:
\begin{enumerate}[label=(\alph*),nosep]
\item (\cite[Thm.~7]{LAK26}, centred form.) There are constants $c,C,\eps_0>0$
  such that for $\dk\ge2$, $\dv\ge1$, $\rho\ge2$ and $0<\eps\le\eps_0/\sqrt{\dv}$ there
  is an instance for which every $\eps$-coreset has
  $|S|\ge c\,\dv\min\{e^{\rho}/(\eps\sqrt{\dv}),\,e^{2\rho},\,\exp(\dk/(C\rho^{2}))\}$;
  in particular $|S|=\Omega(e^{\rho}\sqrt{\dv}/\eps)$ when
  $\eps\sqrt{\dv}\gtrsim e^{-\rho}$ and $\dk\gtrsim\rho^{2}(\rho+\log\frac{1}{\eps\sqrt{\dv}})$.
\item (Mean floor.) For every $\rho\ge0$ there are instances for which every
  $\eps$-coreset has $|S|=\Omega(\min\{\dv,\eps^{-2}\})$.
\item (\cite{CFIKKP26}, transferred.) For scalar values, in the two parameter
  regimes~(a) and~(b) of Fact~\ref{fact:cfikkp}, there are instances for which every
  $\eps$-coreset has $|S|\ge\tilde\Omega(\min\{n,(1/\eps)^{1-1/a-o(1)}\})$,
  respectively $|S|\ge\tilde\Omega(\min\{n,e^{\rho(1-o(1))}/\eps\})$.
\end{enumerate}
\end{theorem}

Part~(a) is the lower bound of~\cite{LAK26}, whose one-bit-per-(key, value
coordinate) gadget appears earlier in the streaming bit model
(\cite[Thm.~6]{HO25}, \cite[Thm.~3.2]{KSHZK25}); Appendix~\ref{app:lb}
gives it a centred instance compatible with~\eqref{eq:rho} and a complete
proof.
Part~(c) is a transfer: the streaming bounds of~\cite{CFIKKP26} are one-way
communication bounds whose proofs use nothing about the data structure beyond
its bit length, and a coreset is such a data structure.

\begin{proposition}[What query-obliviousness costs]
\label{prop:sep}
Under~\eqref{eq:norm} with scalar values $|v_i|\le1$:
\begin{enumerate}[label=(\roman*),nosep]
\item for every fixed $q$ there is a signing, computable in $O(n)$ time from the weights $e^{q^{\!\top}k_i}v_i$ (so in $O(n\dk)$ time in all), with
  $|\sum_i\sigma_ie^{q^{\!\top}k_i}v_i|\le e^{\rho}$; and for every
  $\eps\in(0,1]$ there is a subset $S$ with $|S|\le24e^{\rho}/\eps+3$ and
  $|\Attn_S(q)-\Attn(q)|\le\eps$;
\item for every $r\ge2$, every $\dk\ge r$ and every $0<\rho\le\sqrt r$ there is a
  centred instance of $r$ keys with values $v_i\equiv1$ such that every signing
  satisfies $\sup_{\norm{q}\le\rho}|\sum_i\sigma_ie^{q^{\!\top}k_i}|\ge\rho\sqrt r/e$.
\end{enumerate}
In particular, at $\rho=1$, one signing for the whole ball costs
$\Omega(\sqrt{\dk})$ where each single query costs $O(1)$.
\end{proposition}

\paragraph{Empirical results.}
Section~\ref{sec:exp} reports four experiments. A census of every
key--value head of Qwen2.5-7B-Instruct and Llama-3.1-8B-Instruct finds median
$\rho$ of \vEaQrhoMax\ and \vEaLrhoMax\ at 4k--32k tokens and no head below
\vScQRhoMinHead\ at 512 tokens, so every whole-ball bound with a factor
$e^{\rho}/\eps$ exceeds the cache on these models, and at matched size on
real heads the construction's worst-case error is \vEbQRatioMed\ and
\vEbLRatioMed\ times uniform sampling's at the median cell. Where $\rho\le4$
is reached, on synthetic instances, the construction reaches $\eps=0.1$ with
\vEdGainLo\ to \vEdGainHi\ times fewer pairs than uniform sampling. In a
closed decoding loop a query-aware oracle degrades relative to the open loop
and the query-oblivious methods do not.

\subsection{Related work}

Kernel-density coresets by discrepancy go back to Phillips~\cite{Phi13} and
Phillips--Tai~\cite{PT18}, who also posed the $O(\sqrt d/\eps)$ question;
Karnin--Liberty~\cite{KL19} lifted dot-product kernels by tensor powers;
Bozzai--Rothvoss~\cite{BR23} proved the chaining bound we use, and a
fixed-dimension bound with a $\sqrt{\log\log n}$ that Tai~\cite{Tai22} removes
for the Gaussian kernel; Charikar--Kapralov--Waingarten~\cite{CKW24} gave
quasi-Monte-Carlo structures for smooth kernels. For attention,
BalanceKV~\cite{KSHZK25} and Chen et al.~\cite{CFIKKP26} give discrepancy-based
streaming algorithms with per-query guarantees and streaming lower bounds,
Haris--Onak~\cite{HO25} bit-model lower bounds, Haverbeck et al.~\cite{HKV26}
an averaged risk analysis, Zandieh et al.~\cite{ZHMK24} a sampling bound,
Schr\"oder--Mackey~\cite{SM26} a weighted near-linear-time construction, and
Liberty--Andoni--Kleiner~\cite{LAK26} the existential bounds we improve.
Thinning-based attention approximation~\cite{CGSDM25,GCDM26} selects unweighted
or weighted subsets with guarantees for the observed queries rather than
uniformly over a ball. Wang~\cite{Wang26} proves an oblivious-versus-adaptive
separation for KV caches in a combinatorial pointer-chasing model; the
separation here is one of softmax discrepancy and the two are not comparable. Sui and
Zhang~\cite{SZ26} study the approximation rank of attention, a different
object. Bohbot et al.~\cite{BLPV25} prove a $C(R)/\sqrt n$ token-sample rate
for i.i.d.\ tokens, superficially like our sampling cap but for a different
problem. The heavy-hitter oracle of Section~\ref{sec:exp} is a query-aware
eviction in the style of SnapKV~\cite{Li24}, fitted on the decoding queries;
kernel thinning~\cite{DM21} is the primitive behind the thinning
of~\cite{CGSDM25}.

% =====================================================================
\section{Why one signing for the whole ball}
\label{sec:why}

A coreset built for the cache of a decoding transformer is used by queries that
are computed \emph{from} the compressed cache: the query at step $t+1$ is a
function of the attention outputs at steps $\le t$, which were produced from
$S$. A guarantee that holds with probability $1-\delta$ for each query fixed
in advance is a statement about queries independent of the randomness that
produced $S$; the closed loop makes them dependent, and the failure events at
successive steps are neither independent nor controlled. A guarantee
uniform over the ball, as in~\eqref{eq:coreset}, is insensitive to this. This
is a modelling argument, not a theorem; Proposition~\ref{prop:sep} is the
theorem, and it says the uniform guarantee is a different quantity, priced by
the dimension. Section~\ref{sec:e3} tests the argument in a decoding loop,
a query-aware oracle against query-oblivious methods.

\begin{proof}[Proof of Proposition~\ref{prop:sep}]
(i) Write $w_i:=e^{q^{\!\top}k_i}v_i$ and $M:=\max_i|w_i|\le e^{\rho}$. Process
the indices in any order, maintaining $P_j:=\sum_{i\le j}\sigma_iw_i$, and choose
$\sigma_j$ with $\sigma_jw_jP_{j-1}\le0$. Inductively $|P_{j-1}|\le M$: if
$P_{j-1}=0$ then $|P_j|=|w_j|\le M$; if $P_{j-1}\in(0,M]$ then
$\sigma_jw_j\in[-M,0]$ and $P_j\in(-M,M]$; symmetrically for
$P_{j-1}\in[-M,0)$. So $|P_n|\le M$. The coreset statement is
Corollary~\ref{cor:perquery} below.

(ii) Let $c_j:=e_j-\tfrac1r\sum_{i\le r}e_i\in\R^{r}\subseteq\R^{\dk}$,
$j=1,\dots,r$. Then $\sum_jc_j=0$ and $\norm{c_j}^{2}=1-\tfrac1r\le1$, so the
instance satisfies~\eqref{eq:norm} with the given $\rho$. Fix a signing
$\sigma$, put $T:=\sum_j\sigma_j$, and for $z\in\{\pm1\}$ consider
$q_z:=\rho z\,r^{-1/2}\sum_j\sigma_je_j$, of norm $\rho$. Then
$q_z^{\!\top}c_j=\rho zr^{-1/2}(\sigma_j-T/r)$, so
\[
  \sum_j\sigma_je^{q_z^{\!\top}c_j}
  =e^{-\rho zT/(r\sqrt r)}\Bigl[T\cosh\bigl(\rho/\sqrt r\bigr)+z\,r\sinh\bigl(\rho/\sqrt r\bigr)\Bigr].
\]
One of the two choices of $z$ makes the bracket at least $r\sinh(\rho/\sqrt r)$
in absolute value, and $|T|\le r$ gives $e^{-\rho zT/(r\sqrt r)}\ge e^{-\rho/\sqrt r}\ge e^{-1}$
for $\rho\le\sqrt r$. Since $\sinh x\ge x$, the absolute value is at least
$\rho\sqrt r/e$.
\end{proof}

\begin{lemma}[Three constraints at one query; Beck--Fiala~\cite{BF81}]
\label{lem:bf}
Let $a_1,\dots,a_m\in\R^{3}$ with $|a_i^{(c)}|\le M_c$ for $c=1,2,3$. There is
$\sigma\in\{\pm1\}^{m}$, computable in polynomial time, with
$|\sum_i\sigma_ia_i^{(c)}|\le3M_c$ for each $c$.
\end{lemma}

\begin{proof}
Maintain $x\in[-1,1]^{m}$ with $\sum_ix_ia_i=0$ and the set $F$ of indices with
$|x_i|<1$, starting from $x=0$, $F=[m]$. While $|F|\ge4$ the map
$y\mapsto\sum_{i\in F}y_ia_i$ from $\R^{F}$ to $\R^{3}$ has a nonzero kernel
element $y$; move $x\leftarrow x+ty$ with $|t|$ minimal such that a coordinate
in $F$ reaches $\pm1$. The invariant is preserved and $|F|$ decreases, so the
loop ends with $|F|\le3$ after at most $m$ steps. Put $\sigma_i=x_i$ for
$i\notin F$ and $\sigma_i=\operatorname{sign}(x_i)$ for $i\in F$
($\operatorname{sign}0:=1$); then $|\sigma_i-x_i|\le1$ on $F$ and
$\sum_i\sigma_ia_i=\sum_{i\in F}(\sigma_i-x_i)a_i$ has $c$-th coordinate at
most $3M_c$ in absolute value.
\end{proof}

\begin{corollary}[A coreset for one query]
\label{cor:perquery}
Under~\eqref{eq:norm} with scalar values, for every fixed $q$ with
$\norm q\le\rho$ and every $\eps\in(0,1]$ there is $S\subseteq[n]$, computable
deterministically in polynomial time, with $|\Attn_S(q)-\Attn(q)|\le\eps$
and $|S|\le24e^{\rho}/\eps+3$.
\end{corollary}

\begin{proof}
Write $w_i=e^{q^{\!\top}k_i}\in(0,e^{\rho}]$ and $D:=3e^{\rho}$. Given $S_j$ of
size $n_j$ with $A_j=\sum_{S_j}w_iv_i$, $B_j=\sum_{S_j}w_i$, apply
Lemma~\ref{lem:bf} to $a_i=(w_iv_i,w_i,1)$ with $M_1=M_2=e^{\rho}$, $M_3=1$,
replace $\sigma$ by $-\sigma$ if $\sum\sigma_i<0$, and let
$S_{j+1}=\{i\in S_j:\sigma_i=+1\}$. With $\delta_A=\sum\sigma_iw_iv_i$,
$\delta_B=\sum\sigma_iw_i$ we have $|\delta_A|,|\delta_B|\le D$,
$A_{j+1}=\tfrac12(A_j+\delta_A)$, $B_{j+1}=\tfrac12(B_j+\delta_B)$, and
$n_j/2\le n_{j+1}\le n_j/2+\tfrac32$, hence $n/2^{j}\le n_j\le n/2^{j}+3$. By
Jensen $B_0\ge n$, and by induction $B_j\ge n/2^{j}-D$. If $n/2^{j}\ge4D$ then
$B_j-D\ge n/2^{j+1}$ and, since $|A_j|\le B_j$,
\[
  \Bigl|\frac{A_{j+1}}{B_{j+1}}-\frac{A_j}{B_j}\Bigr|
  =\Bigl|\frac{\delta_AB_j-\delta_BA_j}{B_j(B_j+\delta_B)}\Bigr|
  \le\frac{2D}{B_j-D}\le\frac{4D\,2^{j}}{n}.
\]
Take $\ell\ge0$ maximal with $n/2^{\ell}\ge4D/\eps$ (if none, $n<4D/\eps$ and
$S=[n]$ serves). Then $n/2^{j}\ge4D$ for $j<\ell$, the errors sum to at most
$4D2^{\ell}/n\le\eps$, and $|S_\ell|<8D/\eps+3$.
\end{proof}

\begin{remark}
Proposition~\ref{prop:sep} separates signing discrepancies, $O(1)$ per query
against $\Omega(\sqrt{\dk})$ over the ball at $\rho=1$. At the level of
coreset size the per-query bound
is $O(e^{\rho}/\eps)$ for scalar values, the query-oblivious upper bound is
$O(e^{\rho}\min\{\sqrt{\dk\log(1+\rho)},F(\dk,1)\}/\eps)$, and whether a
query-oblivious coreset for scalar values needs any factor of $\dk$ is the
open key-direction question; the local INDEX architecture cannot settle it,
since one bit per token decoded at one query yields at most
$e^{\rho}\sqrt{\dv}/\eps$ (see~\cite{LAK26} and Appendix~\ref{app:lb}).
From the kernel-density side, uniformity in the query is the standard
definition of kernel discrepancy~\cite{PT18,Tai22,BR23}; the distinction is
meaningful against the attention literature~\cite{KSHZK25,CFIKKP26,HKV26}.
\end{remark}

% =====================================================================
\section{The upper bound}
\label{sec:upper}

\subsection{Two reductions}

Bozzai and Rothvoss~\cite{BR23} consider
$K_e^{(\alpha)}(x,y)=\exp(-\alpha(1-\inner{x}{y}))$ on $S^{d}\times S^{d}$ and
prove $\disc_{K_e^{(\alpha)}}(n)=O(\sqrt{d\log(2\max\{\alpha,1\})})$ with a
randomised polynomial-time colouring (their Theorem~1.11; earlier bounds for
the same quantity are~\cite[Lem.~17]{KL19} and~\cite[Cor.~4.3]{PT18}). Their
theorem is for points on a sphere; the keys here lie in a ball and the
queries in another.

\begin{lemma}[Two-slack lift]
\label{lem:slack}
Under~\eqref{eq:norm} with $\rho>0$ put
$\widetilde k_i:=(k_i,\sqrt{1-\norm{k_i}^{2}},0)$ and
$\widetilde q:=(q/\rho,0,\sqrt{1-\norm{q}^{2}/\rho^{2}})$, both in
$S^{\dk+1}\subset\R^{\dk+2}$. Then $\inner{\widetilde k_i}{\widetilde q}=q^{\!\top}k_i/\rho$
and $e^{q^{\!\top}k_i}=e^{\rho}K_e^{(\rho)}(\widetilde k_i,\widetilde q)$.
\end{lemma}

\begin{proof}
The two slack coordinates are orthogonal, contribute nothing to the inner
product, and normalise both vectors; the identity is exact on the whole key
ball and query ball.
\end{proof}

\begin{corollary}[Scalar discrepancy]
\label{cor:scalar}
Under~\eqref{eq:norm} with $\rho>0$,
\[
  \min_\sigma\,\sup_{\norm q\le\rho}\Bigl|\sum_i\sigma_ie^{q^{\!\top}k_i}\Bigr|
  =O\bigl(e^{\rho}\sqrt{\dk\log(2+\rho)}\bigr),
\]
constructively, by~\cite[Thm.~1.11]{BR23} applied to
$\{\widetilde k_i\}\subset S^{\dk+1}$.
\end{corollary}

Writing $k_i=\kappa\hat k_i$ silently assumes $\norm{k_i}=\kappa$; the lift is
the correct reduction. The guard $\log(1+\rho)$ for $\rho<1$, where the query
ball collapses and the width degenerates, is a computation inside the same
chaining method (Lemma~\ref{lem:dudQ}); \cite{BR23} cap the bandwidth at
$\max\{\alpha,1\}$, so their bound does not vanish as $\alpha\to0$.

For vector values one runs the colouring on the tensor features
$\phi(k_i)\otimes v_i$, the lift of~\cite[Lem.~16]{KL19}, and needs the signed
sum controlled simultaneously for all $q$ and all unit $u\in\R^{\dv}$. Let $\psi(z)=\bigoplus_{m\ge0}z^{\otimes m}/\sqrt{m!}$,
so $\inner{\psi(y)}{\psi(z)}=e^{y^{\!\top}z}$, and, for $\rho>0$, put
$\phi_0(k)=e^{-\rho/2}\psi(\sqrt\rho\,k)$, $\Phi_0(q)=e^{\rho/2}\psi(q/\sqrt\rho)$,
so that $\inner{\Phi_0(q)}{\phi_0(k)}=e^{q^{\!\top}k}$, $\norm{\phi_0(k)}\le1$ and
$\norm{\Phi_0(q)}\le e^{\rho}$ on the respective balls. The test set is
$T=Q\otimes S^{\dv-1}$ with $Q=\{\Phi_0(q):\norm q\le\rho\}$ of radius $e^{\rho}$.

\begin{lemma}[Chevet's inequality, two-sided]
\label{lem:chevet}
For bounded $A\subset H_1$, $B\subset H_2$ in Hilbert spaces, with $w$ the
Gaussian width,
$\max\{\rad(A)w(B),\rad(B)w(A)\}\le w(A\otimes B)\lesssim\rad(A)w(B)+\rad(B)w(A)$.
\end{lemma}

The upper bound is~\cite{Che78}, see~\cite[\S8.7]{Ver18}, and for the sphere
case~\cite[Thm.~1]{TS14}; the lower bound freezes one factor. With $w(S^{\dv-1})=\Theta(\sqrt{\dv})$ and
$w(Q)=O(e^{\rho}\sqrt{\dk\log(1+\rho)})$ (Lemma~\ref{lem:dudQ} and Dudley's
inequality) this gives $w(T)=O(e^{\rho}(\sqrt{\dv}+\sqrt{\dk\log(1+\rho)}))$;
Appendix~\ref{app:coreset} derives the same bound with explicit constants
directly from the entropy integral, which is what the proofs consume.

\subsection{The reduction to a coreset}
\label{sec:reduction}

\begin{definition}[Balanced signing]
\label{def:disc}
$\sigma\in\{\pm1\}^{P}$ is $(D_A,D_B,D_N)$-balanced on $P\subseteq[n]$ if
\begin{equation}
\label{eq:lvl}
  \sup_{\norm q\le\rho}\Bigl\lVert\sum_{i\in P}\sigma_ie^{q^{\!\top}k_i}v_i\Bigr\rVert\le D_A,\quad
  \sup_{\norm q\le\rho}\Bigl|\sum_{i\in P}\sigma_ie^{q^{\!\top}k_i}\Bigr|\le D_B,\quad
  \Bigl|\sum_{i\in P}\sigma_i\Bigr|\le D_N .
\end{equation}
\end{definition}

\begin{lemma}[One halving step]
\label{lem:step}
Let $\sigma$ be $(D_A,D_B,D_N)$-balanced on $P$, $|P|=m$, and let $P'$ be the
larger of the two sign classes. Then (i) $m/2\le|P'|\le m/2+D_N/2$;
(ii) $B_{P'}(q)\ge\tfrac12(B_P(q)-D_B)$ for all $\norm q\le\rho$;
(iii) $\norm{\Attn_{P'}(q)-\Attn_P(q)}\le(D_A+D_B)/(2B_{P'}(q))$ for all $\norm q\le\rho$.
\end{lemma}

\begin{proof}
Replacing $\sigma$ by $-\sigma$ preserves~\eqref{eq:lvl}, so assume $P'$ is the
$+1$ class; then $A_{P'}=\tfrac12(A_P+\delta_A)$, $B_{P'}=\tfrac12(B_P+\delta_B)$
with $\norm{\delta_A}\le D_A$, $|\delta_B|\le D_B$. (i) and (ii) are immediate.
For (iii),
$\Attn_{P'}-\Attn_P=(A_{P'}B_P-A_PB_{P'})/(B_{P'}B_P)=(\delta_A-\Attn_P\,\delta_B)/(2B_{P'})$,
and $\norm{\Attn_P}\le1$ because it is a convex combination of the $v_i$.
\end{proof}

The recursion halves $\ell$ times. Sizes obey $n/2^{j}\le n_j\le n/2^{j}+D_N$.
Denominators obey $B_{P_0}\ge n$ by Jensen and then
$B_{P_j}\ge n/2^{j}-D_B$ by (ii), so $B_{P_j}\ge n/2^{j+1}$ as long as
$n/2^{j}\ge2D_B$; the per-level error is then at most $2^{j+1}(D_A+D_B)/n$ by
(iii), the errors sum to $\le\eps$ when $2^{\ell+1}(D_A+D_B)\le\eps n$, and
the survivors number at most $4(D_A+D_B)/\eps+D_N$
(Appendix~\ref{app:coreset}). This is the halving of~\cite[Thm.~4]{LAK26},
a technique going back to~\cite{CM96,PT18}, with the ratio split into
numerator and denominator discrepancies as in~\cite[\S2.2]{KSHZK25}
and~\cite[Lem.~B.1]{HKV26}; propagating the denominator by (ii) instead of
re-centring at every level is
what removes their side conditions $n\ge2c\rho\sqrt d\log n$ and the
accumulating key-mean bound, and no key-mean constraint is needed for
existence.

\begin{lemma}[One level]
\label{lem:onelevel}
Under~\eqref{eq:norm} every $P\subseteq[n]$ admits a
$(D_A,D_B,D_N)$-balanced signing with
$D_A\le C_1'e^{\rho}(\sqrt{\dv}+\sqrt{\dk\log(1+\rho)})$,
$D_B\le C_1'e^{\rho}(1+\sqrt{\dk\log(1+\rho)})$, $D_N=10\sqrt5$, for an
explicit absolute constant $C_1'$.
\end{lemma}

The proof (Appendix~\ref{app:coreset}) applies Banaszczyk's theorem~\cite{Ban98}
to the data vectors
$x_i=\tfrac12(\phi_0(k_i)\otimes v_i,\;\phi_0(k_i),\;1)$, of norm at most one,
and the body $\{y:\sup_{T_X}\inner{y}{t}\le\theta_X,\ X\in\{A,B,N\}\}$ whose
three blocks are independent under the Gaussian measure and each of measure
at least $0.8$ by Dudley's inequality and Borell--TIS. The first branch of
Theorem~\ref{thm:main} follows.

\subsection{The sampling branch}

The bound is the standard sampling bound in a reproducing-kernel Hilbert
space~\cite{LPMST15,PT18}, with a greedy counterpart in~\cite[Thm.~24]{KL19},
an averaged form in~\cite[Thm.~4.5]{HKV26} and a weighted streaming form in
SubGen~\cite{ZHMK24}; it is run here in a two-block space so that one draw
controls numerator and denominator together.

\begin{theorem}[Sampling cap]
\label{thm:sampling}
Under~\eqref{eq:norm} and $\eps\in(0,1]$ there is an unweighted subset
satisfying~\eqref{eq:coreset} with $|S|\le\min\{n,\lceil32e^{2\rho}/\eps^{2}\rceil\}$.
A uniformly random subset of size $\min\{n,\lceil128e^{2\rho}/\eps^{2}\rceil\}$
satisfies~\eqref{eq:coreset} with probability at least $\tfrac34$.
\end{theorem}

\begin{proof}
In $\mathcal H\oplus(\mathcal H\otimes\R^{\dv})$ put
$z_i:=2^{-1/2}(\phi_0(k_i),\phi_0(k_i)\otimes v_i)$, $\norm{z_i}\le1$. For a
uniformly random subset $S$ of size $s$, sampling without replacement,
$\E\norm{\tfrac1s\sum_{S}z_i-\tfrac1n\sum_iz_i}^{2}\le1/s$; so some $S$ has
norm deviation at most $s^{-1/2}$, and by Markov's inequality a random $S$ has
deviation at most $2s^{-1/2}$ with probability $\ge\tfrac34$. Reading the two
blocks against $\sqrt2\,\Phi_0(q)$ and $\sqrt2\,\Phi_0(q)\otimes u$, $\norm u=1$,
gives simultaneously for all $q$
$\norm{\bar A_S(q)-\bar A(q)}\le\delta$ and $|\bar B_S(q)-\bar B(q)|\le\delta$
with $\delta=\sqrt2e^{\rho}\cdot(\text{deviation})$, bars denoting averages.
By Jensen $\bar B\ge1$, and $\norm{\bar A}\le\bar B$. If $\delta\le\tfrac12$
then $\bar B_S\ge\bar B/2$ and
$\norm{\Attn_S-\Attn}\le\norm{\bar A_S-\bar A}/\bar B_S+\norm{\bar A}|\bar B_S-\bar B|/(\bar B_S\bar B)\le4\delta$.
With $s\ge32e^{2\rho}/\eps^{2}$ (deterministic deviation $s^{-1/2}$), or
$s\ge128e^{2\rho}/\eps^{2}$ (random deviation $2s^{-1/2}$), $4\delta\le\eps\le1$.
\end{proof}

The mean at $q=0$ shows the $\eps^{-2}$ is a phase transition
(Theorem~\ref{thm:main-lb}(b)). Theorem~\ref{thm:main} is the minimum of
Lemma~\ref{lem:onelevel} with the halving, Theorem~\ref{thm:sampling}, and
Theorem~\ref{thm:nolog-coreset} below.

% =====================================================================
\section{Polynomial time}
\label{sec:alg}

Banaszczyk's theorem is existential. The Gram--Schmidt walk~\cite{BDGL18}
outputs, for vectors of norm at most one, a signing whose signed sum is
$\sigma_{\mathrm{GS}}$-subgaussian with $\sigma_{\mathrm{GS}}=\sqrt{40}$, in
time $O(m(m+N)^{\omega})$. Two things change. The uniform constraints are no
longer certified by a body of measure one half; they hold in expectation,
through Dudley's inequality for the subgaussian process
$t\mapsto\inner{\sum\sigma_ix_i}{t}$ on the test sets. And the two
finite-dimensional constraints, cardinality and key sum, are checkable in
$O(m\dk)$ time, so rejection sampling makes them hold surely; the key sum is
the checkable surrogate for the uncheckable denominator bound, which is why
the algorithm carries a fourth block $k_i$ in its data vectors.

\begin{definition}[One level of the algorithm]
\label{def:level}
Given $P$, $|P|=m$, form the Gram matrix
$G_{ij}=\tfrac14[e^{\rho(k_i^{\!\top}k_j-1)}(1+\inner{v_i}{v_j})+k_i^{\!\top}k_j+1]$
of the vectors $x_i=\tfrac12(\phi_0(k_i)\otimes v_i,\phi_0(k_i),k_i,1)$ in
$O(m^{2}(\dk+\dv))$ time, factor $G=RR^{\!\top}$ by pivoted Cholesky in
$O(m^{3})$, and run the Gram--Schmidt walk on the rows $r_i$ of $R$, which
realise the $x_i$ up to a linear isometry of their span. Accept $\sigma$ if
$|\sum_{i\in P}\sigma_i|\le4\sigma_{\mathrm{GS}}$ and
$\norm{\sum_{i\in P}\sigma_ik_i}\le4\sigma_{\mathrm{GS}}\sqrt{\dk}$; otherwise
repeat with fresh randomness. Output the larger sign class.
\end{definition}

\begin{lemma}[Acceptance and conditional bounds]
\label{lem:accept}
Each trial is accepted with probability at least $\tfrac12$. Conditional on
acceptance, with $Z_X:=\sup_{t\in T_X}\inner{\sum_i\sigma_ix_i}{t}$ for
$X\in\{A,B\}$, $\E[Z_X\mid\mathrm{acc}]\le2C_{\mathrm D}K_{\ast}\Dud(T_X)$ and
$\Prob(Z_X>\tfrac12C_{\mathrm D}K_{\ast}[\Dud(T_X)+4ue^{\rho}]\mid\mathrm{acc})\le4e^{-u^{2}}$
for all $u\ge0$, where $K_{\ast}$ is an absolute multiple of
$\sigma_{\mathrm{GS}}$, $C_{\mathrm D}$ is Dudley's constant,
$\Dud(T):=\int_0^\infty\sqrt{\log N(T,\eta)}\,d\eta$ is the entropy integral in
the Euclidean metric, and
$\Dud(T_A)\le21e^{\rho}\sqrt{\dk\log(1+\rho)}+8e^{\rho}\sqrt{\dv}$,
$\Dud(T_B)\le21e^{\rho}\sqrt{\dk\log(1+\rho)}+4e^{\rho}$.
\end{lemma}

\begin{proof}[Proof of Theorem~\ref{thm:main-alg}(a),(b)]
Full details are in Appendix~\ref{app:alg}. (a) Let $\bar D$ be an explicit
majorant of the sum of the two conditional expectation bounds. Run
Definition~\ref{def:level} for $\ell$ levels, $\ell$ maximal with
$2^{\ell+2}\bar D\le\eps n$. The checks hold surely, so sizes obey
$n/2^{j}\le n_j\le n/2^{j}+4\sigma_{\mathrm{GS}}$ and the key sum stays
bounded by $4\sigma_{\mathrm{GS}}\sqrt{\dk}$, whence by Jensen
$B_{P_j}\ge n_j/2$ once $n_j\ge\rho\cdot4\sigma_{\mathrm{GS}}\sqrt{\dk}/\log2$,
which the stopping rule guarantees. Lemma~\ref{lem:step}(iii) bounds the
level-$j$ error by $2(a_j+b_j)/n_j$ with $a_j,b_j$ the realised suprema, and
since the accepted signing at level $j$ depends on the history only through
$P_j$, the tower property and Lemma~\ref{lem:accept} give
$\E\sup_q\norm{\Attn_{P_\ell}-\Attn}\le\sum_{j<\ell}2^{j+1}\bar D/n\le\eps/2$;
Markov's inequality finishes. The size is $<8\bar D/\eps+4\sigma_{\mathrm{GS}}$
by maximality. (b) Apply the conditional tail with $u=\sqrt{\log(16L/\delta)}$,
$L=\max\{1,\lceil\log_2n\rceil\}$, at each of at most $L$ levels; a union
bound over $2L$ events leaves probability $1-\delta/2$ on which every
$a_j+b_j$ is at most $D^{\ast}=O(e^{\rho}(\sqrt{\dv}+\sqrt{\dk\log(1+\rho)}+\sqrt{\log(L/\delta)}))$,
and the error telescopes deterministically; $\log(16L/\delta)=O(\log(\log n/\delta))$.
\end{proof}

\begin{corollary}[Sample, then balance]
\label{cor:sample-balance}
Draw a uniformly random $S_0\subseteq[n]$ of size $s_0$, where
$s_0=\min\{n,\lceil2048e^{2\rho}/\eps^{2}\rceil\}$, redrawing until
$\norm{\bar k_{S_0}}\le2s_0^{-1/2}$ (a check costing $O(s_0\dk)$ time); centre
the sampled keys at $\bar k_{S_0}$, rescale keys and queries by the centred
radius, and run the algorithm of Theorem~\ref{thm:main-alg}(a) on
$(K_{S_0},V_{S_0})$ with target $\eps/2$. With probability at least $\tfrac14$
the output satisfies~\eqref{eq:coreset}; its size is
$O(e^{\rho}(\sqrt{\dv}+\sqrt{\dk\log(1+\rho)})/\eps)$ always, and the expected
running time is $O(n\dk)+O(s_0^{\omega+1}+s_0^{2}(\dk+\dv))$.
\end{corollary}

\begin{proof}
If $s_0=n$ the check passes trivially and Theorem~\ref{thm:main-alg}(a)
applies directly. Otherwise, $\E\norm{\bar k_{S_0}}^{2}\le1/s_0$ for sampling
without replacement from centred vectors of norm at most one, so by Markov
each draw passes the check $\norm{\bar k_{S_0}}\le2s_0^{-1/2}$ with
probability $\ge\tfrac34$, and the expected number of draws is at most
$\tfrac43$. By Theorem~\ref{thm:sampling} with $\eps/2$, a single draw has
$\sup_q\norm{\Attn_{S_0}-\Attn}\le\eps/2$ except with probability $\le\tfrac14$,
hence except with conditional probability $\le\tfrac13$ given the check.
Given the check, the sampled instance, re-centred at $\bar k_{S_0}$ and
rescaled, surely has centred radius $\kappa'\le1+2s_0^{-1/2}$ and parameter
$\rho'=\rho\kappa'\le\rho(1+2s_0^{-1/2})$. Since $s_0\ge2048e^{2\rho}$ gives
$2s_0^{-1/2}\le e^{-\rho}/(16\sqrt2)$ and $\rho e^{-\rho}\le1/e$, we have both
$\rho'\le\rho+0.017$ and $\rho'\le1.05\rho$; hence $e^{\rho'}\le1.02e^{\rho}$
and $\log(1+\rho')\le\log(1+1.05\rho)\le1.05\log(1+\rho)$, the last step by
concavity ($c\log(1+\rho)-\log(1+c\rho)$ is nondecreasing in $\rho$ for
$c\ge1$); so the size bound of Theorem~\ref{thm:main-alg}(a) on the sample,
which holds always, is of the stated order. Theorem~\ref{thm:main-alg}(a) on
the re-centred, rescaled sample with target $\eps/2$ succeeds with conditional
probability $\ge\tfrac12$, its randomness being independent of $S_0$; its
guarantee is uniform over the original ball $\norm q\le\rho$, because
translating the keys and jointly rescaling keys and queries leave every
$\Attn_{S}(q)$ unchanged; and the triangle inequality gives error $\le\eps$
with probability $\ge\tfrac23\cdot\tfrac12=\tfrac13\ge\tfrac14$. Times follow
from Theorem~\ref{thm:main-alg}(a) with $n$ replaced by $s_0$, plus $O(n\dk)$
to compute the centring and the sample and $O(s_0\dk)$ per check.
\end{proof}

\begin{remark}[When the sampling stage is active]
\label{rem:s0}
The corollary improves on Theorem~\ref{thm:main-alg}(a) only when
$s_0<n$, that is when $2048e^{2\rho}/\eps^{2}<n$. At $\eps=0.1$ this needs
$n>2\cdot10^{5}e^{2\rho}$, which no cache of at most $2\cdot10^{5}$ pairs
meets at any $\rho\ge0$.
\end{remark}

\begin{openproblem}[Verification]
\label{prob:verify}
Is there a polynomial-time algorithm that, given keys, values and a signing,
approximates $\sup_{\norm q\le\rho}\norm{\sum_i\sigma_ie^{q^{\!\top}k_i}v_i}$ within
a constant factor? A verifier would remove the additive term of
Theorem~\ref{thm:main-alg}(b) by repetition, and would amplify
Corollary~\ref{cor:sample-balance}. In fixed dimension Tai's grid certificate
is such a verifier for the third branch.
\end{openproblem}

% =====================================================================
\section{Fixed dimension: the logarithm disappears}
\label{sec:nolog}

Write $T(x,p):=e^{-\norm{x-p}^{2}}$ for the Gaussian kernel in Tai's
normalisation.

\begin{lemma}[Transfer]
\label{lem:transfer}
Let $\norm{k_i}\le1$ in $\R^{\dk}$ and $\rho>0$. Put
$k_i':=(k_i,\sqrt{1-\norm{k_i}^{2}})\in S^{\dk}$, $p_i:=\sqrt{\rho/2}\,k_i'$, and
$x_q:=(q,0)/\sqrt{2\rho}$ for $\norm q\le\rho$. Then
$e^{q^{\!\top}k_i}=e^{\norm q^{2}/(2\rho)+\rho/2}\,T(x_q,p_i)$, and for all real
$c_i$,
$\sup_{\norm q\le\rho}|\sum_ic_ie^{q^{\!\top}k_i}|\le e^{\rho}\sup_{x\in\R^{\dk+1}}|\sum_ic_iT(x,p_i)|$.
\end{lemma}

\begin{proof}
$\norm{x_q-p_i}^{2}=\norm q^{2}/(2\rho)+\rho/2-q^{\!\top}k_i$; exponentiate, use
$\norm q^{2}/(2\rho)\le\rho/2$, and enlarge the supremum.
\end{proof}

The identity is not new: for points on a sphere it is how~\cite[\S4.2]{BR23}
prove their Theorem~1.11, the attention form with a weighted context measure
is~\cite[eq.~(2)]{HKV26}, and the algebra is in~\cite[eq.~(3)]{Tai22}. What the lemma adds is the completion of the ball
to a sphere: every key lands on one sphere of radius $\sqrt{\rho/2}$, so the
weights $e^{\norm{k_i'}^{2}/2}$ of the weighted form are uniform and an
\emph{unweighted} colouring theorem applies.

\begin{fact}[Tai~\cite{Tai22}, Lemma~13 and Algorithm~1]
\label{fact:tai}
For every $D$ there is a finite $f_{\mathrm T}(D)$ such that every finite
$P\subset\R^{D}$ admits $\sigma\in\{\pm1\}^{P}$ with
$\sup_{x\in\R^{D}}|\sum_{p\in P}\sigma(p)T(x,p)|\le f_{\mathrm T}(D)$, with no
assumption on the diameter of $P$, computable by a Las Vegas algorithm in
expected time polynomial in $|P|$ for fixed $D$. The proof partitions $\R^{D}$
into the cells of side two centred at $(2\mathbb Z)^{D}$, colours each cell by
the Gram--Schmidt walk until a finite grid certifies
$|\sum_{p\in Q_g}\sigma_g(p)T(x,p)|\le C_De^{-\norm{x-g}^{2}/3}$ for all $x$ and
the cell is balanced, $|\sum_{p\in Q_g}\sigma_g(p)|\le1$, and sums the resulting convergent series
$\sum_{r\ge1}(2r+1)^{D}e^{-4(r-1)^{2}/3}$ over shells of cells.
\end{fact}

\begin{lemma}[Global balance]
\label{lem:balance}
The colouring of Fact~\ref{fact:tai} can be modified, without changing
$f_{\mathrm T}(D)$ or the running time, so that $|\sum_{p\in P}\sigma(p)|\le1$.
\end{lemma}

\begin{proof}
Lemma~13 of~\cite{Tai22} infers global balance from per-cell balance, which
fails when several cells have odd size with uncoordinated excesses. Negating
$\sigma_g$ on one cell preserves $|\sum_{p\in Q_g}\sigma_g(p)T(x,p)|$ pointwise,
hence the certificate. Even cells contribute $0$ to $\sum_p\sigma(p)$ and odd
cells $\pm1$; enumerate the odd cells $g_1,\dots,g_r$ and negate $\sigma_{g_j}$
whenever its contribution differs from $(-1)^{j}$. The total is then $0$ or
$\pm1$. Negating every second odd cell would not suffice, since the excesses
need not agree beforehand.
\end{proof}

\begin{proof}[Proof of Theorem~\ref{thm:main-nolog}, scalar part]
Apply Fact~\ref{fact:tai} with Lemma~\ref{lem:balance} to $(p_i)$ from
Lemma~\ref{lem:transfer} and use the inequality there with $c_i=\sigma_i$.
Coincident points may first be perturbed by $1/(2n)$ each, which moves every
signed kernel sum by at most $\tfrac12$ since $|T(x,p)-T(x,p')|\le\norm{p-p'}$;
this is the $+1$ in $g(\dk)\le f_{\mathrm T}(\dk+1)+1$. If some $c>0$ worked for
all large $\rho$ then $g(\dk)\ge c\sqrt{\dk\log(1+\rho)}\to\infty$.
\end{proof}

Values are absorbed geometrically.

\begin{lemma}[Two-coordinate encoding]
\label{lem:encode}
For $w\in[\tfrac12,1]$ let $R=\sqrt{\log(1/w)}<1$ and $a(w)=(R^{2},\sqrt{R^{2}-R^{4}})$.
Then $T((x,0),(p,a(w)))=wT(x,p)$ and $T((x,\tfrac12,0),(p,a(w)))=e^{-1/4}T(x,p)$.
\end{lemma}

\begin{proof}
$T((x,z),(p,a))=T(x,p)e^{-\norm{z-a}^{2}}$, $\norm a^{2}=R^{2}$, and
$\norm{(\tfrac12,0)-a}^{2}=(\tfrac12-R^{2})^{2}+R^{2}-R^{4}=\tfrac14$.
\end{proof}

\begin{theorem}[No logarithm: the coreset]
\label{thm:nolog-coreset}
Put $D=\dk+1+2\dv$ and $F(\dk,\dv)=7\sqrt{\dv}\,e^{\dv/4}(f_{\mathrm T}(D)+1)$.
Under~\eqref{eq:norm}, every $P\subseteq[n]$
admits a $(Fe^{\rho},Fe^{\rho},1)$-balanced signing, computable in expected
polynomial time for fixed $(\dk,\dv)$, and hence for every $\eps\in(0,1]$ an
unweighted $\eps$-coreset of size $\min\{n,8Fe^{\rho}/\eps+1\}$ exists and is
computable in expected polynomial time for fixed $(\dk,\dv)$.
\end{theorem}

\begin{proof}
For $i\in P$, $j\in[\dv]$ put $w_{ij}=(3+v_{ij})/4\in[\tfrac12,1]$ and
$\hat p_i=(p_i,a(w_{i1}),\dots,a(w_{i\dv}))\in\R^{D}$. Colour by
Fact~\ref{fact:tai} with Lemma~\ref{lem:balance} (after the perturbation of
the previous proof), obtaining $\sup_{\hat x}|\sum_i\sigma_iT(\hat x,\hat p_i)|\le f:=f_{\mathrm T}(D)+1$.
Reading the sum at $(x,(\tfrac12,0),\dots,(\tfrac12,0))$ and at the same point
with the $j$th pair replaced by $(0,0)$, Lemma~\ref{lem:encode} gives
$|\sum_i\sigma_iT(x,p_i)|\le e^{\dv/4}f$ and $|\sum_i\sigma_iT(x,p_i)w_{ij}|\le e^{(\dv-1)/4}f$
for all $x$. Since $v_{ij}=4w_{ij}-3$, each value coordinate of the signed sum
is bounded by $(4e^{(\dv-1)/4}+3e^{\dv/4})f\le7e^{\dv/4}f$, and the Euclidean
norm over $j$ by $\sqrt{\dv}$ times that. Lemma~\ref{lem:transfer} converts
both bounds at the cost $e^{\rho}$. The halving of Section~\ref{sec:reduction}
with $D_A=D_B=Fe^{\rho}$, $D_N=1$ gives size $\min\{n,4(D_A+D_B)/\eps+1\}$.
\end{proof}

\begin{remark}[Kernel-density reading]
\label{rem:br23}
Bozzai and Rothvoss ask in their conclusions whether the $\sqrt{\log\log n}$ of
their fixed-dimension Theorem~1.9 can be dropped for kernels other than the
Gaussian, and remark that Tai's technique ``does not appear to generalize to
other kernels''. For the exponential kernel $K_e^{(\alpha)}$ on $S^{d}$ and
the Hellinger kernel $\exp(-\alpha\norm{\sqrt x-\sqrt y}^{2})$ on the simplex
it does, through the identity they themselves use: both are the Gaussian
kernel $T$ on the point set $\sqrt{\alpha/2}\,x_i\subset\R^{d+1}$, respectively
$\sqrt\alpha\,\sqrt{x_i}\subset\R^{d}$, so Fact~\ref{fact:tai} with
Lemma~\ref{lem:balance} gives discrepancy $f_{\mathrm T}(d+1)+1$, respectively
$f_{\mathrm T}(d)+1$, uniformly in $\alpha$, and coreset complexity $O_d(1/\eps)$
by halving. This settles the Gaussian-restriction case of their question; the
Laplacian and Jensen--Shannon kernels remain open. In attention language it is
Theorem~\ref{thm:main-nolog} without the sphere completion.
\end{remark}

\begin{openproblem}[Growth of the diameter-free constant]
\label{prob:g}
Is $f_{\mathrm T}(D)=O(\sqrt D)$? This is the discrepancy form of the
$O(\sqrt d/\eps)$ question of~\cite[\S5]{PT18} and the first open question
of~\cite{BR23}. A positive answer would give attention coresets of size
$O(e^{\rho}\sqrt{\dv}e^{\dv/4}\sqrt{\dk+\dv}/\eps)$ through
Theorem~\ref{thm:nolog-coreset}, and removing the $e^{\dv/4}$ would need a
weighted version of Fact~\ref{fact:tai}.
\end{openproblem}

% =====================================================================
\section{Lower bounds}
\label{sec:lb}

Table~\ref{tab:lb} is the summary. Part~(b) of Theorem~\ref{thm:main-lb} is
one line: with $n=\dv$ and $v_i=e_i$, attention at $q=0$ is the mean
$\dv^{-1}\sum e_i$, and a subset of size $s$ has error
$\sqrt{1/s-1/\dv}$. Part~(a) is proved in Appendix~\ref{app:lb} by a
distributional one-way INDEX argument on the instance
$\{(u_i,b_{ij}x_{ij}e_j),(-u_i,0)\}$ built from a spherical code
$|\inner{u_i}{u_h}|\le1/\rho$; the antipodal keys make the instance centred,
and the $-u_h$ weights are the inverses of the $u_h$ weights, which keeps the
denominator $\le\dv(2e^{\rho}+2em)$. The floor $\min\{\dv,\eps^{-2}\}$ is the
$r\to0$ form of~\cite[Thm.~3.2]{KSHZK25}. On this instance weighting does not
help: at $q=0$ the output of any weighted subset is a probability vector
supported on $S$, whose distance from $\dv^{-1}\mathbf 1$ is at least
$\sqrt{(1-s/\dv)^{2}/s+(\dv-s)/\dv^{2}}$ by Cauchy--Schwarz, so weighted
subsets face the same floor. (On the lookup instance of~\cite[Lem.~B.4]{HKV26},
by contrast, two weighted atoms reproduce attention exactly at $q=0$ and at one
further query~\cite[Lem.~B.8]{HKV26}.)

\begin{table}[hbp]
\centering\small
\begin{tabular}{@{}p{3.3cm}p{4.9cm}p{4.4cm}p{2.7cm}@{}}
\toprule
Regime & Best upper bound & Best lower bound & Gap \\
\midrule
$\dk=\dv=d$, $\eps\sqrt d\gtrsim e^{-\rho}$, $\dk\gtrsim\rho^{3}$ &
  $e^{\rho}\sqrt{d\log(1+\rho)}/\eps$ &
  $e^{\rho}\sqrt d/\eps$ \quad(\ref{thm:main-lb}a, \cite{LAK26}) &
  $\sqrt{\log(1+\rho)}$ \\
$\dk=\dv=d$ fixed, $\rho\to\infty$ &
  $F(d,d)e^{\rho}/\eps$ &
  $\Omega(d)$ only (\ref{thm:main-lb}a; its $e^{\rho}$ branch needs $\dk\gtrsim\rho^{3}$) &
  $e^{\rho}/\eps$ unmatched from below \\
scalar values; $\dk\gtrsim\rho^{3}$ or $\dk=\mathrm{polylog}\,n$ &
  $e^{\rho}\min\{\sqrt{\dk\log(1+\rho)},\allowbreak F(\dk,1)\}/\eps$ &
  $e^{\rho}/\eps$ \quad(\ref{thm:main-lb}a at $\dv=1$; \ref{thm:main-lb}c) &
  whole $\dk$ factor \\
$\eps\ll e^{-\rho}$, $\dk=\Theta(\log^{a}n)$, scalar &
  $\min\{n,\,e^{\rho}(1+\sqrt{\dk\log(1+\rho)})/\eps\}$, i.e.\ $n^{o(1)}/\eps$ &
  $(1/\eps)^{1-1/a-o(1)}$ \quad(\ref{thm:main-lb}c) &
  $(1/\eps)^{1/a+o(1)}n^{o(1)}$ \\
$\rho\to0$ &
  $\min\{n,\eps^{-2},\ldots\}$ &
  $\min\{\dv,\eps^{-2}\}$ \quad(\ref{thm:main-lb}b) &
  none for $\dv\ge\eps^{-2}$ \\
\bottomrule
\end{tabular}
\caption{Lower-bound status under~\eqref{eq:norm}. The sampling branch $e^{2\rho}/\eps^{2}$ is unmatched from below except through the mean floor.}
\label{tab:lb}
\end{table}

\begin{fact}[\cite{CFIKKP26}, Lemmas 5.3, 6.7, 7.7 and Theorem 4.12]
\label{fact:cfikkp}
Let $r$ bound $\norm{k_i}$ and $\norm q$ for the un-normalised kernel
$e^{q^{\!\top}k}$. Their streaming attention problem (Definition~2.2) asks for
$\hat z$ with $\norm{\hat z-\Attn(q)}\le\eps\norm{\mathrm{softmax}(K,q)}_2\norm V_F$
with probability $1-\delta$ for each fixed query. The proofs of the lower
bounds in their Theorems~7.1 and~7.2 construct, from any data structure for
this problem occupying $S$ bits, a public-coin one-way protocol for
$\INDEX_{2m,m}$ with message length $S+o(m)$: the $2m$ keys are public random
points, the values are Alice's bits $v_i=x_i\in\{0,1\}$ (one-dimensional), and
Bob's query is a public key, rescaled in regime~(a). The only property of
the data structure used is that it is a function of Alice's input of at most
$S$ bits from which Bob's single query can be answered. With
$R^{\mathrm{pub},\to}(\INDEX_{2m,m})=\Omega(m)$ this gives:
(a) if $r=\log^{o(1)}n$, $d=\Theta(\log^{a}n)$ for a constant $a>2$, and
$n^{-100}<\eps<n^{-c}$, then $S\ge\tilde\Omega(\min\{n,(1/\eps)^{1-1/a-o(1)}\})$;
here the keys are $N(0,(r/2)^{2}I_d/d)$ and Bob's query is $k_j\cdot(r/2)/\norm{k_j}$;
(b) if $\Omega(\log n)\le r^{2}<\tfrac12\log n$ and $d=(\log n)^{1+\Omega(1)}$,
then $S\ge\tilde\Omega(\min\{n,e^{r^{2}(1-o(1))}/\eps\})$; here the keys lie on
the sphere of radius $r$ and Bob's query is a key.
\end{fact}

\begin{proof}[Proof of Theorem~\ref{thm:main-lb}(c)]
Consider the family of instances of Fact~\ref{fact:cfikkp}: $2m$ public keys,
values $v_i=x_i$, Bob's query $q$. With probability $1-o(1)$ over the keys
all norms are as required and, the keys being i.i.d.\ isotropic,
$\norm{\bar k}=o(r)$; call such key sets good. Rescale keys by $1/r$ and the
query by $r$, which leaves $e^{q^{\!\top}k_i}$ unchanged; translate the keys by
$-\bar k$, which leaves $\Attn$ unchanged; and rescale by the centred radius
$\kappa$. A good instance now satisfies~\eqref{eq:norm} with
$\rho=\Theta(r^{2})$ in regime~(a) (keys of norm $\approx r/2$, query of norm
$r/2$) and $\rho=r^{2}(1+o(1))$ in regime~(b); the regime hypotheses of
Fact~\ref{fact:cfikkp} are conditions on their $r$ and are unaffected by the
normalisation, and in~(b) $e^{r^{2}(1-o(1))}=e^{\rho(1-o(1))}$.

Suppose every good instance admitted a query-oblivious unweighted
$\eps$-coreset with $|S|\le s$. Exactly $m$ values equal $1$, so
$\norm V_F=\sqrt m$, and $\mathrm{softmax}(K,q)$ has $2m$ nonnegative entries
summing to $1$, so $\norm{\mathrm{softmax}(K,q)}_2\ge(2m)^{-1/2}$; hence
$|\Attn_S(q)-\Attn(q)|\le\eps$ at Bob's query meets their guarantee with
parameter $\sqrt2\,\eps$ and no failure probability. Alice, who knows $x$ and
the public keys, sends the indices in $S$ and the bits $x_i$, $i\in S$:
at most $\lceil\log_2(2m+1)\rceil+s(1+\lceil\log_2(2m)\rceil)=O(s\log m)$ bits,
the first term encoding $|S|$. Bob evaluates $\Attn_S(q)$ from the
public keys and the message. This is a data structure of $O(s\log m)$ bits
answering Bob's query on all good key sets, the bad ones charged to the
reduction's failure budget, so the reductions of Fact~\ref{fact:cfikkp}
give $O(s\log m)+o(m)\ge\Omega(m)$, i.e.\ $s\ge\Omega(m/\log m)$. Substituting
$m$ from the two regimes gives the two bounds.
\end{proof}

\begin{remark}[Where the transferred bounds sit]
Regime~(a) has $e^{\rho}=n^{o(1)}$ and $\eps<n^{-c}$, so $\eps\ll e^{-\rho}$:
there Theorem~\ref{thm:main-lb}(a) gives only
$c\,\dv\min\{e^{2\rho},\exp(\dk/C\rho^{2})\}=n^{o(1)}$ and the mean floor gives
$1$ for scalar values, while the transferred bound is polynomial in $1/\eps$.
Regime~(b) has $e^{\rho}\in[n^{\Omega(1)},n^{1/2+o(1)}]$ and the shape
$e^{\rho}/\eps$ of Theorem~\ref{thm:main-lb}(a) at $\dv=1$, with $e^{o(\rho)}$
and polylogarithmic losses in place of explicit constants and with $\dk$
polylogarithmic in $n$ in place of $\dk\gtrsim\rho^{3}$. Whether the
side-information decoding of~\cite{CFIKKP26} composes with multiplexing over
$\dv$ value coordinates to give $\sqrt{\dv}$ times these floors is open, as is
any lower bound with a factor of $\dk$.
\end{remark}

% =====================================================================
% Experiments section.  Every measured number is a macro defined in
% results_values.tex, which experiments/scripts/fill_results_tex.py generates
% from the run summaries and experiments/scripts/audit_generated_tables.py
% recomputes independently.  Protocol constants (sizes, seeds, the 8192-pair
% cap, eps = 0.1, 256 generated tokens) are typed; nothing else is.
% =====================================================================
\section{Experiments}
\label{sec:exp}

Every bound in this paper carries $e^{\rho}$, and whether that factor is
moderate on transformer caches is a question the theory cannot answer. Four
experiments follow: a census of $\rho$ over every layer and key--value head
of two transformers (E1); a comparison of coreset constructions at matched
size on real heads (E2); a closed-loop decoding test of the modelling
argument of Section~\ref{sec:why} (E3); and a scaling study on synthetic
instances, including the instance of Proposition~\ref{prop:sep} (E4). Each
experiment is reported in turn, then what the experiments establish
(Section~\ref{sec:exp-summary}).

% ---------------------------------------------------------------------
\subsection{Setup}
\label{sec:exp-setup}

\paragraph{Models and corpus.}
Qwen2.5-7B-Instruct~\cite{Qwen25} (\vEaQLayers\ layers, \vEaQKvHeads\
key--value heads, \vEaQGroup\ query heads per key--value head) and
Llama-3.1-8B-Instruct~\cite{Llama3} (\vEaLLayers\ layers, \vEaLKvHeads\
key--value heads, \vEaLGroup\ query heads per key--value head), both with
$\dk=\dv=\vEaQHeadDim$. The corpus holds 200 natural prompts of 4k--32k
tokens from LongBench~\cite{Bai24} and QASPER~\cite{Dasigi21} and 50 prompts
of random tokens as a control, each document filed under the longest of the
lengths 4096, 8192, 16384 and 32768 tokens that it fills to at least nine
tenths, so that no prompt is a short document padded to a length. A prompt that falls short of the
4096-token floor, or of its bucket by more than a tenth, under a model's own
tokeniser is skipped: Qwen measures \vEaQPromptsMeasured\ prompts and Llama
\vEaLPromptsMeasured\ (Appendix~\ref{app:corpus}).

\paragraph{The measured quantities.}
After the prefill, for every key--value head, the keys are read after the
rotary embedding and the queries after the rotary embedding with the model's
scaling $1/\sqrt{\dk}$ folded in, so that $q^{\!\top}k$ is the logit the model
computes. Then $\kappa=\max_i\norm{k_i-\bar k}$ is taken over the cache and
$\vrho$ over the decoding queries in two conventions: the maximum, which is
what~\eqref{eq:coreset} asks for, and the 99th percentile, which shows
whether one query decides the answer. Finally $\rho=\vrho\kappa$. Under
grouped-query attention one key--value head serves several query heads and
$\vrho$ is the maximum over the group. The uncentred product
$\max_t\norm{q_t}\max_i\norm{k_i}$ is recorded as a control: it is the
quantity a theorem phrased with $\max_i\norm{k_i}$ would use.

\paragraph{Methods.}
Five query-oblivious methods are compared. Uniform sampling draws the
target size without replacement. BalanceKV~\cite{KSHZK25} runs the
self-balancing walk~\cite{ALS21} inside halving on the numerator and
denominator features. Kernel halving~\cite{DM21,CGSDM25} is a
query-oblivious variant of the thinning of~\cite{CGSDM25} on the same
features. The construction of this paper is run in two variants: halving
with the Gram--Schmidt walk~\cite{BDGL18} on the Gram matrix of
Definition~\ref{def:level} with its two rejection checks
(Theorem~\ref{thm:main-alg}(a)), and the same halving after the input has
been re-centred at its key mean and rescaled
(Corollary~\ref{cor:sample-balance}). Two query-aware reference methods are
shown the queries the others may not see: a heavy-hitter oracle in the style
of SnapKV~\cite{Li24}, which keeps the keys with the largest attention mass
under the decoding queries, and the per-query reference, a scalarised form
of the coreset of Corollary~\ref{cor:perquery}, which is stated for scalar
values, built for the first decoding query with the vector values projected
on their first principal direction (Appendix~\ref{app:methods}). The set
handed to a Gram--Schmidt or self-balancing walk is capped at 8192 pairs by
uniform subsampling; the cap binds on \vEbQCapDumps\ of the \vEbQDumps\ E2
instances of each model (Section~\ref{sec:e2}) and on every E4 instance,
since there $n=2^{16}$. At a requested
size of 8192 pairs or more a capped method returns that subsample itself, so
the $n/4$ cell of the 32768-pair E2 instances and the $n/8$ and $n/4$ rows of
E4 are, for the construction and for BalanceKV, a uniform sample of 8192
pairs recorded under the requested size (for the re-centred variant, a draw
made under its key-mean rejection rule).
Kernel halving runs on the whole cache, an
asymmetry in its favour on the capped instances. BalanceKV and kernel halving
are reimplementations, the latter a query-oblivious variant of the published
method without its refinement step and its queries;
Appendix~\ref{app:methods} specifies both.

The two variants of the construction are one pipeline. The sampling stage of
Corollary~\ref{cor:sample-balance} draws
$s_0=\min\{n,\lceil2048e^{2\rho}/\eps^{2}\rceil\}$ pairs; at $\eps=0.1$ the
second term is at least $2048\cdot10^{2}$ already at $\rho=0$, so $s_0=n$ on
every cache of at most $204\,800$ pairs, which includes every instance below,
real or synthetic. The 8192-pair cap is therefore the only subsampling, and
the two variants are Theorem~\ref{thm:main-alg}(a) run on the capped input,
differing in whether that input is re-centred at its own key mean and in the
random seed. This is why their cells in Table~\ref{tab:e2} differ by at most
two hundredths; on the synthetic instances their fitted slopes differ by
about a tenth at most (Table~\ref{tab:e4slopefull}) and their size ratios
by up to three hundredths (\vEdRatioRc\ against \vEdOsbRatioRc).

\paragraph{Metric, query populations, seeds.}
The error of a subset $S$ at a query is
$\norm{\Attn_S(q)-\Attn(q)}_2/\max_i\norm{v_i}$, the unit of~\eqref{eq:norm},
and every cell reports the worst case over a query population. Two
populations are used and never pooled: the model's own decoding queries, of
which an instance has up to \vEbQRealQueriesMax\ on Qwen and
\vEbLRealQueriesMax\ on Llama, and $10^{4}$ queries drawn uniformly from the
sphere of radius $\vrho$. The sphere is a subset of the ball
that~\eqref{eq:coreset} quantifies over and $10^{4}$ draws are a subset of
the sphere, so the sphere column is a lower bound on the supremum
in~\eqref{eq:coreset}; the real column is what a deployed cache meets and is
outside the theorem. Randomised methods are run with \vEbQSeeds\ seeds and a
cell is the median over seeds and instances pooled; the two reference methods
are deterministic and their median is over instances; Figure~\ref{fig:e2}
shows the range over seeds. Hardware and wall-clock times are in
Appendix~\ref{app:hardware}.

% ---------------------------------------------------------------------
\subsection{E1: the census of centred \texorpdfstring{$\rho$}{rho}}
\label{sec:e1}

Table~\ref{tab:e1} and Figure~\ref{fig:e1} report $\rho$ over every layer and
key--value head. Each head contributes one value, the median over the prompts
it was measured on, and the medians and fractions of the table are taken over
heads.

\begin{table}[htbp]
\centering\small
\begin{tabular}{@{}lrrrrrrr@{}}
\toprule
& \multicolumn{2}{c}{median $\rho$} & \multicolumn{2}{c}{median $e^{\rho}$}
& \multicolumn{2}{c}{fraction of heads} & uncentred \\
\cmidrule(lr){2-3}\cmidrule(lr){4-5}\cmidrule(lr){6-7}
Model & max & p99 & max & p99 & $e^{\rho}\le10^{3}$ & $e^{\rho}\ge10^{6}$ & median $\rho$ \\
\midrule
Qwen2.5-7B-Instruct & \vEaQrhoMax & \vEaQrhoPnn & \vEaQerhoMax & \vEaQerhoPnn & \vEaQfracLo & \vEaQfracHi & \vEaQuRhoMax \\
Llama-3.1-8B-Instruct & \vEaLrhoMax & \vEaLrhoPnn & \vEaLerhoMax & \vEaLerhoPnn & \vEaLfracLo & \vEaLfracHi & \vEaLuRhoMax \\
\midrule
& \multicolumn{5}{c}{median $\rho$ by layer quintile, shallowest to deepest} & & \\
\cmidrule(lr){2-6}
Model & first & second & third & fourth & fifth & & \\
\midrule
Qwen2.5-7B-Instruct & \vEaQqA & \vEaQqB & \vEaQqC & \vEaQqD & \vEaQqE & & \\
Llama-3.1-8B-Instruct & \vEaLqA & \vEaLqB & \vEaLqC & \vEaLqD & \vEaLqE & & \\
\bottomrule
\end{tabular}
\caption{E1. Centred $\rho$ over all layers and key--value heads at 4k--32k
tokens: \vEaQKvHeads\ heads in \vEaQLayers\ layers on Qwen
(\vEaQRows\ measurements) and \vEaLKvHeads\ heads in \vEaLLayers\ layers on
Llama (\vEaLRows). The maximum and the 99th percentile are the two conventions
for $\vrho$; the last column is the uncentred control
$\max_t\norm{q_t}\max_i\norm{k_i}$ in the maximum convention. Lower block:
medians over the heads of each layer quintile, the quintiles being layers
\vEaQBands\ on Qwen and \vEaLBands\ on Llama.}
\label{tab:e1}
\end{table}

\begin{figure}[hbp]
\centering
\includegraphics[width=\textwidth]{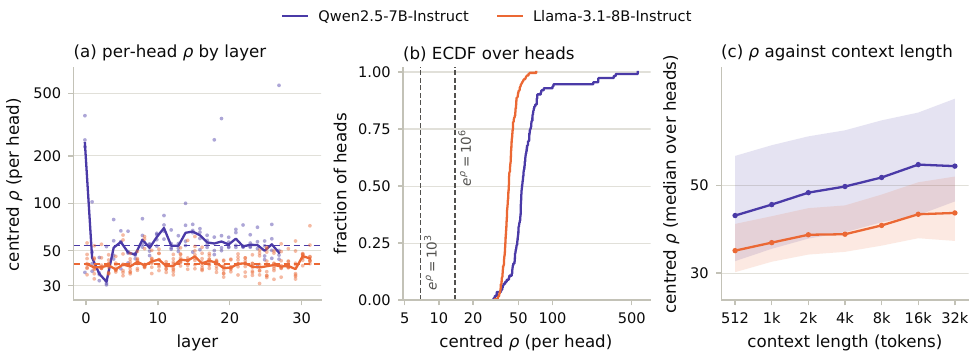}
\caption{E1. (a) Centred $\rho$ of every key--value head against its layer,
one point per head (the median over prompts), with the median over the heads
of each layer joined and the overall median dashed. (b) Empirical distribution
of the per-head $\rho$ on both models; the dashed lines mark
$e^{\rho}=10^{3}$ and $e^{\rho}=10^{6}$. (c) Median $\rho$ against context
length from 512 to 32768 tokens, with the band from the 10th to the 90th
percentile over heads in each model's colour; the corpus at each length is
the set of documents that fill it.}
\label{fig:e1}
\end{figure}

The median head has $\rho=\vEaQrhoMax$ on Qwen and \vEaLrhoMax\ on Llama in
the maximum convention, and \vEaQrhoPnn\ and \vEaLrhoPnn\ at the 99th
percentile, so no single outlier query is responsible. Over heads $\rho$ runs
from \vEaQrhoMinHead\ to \vEaQrhoMaxHead\ on Qwen, with the 10th and 90th
percentiles at \vEaQrhoPtenHead\ and \vEaQrhoPninetyHead, and from
\vEaLrhoMinHead\ to \vEaLrhoMaxHead\ on Llama, with percentiles
\vEaLrhoPtenHead\ and \vEaLrhoPninetyHead. At the median head $e^{\rho}$ is
\vEaQerhoMax\ and \vEaLerhoMax; $e^{\rho}\ge10^{6}$ holds on \vEaQfracHi\ and
\vEaLfracHi\ of heads and $e^{\rho}\le10^{3}$ on \vEaQfracLo\ and \vEaLfracLo;
the sampling branch $128e^{2\rho}/\eps^{2}$ of Theorem~\ref{thm:sampling} is
below the cache size at $\eps=0.1$ on \vEaQfracSampling\ and
\vEaLfracSampling\ of heads. The two factors of $\rho$ separate at the median
head: the median $\kappa$ is \vEaQkappaMed\ on Qwen and \vEaLkappaMed\ on
Llama, the median $\vrho$ is \vEaQvrhoMed\ and \vEaLvrhoMed, so there the
query radius is near two and the centred key radius carries $\rho$; on the
Qwen heads above $\rho=200$ the roles reverse, with an ordinary $\kappa$ and
a query radius several times the median. Random-token prompts give a smaller
$\rho$ than natural text, \vEaQrhoControl\ against \vEaQrhoNatural\ on Qwen
and \vEaLrhoControl\ against \vEaLrhoNatural\ on Llama, and of the same
order: at the median head the text raises $\rho$ by a fifth on Qwen and an
eighth on Llama, and the model sets its order.

\paragraph{Layer profile, query-head groups and attention sinks.}
The medians by layer quintile in Table~\ref{tab:e1} and the per-head profile
in Figure~\ref{fig:e1}(a) show no band of layers approaching the thresholds
marked in Figure~\ref{fig:e1}(b). Two explanations of a large $\rho$ were tested
(Table~\ref{tab:e1gqa}, Appendix~\ref{app:exptables}); neither brings any
head near the regime. Within a group of query heads sharing one key--value
head, the largest and smallest query radii differ by a median factor of
\vEaQgqaRatio\ on Qwen and \vEaLgqaRatio\ on Llama, so at the median head no
single query head sets $\rho$ for its group; the Qwen heads above $\rho=200$
are the exception, where the largest query radius in the group is many times
the smallest: on all but one of them $\rho$ under the second-largest query
head falls to within a factor of about two of the median head, while at
L27H01 two query heads share the excess. Removing
the 32 keys furthest from the centroid leaves $\kappa$ at
\vEaQtrimThirtytwo\ and \vEaLtrimThirtytwo\ of its value at the median head,
and at no less than two thirds of it on any head, so $\kappa$ is a bulk
property of the key cloud and not the work of a few attention sinks.

\paragraph{The uncentred control.}
The uncentred product has median \vEaQuRhoMax\ on Qwen and \vEaLuRhoMax\ on
Llama, so $e^{\rho}$ of \vEaQuErhoMax\ and \vEaLuErhoMax. On Qwen, whose key
projection carries a bias vector, the uncentred key radius equals the norm of
that bias to within \vEaQBiasGapLarge\ on the \vEaQBiasHeadsLarge\ heads
whose bias norm exceeds 100 (the largest being \vEaQBiasMax), and the
centring of~\eqref{eq:rho} removes it. The centred convention gives the
smaller $\rho$ on every
Qwen head and on all but one Llama head, where the two differ by under
three percent, and it is still far outside the regime.

\paragraph{Context length.}
Three further censuses per model at 512, 1024 and 2048 tokens extend the
range of lengths to a factor of
\vScQLever\ (Table~\ref{tab:scale} in Appendix~\ref{app:exptables},
Figure~\ref{fig:e1}(c)). A least-squares fit of $\log\rho$ on $\log n$ over
the seven lengths gives $\rho\propto n^{\beta}$ with
$\beta=\vScQExp\pm\vScQExpSe$ on Qwen and $\vScLExp\pm\vScLExpSe$ on Llama,
with bootstrap intervals over heads of $[\vScQExpLo,\vScQExpHi]$ and
$[\vScLExpLo,\vScLExpHi]$; the growth is mostly in $\kappa$ (exponents
\vScQExpKappa\ and \vScLExpKappa), with $\vrho$ carrying the rest
(\vScQExpVrho\ and \vScLExpVrho). The corpus is not the same at every
length, since only long documents fill the long buckets, so part of the
trend is a change of corpus. At the
shortest length, \vScQShortest\ tokens, the smallest $\rho$ of any head is
\vScQRhoMinHead\ on Qwen and \vScLRhoMinHead\ on Llama. For comparison, the
discrepancy branch of~\eqref{eq:envelope} at $\eps=0.1$ and $\dk=\dv=128$ is below
the cache size, with its constant set to one, only for $\rho<2.9$ when
$n=4096$ and only for $\rho<6$ when $n=10^{5}$.

% ---------------------------------------------------------------------
\subsection{E2: constructions at matched size}
\label{sec:e2}

\paragraph{Design.}
Nine heads per model, stratified by terciles of the E1 distribution of
$\rho$, were each captured at six prompts spanning cache lengths from
\vEbQLengthsMin\ to \vEbQLengthsMax\ pairs on Qwen and from \vEbLLengthsMin\
to \vEbLLengthsMax\ on Llama, giving \vEbQDumps\ instances per model. Sizes
$n/64$, $n/16$ and $n/4$ are shown; $n/32$ and $n/8$ were also swept, and
the mean error was measured beside the worst case; neither is shown.
Instances whose softmax is a near hard max, with a mean participation ratio
of at most 16 keys over the decoding queries (Appendix~\ref{app:defs}), are
excluded from every cell: on
such an instance any method that keeps the few keys carrying the mass scores
well, and a tie is not evidence. The exclusion removes
\vEbQHardMaxDumps\ of \vEbQDumps\ instances on Qwen and \vEbLHardMaxDumps\
of \vEbLDumps\ on Llama, among them every instance of \vEbQHardMaxHeads\
heads on Qwen and of \vEbLHardMaxHeads\ on Llama. On Qwen the tercile design
therefore does not survive into Table~\ref{tab:e2}, which rests on
\vEbQKeptDumps\ instances of \vEbQKeptHeads\ heads on Qwen and
\vEbLKeptDumps\ instances of all \vEbLKeptHeads\ heads on Llama.

\begin{table}[hbp]
\centering\small
\begin{tabular}{@{}lrrrrrr@{}}
\toprule
& \multicolumn{3}{c}{real decoding queries} & \multicolumn{3}{c}{sphere of radius $\vrho$}\\
\cmidrule(lr){2-4}\cmidrule(lr){5-7}
Method & $n/64$ & $n/16$ & $n/4$ & $n/64$ & $n/16$ & $n/4$ \\
\midrule
\multicolumn{7}{@{}l@{}}{\itshape Qwen2.5-7B-Instruct}\\
Uniform sampling & \vEbQUnifRealA & \vEbQUnifRealB & \vEbQUnifRealC & \vEbQUnifRandA & \vEbQUnifRandB & \vEbQUnifRandC \\
BalanceKV & \vEbQBkvRealA & \vEbQBkvRealB & \vEbQBkvRealC & \vEbQBkvRandA & \vEbQBkvRandB & \vEbQBkvRandC \\
Kernel halving & \vEbQThinRealA & \vEbQThinRealB & \vEbQThinRealC & \vEbQThinRandA & \vEbQThinRandB & \vEbQThinRandC \\
\textbf{Halving} & \vEbQOhRealA & \vEbQOhRealB & \vEbQOhRealC & \vEbQOhRandA & \vEbQOhRandB & \vEbQOhRandC \\
\textbf{Halving, re-centred subsample} & \vEbQOsbRealA & \vEbQOsbRealB & \vEbQOsbRealC & \vEbQOsbRandA & \vEbQOsbRandB & \vEbQOsbRandC \\
\addlinespace
Heavy hitters (oracle) & \vEbQHhRealA & \vEbQHhRealB & \vEbQHhRealC & \vEbQHhRandA & \vEbQHhRandB & \vEbQHhRandC \\
Per-query & \vEbQPqRealA & \vEbQPqRealB & \vEbQPqRealC & \vEbQPqRandA & \vEbQPqRandB & \vEbQPqRandC \\
\midrule
\multicolumn{7}{@{}l@{}}{\itshape Llama-3.1-8B-Instruct}\\
Uniform sampling & \vEbLUnifRealA & \vEbLUnifRealB & \vEbLUnifRealC & \vEbLUnifRandA & \vEbLUnifRandB & \vEbLUnifRandC \\
BalanceKV & \vEbLBkvRealA & \vEbLBkvRealB & \vEbLBkvRealC & \vEbLBkvRandA & \vEbLBkvRandB & \vEbLBkvRandC \\
Kernel halving & \vEbLThinRealA & \vEbLThinRealB & \vEbLThinRealC & \vEbLThinRandA & \vEbLThinRandB & \vEbLThinRandC \\
\textbf{Halving} & \vEbLOhRealA & \vEbLOhRealB & \vEbLOhRealC & \vEbLOhRandA & \vEbLOhRandB & \vEbLOhRandC \\
\textbf{Halving, re-centred subsample} & \vEbLOsbRealA & \vEbLOsbRealB & \vEbLOsbRealC & \vEbLOsbRandA & \vEbLOsbRandB & \vEbLOsbRandC \\
\addlinespace
Heavy hitters (oracle) & \vEbLHhRealA & \vEbLHhRealB & \vEbLHhRealC & \vEbLHhRandA & \vEbLHhRandB & \vEbLHhRandC \\
Per-query & \vEbLPqRealA & \vEbLPqRealB & \vEbLPqRealC & \vEbLPqRandA & \vEbLPqRandB & \vEbLPqRandC \\
\bottomrule
\end{tabular}
\caption{E2. Worst-case relative error at matched coreset size, lower is
better. A cell is the median over the \vEbQKeptDumps\ (Qwen) or
\vEbLKeptDumps\ (Llama) instances that are not near hard max and over
\vEbQSeeds\ seeds, pooled; the last two rows of each panel, set off by a
gap, are the query-aware references and are deterministic. The real column is the worst case over the
model's decoding queries and the sphere column over $10^{4}$ queries of norm
$\vrho$. Bold marks the two variants of the construction of this paper.}
\label{tab:e2}
\end{table}

\begin{figure}[hbp]
\centering
\includegraphics[width=\textwidth]{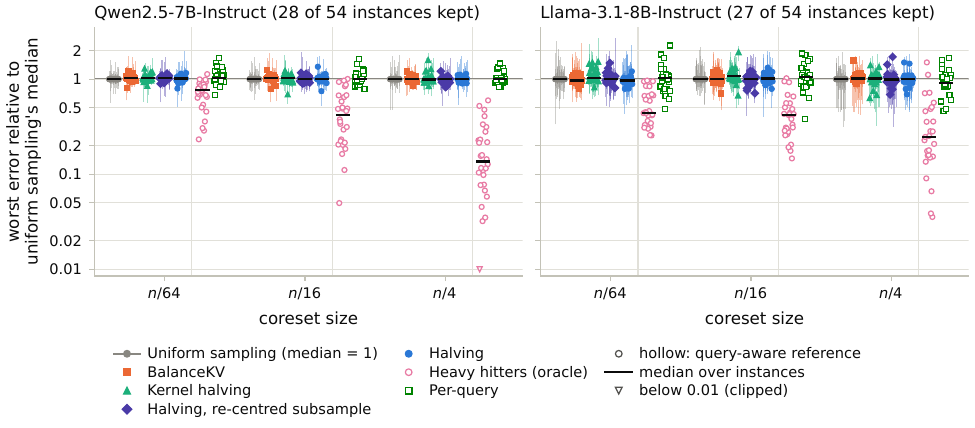}
\caption{E2. Worst-case error on the real decoding queries relative to
uniform sampling: for each kept instance, method and size, the method's
median over seeds divided by uniform sampling's median over seeds, one point
per instance, the whisker spanning the method's range over seeds, and the bar
the median over instances; uniform sampling's own strip shows its range over
seeds around one. Hollow markers are the two query-aware references, set
off from the query-oblivious methods by a rule. Instances near hard max are
excluded as in Table~\ref{tab:e2}.}
\label{fig:e2}
\end{figure}

\paragraph{The result on the real decoding queries.}
On the real decoding queries the five query-oblivious rows of
Table~\ref{tab:e2} lie within five hundredths of one another at every size
on both models; only the two query-aware references lie further from uniform
sampling, the oracle below it at every size and the per-query reference
above it at $n/64$ on Qwen. Counting (instance, size) cells over the five
swept sizes, the construction beats uniform sampling in \vEbQOhWins\ of
\vEbQOhCells\ cells on Qwen and \vEbLOhWins\ of \vEbLOhCells\ on Llama, and
its re-centred variant in \vEbQOsbWins\ of \vEbQOsbCells\ and \vEbLOsbWins\
of \vEbLOsbCells. The per-cell ratio of the construction's error to uniform
sampling's has median \vEbQRatioMed\ on Qwen and \vEbLRatioMed\ on Llama,
with 80\% of cells between \vEbQRatioQlo\ and \vEbQRatioQhi\ on Qwen and
between \vEbLRatioQlo\ and \vEbLRatioQhi\ on Llama; Figure~\ref{fig:e2}
shows the per-instance ratios of every method at the three tabulated sizes.
The cells are not independent, since instances of one head at different
lengths share the head and the two variants share the pipeline, so the
result is stated as an effect size: on these heads, in at least eight cells
of ten, the construction's worst-case error is within sixteen percent of
uniform sampling's, in either direction. The ratio by tercile of
$\rho$ is \vEbQTercRatioA, \vEbQTercRatioB\ and \vEbQTercRatioC\ on Qwen and
\vEbLTercRatioA, \vEbLTercRatioB\ and \vEbLTercRatioC\ on Llama (terciles
cut on all \vEbQDumps\ dumped instances, over which the kept cells are
unequally spread), so the gap does not widen with $\rho$. Without the near-hard-max exclusion the
construction's count is \vEbQBeatsRealAll\ of \vEbQCellsRealAll\ on Qwen and
\vEbLBeatsRealAll\ of \vEbLCellsRealAll\ on Llama (both variants pooled).

% Table 4 and Figure 2 fill the page before this paragraph; without the
% break the paragraph began under them and continued after them.
\newpage
\paragraph{Matched error.}
The same sweep read at matched error rather than matched size asks, for each
instance and each swept size $T$, for the smallest swept size at which a
method reaches uniform sampling's error at $T$ (Table~\ref{tab:matched},
Appendix~\ref{app:exptables}). Over the cells in which a method reaches that
error at some swept size, no query-oblivious method needs fewer keys than
uniform sampling in more than \vMeQMaxFrac\ of cells on Qwen and
\vMeLMaxFrac\ on Llama, where the oracle does in \vMeQHhFrac\ and
\vMeLHhFrac. The two readings agree.

\paragraph{The sphere queries.}
On the sphere of radius $\vrho$ uniform sampling has worst-case error
\vEbQUnifRandA\ on Qwen and \vEbLUnifRandA\ on Llama at $n/64$, and still
\vEbQUnifRandC\ and \vEbLUnifRandC\ at $n/4$; every query-oblivious method
lies within three hundredths of it at every size, and the oracle lies above
it at each of the three sizes shown on both models, by \vEbLHhRandB\
against \vEbLUnifRandB\ at $n/16$ on Llama at the widest. The arithmetic
predicts this: a query of norm $\vrho$ at these $\rho$ spans a logit range of up
to $2\rho$, so the softmax is close to a hard max in most directions, and the
worst case over directions is then the distance between the value of the
winning key and the value of its replacement in $S$, of order one for any
$S$ that misses one extreme key. This is consistent with a size bound larger
than $n$: at these $\rho$ no method tested here, the per-query reference
included, brings the worst case over the $10^{4}$ sampled directions below
\vEbQThinRandC\ on Qwen or \vEbLPqRandC\ on Llama at a quarter of the
cache.

\paragraph{Arithmetic of the walk.}
Two failure mechanisms independent of the constants would also produce a
null: that at the measured $\rho$ the Gram matrix of
Definition~\ref{def:level} is numerically diagonal, or that it collapses to
near rank one, so that the walk balances noise. Recomputing the Gram matrix
from the E2 instances (Table~\ref{tab:kernel}, Appendix~\ref{app:exptables})
refutes both below $\rho=\vKnQCollapse$: the pivoted Cholesky factorisation
reaches its rank cap of \vKnQCap\ on \vKnQOperSat\ of Qwen's instances and
\vKnLOperSat\ of Llama's, the median largest off-diagonal correlation of the
exponential blocks lies between \vKnQAbLo\ and \vKnQAbHi\ by band on Qwen
and between \vKnLAbLo\ and \vKnLAbHi\ on Llama, and the largest relative
residual of any factorisation is \vKnQResid\ on Qwen and \vKnLResid\ on
Llama. Above $\rho=\vKnQCollapse$, where Qwen has \vKnQHeadsHi\ of its
\vKnQHeads\ heads and Llama \vKnLHeadsHi, the largest off-diagonal
correlation of the exponential blocks falls below $10^{-2}$ on every instance
and below the degeneracy threshold of $10^{-8}$ on \vKnQDegenHi\ of the
\vKnQDumpsHi\ instances, so that there the walk balances little beyond the
key-sum and cardinality blocks; on Llama, whose $[50,100)$ band has median
\vKnLAbLo, about a quarter of the kept instances lie below $10^{-2}$ as
well, none below $10^{-8}$. The census therefore rules out arithmetic
failure below $\rho=\vKnQCollapse$; it does not show that the directions the
walk balanced are the ones that matter for the attention error at a real
query, and at $\rho$ near \vEaQrhoRound\ the exponential kernel
$e^{\rho(k_i^{\!\top}k_j-1)}$ is negligible except within tight clusters of
keys. The \vKnQDumpsHi\ Qwen instances above $\rho=\vKnQCollapse$ are all
among the \vEbQKeptDumps\ that Table~\ref{tab:e2} keeps, so on Qwen, as on
Llama, the null rests in part on instances on which the walk had little to
balance.

\paragraph{The low-$\rho$ tail, the cap and the clean instances.}
Restricted to the third of instances with the smallest $\rho$,
$\rho\le\vEbQTailRhoMax$ on
Qwen (\vEbQTailKeptDumps\ instances kept of \vEbQTailDumps) and
$\rho\le\vEbLTailRhoMax$ on Llama (\vEbLTailKeptDumps\ of \vEbLTailDumps),
the construction wins \vEbQTailBeats\ of \vEbQTailCells\ and \vEbLTailBeats\
of \vEbLTailCells\ cells (both variants pooled). The cap of 8192 pairs binds
on every instance whose cache exceeds it, the two longest of each head's six
prompts: \vEbQCapDumps\ of \vEbQDumps\ instances from \vEbQCapHeads\ heads on
Qwen and \vEbLCapDumps\ of \vEbLDumps\ from \vEbLCapHeads\ on Llama. On the
\vEbQCleanDumps\ (Qwen) and \vEbLCleanDumps\ (Llama)
instances that are neither capped, nor near hard max, nor above
$\rho=\vKnQCollapse$ (three heads on Qwen, all nine on Llama), the construction
wins \vEbQBeatsRealClean\ of \vEbQCellsRealClean\ and \vEbLBeatsRealClean\
of \vEbLCellsRealClean\ cells (both variants pooled).

\paragraph{Discussion.}
At the median head the size that the discrepancy branch
of~\eqref{eq:envelope} prescribes exceeds the cache by many orders of
magnitude, and the construction matches uniform sampling: the constants of
Section~\ref{sec:alg} dominate at these $n$. The same $e^{\rho}$ appears in the whole-ball bounds of uniform sampling and of
the construction, BalanceKV and kernel halving carry no whole-ball bound at
all, and none of the four separates from uniform sampling here by more than
five hundredths of worst-case error at any size.

% ---------------------------------------------------------------------
\subsection{E3: closed-loop decoding}
\label{sec:e3}

\paragraph{Design.}
On the nine E2 heads of each model, \vEcQSteps\ tokens are generated from a
compressed cache in two ways: open loop, where the tokens of uncompressed
decoding are fed back at every step regardless of what the masked model
would produce, and closed loop, where the model decodes freely and its
queries are computed from the compressed cache. Eviction is simulated by
masking the evicted positions at the decode steps only; the prefill is never
masked, so the cache being compressed is exactly the cache the coreset was
built from, and only the nine selected heads are compressed. Each loop's
error is the worst case over the queries it issued. Four methods were run,
uniform sampling, BalanceKV, the re-centred variant of the construction and
the heavy-hitter oracle, at requested sizes $n/16$ and $n/4$, on the six E2
documents per model re-tokenised to at most \vEcQTokenLength\ tokens, which
shortens the two longest E2 instances of each model and lengthens the
instances E2 held at its 4096-token bucket; the realised caches hold
\vEcQCacheMin\ to \vEcQCacheMax\ pairs on Qwen and \vEcLCacheMin\ to
\vEcLCacheMax\ on Llama. This is a subset of the E2 methods on caches that
coincide with the E2 instances only where neither truncation applied; the
oracle is fitted on the queries of the reference generation.

\paragraph{Results.}
The closed/open ratio is taken over the rows in which the closed loop
diverged from the uncompressed model, because a row that tracks it to the
end is the open loop and
contributes exactly one; over all rows the oracle's median ratio is
\vEcQHhRatioAll\ on Qwen and \vEcLHhRatioAll\ on Llama. Pooled over
sizes, the oracle's ratio over diverged rows is \vEcQHhRatio\ on Qwen and
\vEcLHhRatio\ on Llama; the re-centred variant's is \vEcQOsbRatio\ and
\vEcLOsbRatio, uniform sampling's \vEcQUnifRatio\ and \vEcLUnifRatio,
BalanceKV's \vEcQBkvRatio\ and \vEcLBkvRatio. Among the oracle's diverged
rows at $n/16$ the closed loop is worse than the open loop in \vEcQHhBGt\
rows, better in \vEcQHhBLt\ and equal in \vEcQHhBEq\ on Qwen (\vEcLHhBGt,
\vEcLHhBLt\ and \vEcLHhBEq\ on Llama); at $n/4$ the counts are \vEcQHhCGt,
\vEcQHhCLt\ and \vEcQHhCEq\ on Qwen and \vEcLHhCGt, \vEcLHhCLt\ and
\vEcLHhCEq\ on Llama. Aggregated by head, the oracle's median ratio exceeds
one on \vEcQHhBHeadsAbove\ of \vEcQHhBHeads\ heads at $n/16$ and
\vEcQHhCHeadsAbove\ of \vEcQHhCHeads\ at $n/4$ on Qwen, and on
\vEcLHhBHeadsAbove\ of \vEcLHhBHeads\ and \vEcLHhCHeadsAbove\ of
\vEcLHhCHeads\ on Llama. For the three query-oblivious methods neither
direction dominates: the counts lean either way by model and size, with a
share of rows at exactly one whose worst error fell inside the prefix both
loops share; for the re-centred variant at $n/16$ they are \vEcQOsbBGt,
\vEcQOsbBLt\ and \vEcQOsbBEq\ on Qwen and \vEcLOsbBGt, \vEcLOsbBLt\ and
\vEcLOsbBEq\ on Llama. The oracle's
absolute error is the lowest at every size, so its closed-loop penalty is
measured on a method with something to lose. The query-oblivious methods
enter the closed loop with worst-case error already between \vEcQUnifOpenB\
and \vEcQOsbOpenB\ at $n/16$ on Qwen and with token agreement below a quarter
on both models; their closed-loop tokens leave the uncompressed model's after
a median of at most ten steps in every query-oblivious row of
Table~\ref{tab:e3} but two, uniform sampling at $n/4$ on Qwen and the
re-centred variant at $n/4$ on Llama, so the absence of a closed-loop
penalty is in part a ceiling effect. The
per-step trajectories are Figure~\ref{fig:e3} in
Appendix~\ref{app:exptables}, and the ratio of the mean error from the first
divergent step on is in Table~\ref{tab:e3post}.

\begin{table}[hbp]
\centering\footnotesize
\setlength{\tabcolsep}{4.5pt}
\begin{tabular}{@{}llrrrrrr@{}}
\toprule
Method & size & open & closed & closed/open & diverged & tracked & agreement \\
\midrule
\multicolumn{8}{@{}l@{}}{\itshape Qwen2.5-7B-Instruct}\\
\vEcQBody
\midrule
\multicolumn{8}{@{}l@{}}{\itshape Llama-3.1-8B-Instruct}\\
\vEcLBody
\bottomrule
\end{tabular}
\caption{E3. Worst-case relative error over the queries each loop issued,
median over the \vEcQRows\ rows (prompt, head) of each size; the closed/open
ratio is the median over the rows in which the closed loop diverged from the
uncompressed model, counted in the ``diverged'' column. ``Tracked'' is the
median number of decode steps before the closed loop first diverged and
``agreement'' the median fraction of the \vEcQSteps\ generated tokens that
agree with the uncompressed model. Above one the closed loop is worse.}
\label{tab:e3}
\end{table}

\paragraph{Discussion.}
The heavy-hitter oracle degrades in the closed loop over the diverged rows
on both models, and the three query-oblivious methods do not. The argument of
Section~\ref{sec:why} concerns guarantees tied to queries fixed in advance:
the oracle, fitted to the queries, shows the penalty, while BalanceKV, whose
guarantee is per query, shows none at this ceiling (\vEcQBkvRatio,
\vEcLBkvRatio), so the experiment does not separate the form of a method's
guarantee from its having been fitted to the queries: it is consistent with
the argument and is not a strong test of it.

% ---------------------------------------------------------------------
\subsection{E4: scaling on synthetic instances}
\label{sec:e4}

\paragraph{Design.}
Instances with $\dk\in\{2,8,64\}$, $\dv\in\{1,8\}$, $\rho\in\{1,2,4,6\}$ and
$n=2^{16}$, keys uniform on the sphere or clustered, \vEdSeeds\ seeds
(Appendix~\ref{app:e4}). For each method the size needed to reach $\eps=0.1$
is the crossing of a least-squares fit of $\log(\text{worst error})$ on
$\log(\text{requested size})$, the worst error being over 2000 queries of
norm $\rho$, over the sizes from $n/16384$ to $n/4$ in factors of two for
uniform sampling and kernel halving and from $n/16384$ to $n/16$ for the
three capped methods, whose rows at $n/8$ and $n/4$ are the 8192-pair cap
subsample itself (Section~\ref{sec:exp-setup}). The slope of
$\log(\text{size})$ against $\rho$ is fitted over $\rho\le4$ per seed and
reported as the median over seeds, because at $\rho=6$, $\dk=64$, $\dv=8$ the
discrepancy branch of~\eqref{eq:envelope},
$e^{\rho}(\sqrt{\dv}+\sqrt{\dk\log(1+\rho)})/\eps$, comes within a factor
$1.2$ of $n$ before constants. A fit is excluded when at least half of its crossings fall outside
the sizes it was fitted over and marked $\dagger$ when any does
(Appendix~\ref{app:e4}); of the \vEdFits\ fits, \vEdCensoredFits\ are
excluded, \vEdFullDaggered\ are marked, and \vEdOffGridFrac\ of all
\vEdOffGridTotal\ crossings lie outside those sizes.

\paragraph{The slope in $\rho$.}
A size proportional to $e^{\rho}$ has slope one. A slope within $1\pm0.1$
over $\rho\le4$ is found in \vEdFullWithin\ of the \vEdFits\ fits, and for
the construction in \vEdOhWithin\ of its \vEdConstructionCells\
configurations (its re-centred variant in \vEdOsbWithin).
Theorem~\ref{thm:main-alg}(a) is an upper bound, and a tolerance of a tenth
about slope one treats it as tight; the bound's own size expression
$e^{\rho}(\sqrt{\dv}+\sqrt{\dk\log(1+\rho)})$ has slope between $1.05$ and
$1.12$ over $\rho\in\{1,2,4\}$ on the six $(\dk,\dv)$ pairs, the fitted slopes
range up to \vEdMaxSlope, and their ranges over seeds are often several
tenths wide (Table~\ref{tab:e4slope}). The sweep is consistent with the
bound and does not determine its rate: three values of $\rho$ fitted over
ten seeds do not fix an exponent to a tenth. Both variants of the
construction have a higher slope than every baseline in \vEdOursHighest\
configurations, so the difference in scaling is systematic rather than one
configuration; a lower slope for a baseline is compatible with the bounds,
since uniform sampling's (Theorem~\ref{thm:sampling}) is an upper bound of
slope two in $\rho$ and the other baselines carry no whole-ball bound.

\begin{table}[hbp]
\centering\small
\begin{tabular}{@{}llrrr@{}}
\toprule
& & \multicolumn{3}{c}{fitted slope of $\log(\text{size})$ against $\rho$, and its range over seeds}\\
\cmidrule(lr){3-5}
Keys & Method & $\dk=2$ & $\dk=8$ & $\dk=64$ \\
\midrule
\vEdSlopeBody
\bottomrule
\end{tabular}
\caption{E4. Scaling of the size needed to reach $\eps=0.1$ in $\rho$, fit
over $\rho\in\{1,2,4\}$, median over \vEdSeeds\ seeds with the range over
seeds in brackets, for uniform sampling and Halving; all \vEdFits\ fits are
in Table~\ref{tab:e4slopefull}. Bold marks a slope within $1\pm0.1$;
$\dagger$ as in Appendix~\ref{app:e4}.}
\label{tab:e4slope}
\end{table}

\paragraph{Size at the tolerance.}
Table~\ref{tab:e4ratio} and Figure~\ref{fig:e4} hold the error fixed at
$\eps$ and read the size. The construction needs the smaller coreset in
\vEdRatioWinsRa\ of \vEdRatioCellsRa\ configurations at $\rho=1$,
\vEdRatioWinsRb\ of \vEdRatioCellsRb\ at $\rho=2$ and \vEdRatioWinsRc\ of
\vEdRatioCellsRc\ at $\rho=4$, with median ratios \vEdRatioRa, \vEdRatioRb\
and \vEdRatioRc\ (the re-centred variant: \vEdOsbRatioRa, \vEdOsbRatioRb\
and \vEdOsbRatioRc): between \vEdGainLo\ and \vEdGainHi\ times fewer pairs
than uniform sampling at the median configuration, an advantage that narrows
as $\rho$ grows and closes on the sphere at $\dk=64$ by $\rho=4$. The one
cell printed as a dash, clustered keys at $\dk=2$, $\dv=1$ and $\rho=1$, is one on
which the construction reaches $\eps$ below the smallest swept size in most
seeds, so that its crossing would be an extrapolation below the grid. At a
fixed size of \vEfSize\ of \vEfN\ pairs the construction's worst error is
\vEfOhRa\ against uniform sampling's \vEfUnifRa\ at $\rho=1$ and \vEfOhRd\
against \vEfUnifRd\ at $\rho=6$ (Table~\ref{tab:e4fixed},
Appendix~\ref{app:exptables}). The worst error over the
2000 test queries is a lower bound on the supremum over the ball, so the
fitted size-to-$\eps$ is a lower bound on the size the guarantee needs: what
is verified is that this lower bound is smaller than uniform sampling's in
\vEdRatioWinsRa\ of \vEdRatioCellsRa, \vEdRatioWinsRb\ of \vEdRatioCellsRb\
and \vEdRatioWinsRc\ of \vEdRatioCellsRc\ configurations at $\rho=1,2,4$,
not the supremum guarantee itself and not its rate in $\rho$, which no
polynomial-time verifier is known to certify (Open
Problem~\ref{prob:verify}). The worst error is also not monotone in the size
on this grid, decreasing in only \vEdMonotone\ of the (configuration, method,
$\rho$, seed) groups at $\rho\le4$, with a median upward step of
\vEdViolation; the fit over the swept sizes is what makes the crossing well
defined. The prediction of Theorem~\ref{thm:main-nolog} at $\dk=2$, that the
size stops growing beyond $e^{\rho}$, is not tested here: it needs the cell-wise
colouring of Fact~\ref{fact:tai}, and the Gram--Schmidt construction's own
factor $\sqrt{\log(1+\rho)}$ cannot be told
from a constant over $\rho\in[1,4]$.

\begin{figure}[hbp]
\centering
\includegraphics[width=\textwidth]{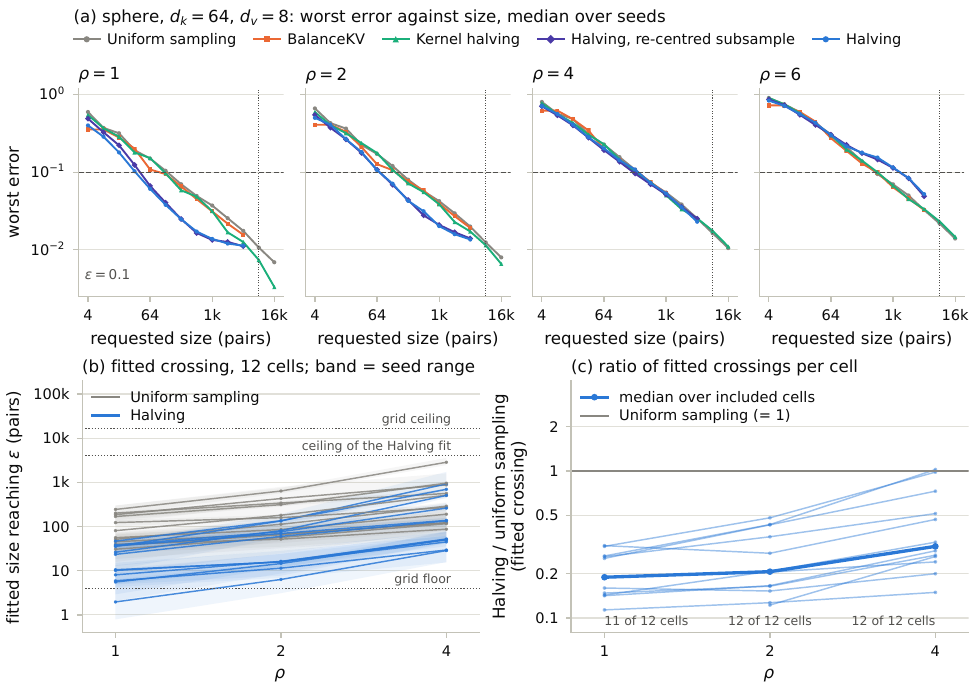}
\caption{E4. (a) Worst-case error against the requested coreset size at each
$\rho$ on the sphere with $\dk=64$, $\dv=8$, median over seeds, with the
tolerance $\eps=0.1$ dashed and the 8192-pair cap dotted; the curves of the
three capped methods end at $n/16$, the largest size below the cap. (b)
Fitted size to reach $\eps$ against $\rho$ for Halving and uniform sampling
over the twelve (keys, $\dk$, $\dv$) configurations, with the range over seeds
shaded; the dotted lines are the smallest and largest swept sizes and the
largest size below the cap, which bounds the Halving fit. (c) The ratio of
the two, by configuration, as in Table~\ref{tab:e4ratio}; the count below
each $\rho$ is the number of configurations whose ratio is defined.}
\label{fig:e4}
\end{figure}

\begin{table}[hbp]
\centering\small
\begin{tabular}{@{}lrrr@{}}
\toprule
& \multicolumn{3}{c}{size to reach $\eps$: construction over uniform sampling}\\
\cmidrule(lr){2-4}
Configuration & $\rho=1$ & $\rho=2$ & $\rho=4$ \\
\midrule
\vEdRatioBody
\bottomrule
\end{tabular}
\caption{E4. Ratio of the fitted size at which Halving
(Theorem~\ref{thm:main-alg}(a)) reaches $\eps=0.1$ to the size at which
uniform sampling does, median over the seeds in which both crossings lie
inside the sizes their fits used; a cell prints a dash when fewer than half
the seeds qualify. Below one the construction needs the smaller coreset. The
fit uses the sizes below the cap; over all thirteen sizes the ratios are
larger.}
\label{tab:e4ratio}
\end{table}

\begin{table}[htbp]
\centering\small
\begin{tabular}{@{}rrrr@{}}
\toprule
& \multicolumn{2}{c}{query-oblivious signing at $q_z$} & per-query greedy at $q_z$ \\
\cmidrule(lr){2-3}\cmidrule(lr){4-4}
$\dk$ & bound $\rho\sqrt r/e$ & $r\sinh(\rho/\sqrt r)$ & worst of \vSepSeeds\ seeds \\
\midrule
$256$  & \vSepKaBound & \vSepKaWitness & \vSepKaPq \\
$1024$ & \vSepKbBound & \vSepKbWitness & \vSepKbPq \\
$4096$ & \vSepKcBound & \vSepKcWitness & \vSepKcPq \\
\bottomrule
\end{tabular}
\caption{The two inequalities of Proposition~\ref{prop:sep} on its instance at
$\rho=1$. Any balanced signing, the walk's included, sits at
$r\sinh(\rho/\sqrt r)$, above the bound; the greedy stays below
$e^{\rho}=2.718$.}
\label{tab:sep}
\end{table}

\paragraph{The separation instance.}
On the instance of Proposition~\ref{prop:sep}(ii) with $n=r=\dk$ and
$\rho=1$, at $\dk\in\{256,1024,4096\}$ and \vSepSeeds\ seeds,
Table~\ref{tab:sep} gives the signed sum at the witness query $q_z$ of the
proof for a balanced signing, which is what the Gram--Schmidt walk on the
pivoted-Cholesky features returns (its signings have
$\sum_j\sigma_j\in\{0,\pm2\}$ and the best seed at each dimension is
balanced; the factor is truncated at $\dk\ge1024$, where the rank cap of
\vKnQCap\ binds), and the value of the greedy of
Proposition~\ref{prop:sep}(i) at the same query. Both inequalities hold: a
balanced signing sits above $\rho\sqrt r/e$ and the greedy below $e^{\rho}$
at every dimension. The factor $e$ between the bound and the balanced value
is slack in the proof, most of it from the step $e^{-\rho/\sqrt r}\ge e^{-1}$,
a worst case over $\rho\le\sqrt r$; and the greedy's value is small because
at $q_z$ every weight $e^{q_z^{\!\top}c_j}$ lies within a factor of about
$e^{\rho/\sqrt r}$ of one, which is $1.06$ at $\dk=256$. The witness value is
a lower bound on the supremum over the ball.

% ---------------------------------------------------------------------
\subsection{What the experiments establish}
\label{sec:exp-summary}

On every key--value head of two instruction-tuned models of seven and eight
billion parameters, at
every context length from 512 to 32768 tokens, the centred parameter $\rho$
is no smaller than \vScQRhoMinHead, reaches \vEaQrhoMaxHead\ at the largest
head at 4k--32k tokens, and has medians of \vEaQrhoRound\ and \vEaLrhoRound\
there. It grows as $n^{\vScQExp}$ and $n^{\vScLExp}$ with the context length
over a corpus that changes with the length, and it is a bulk property of the
key cloud rather than the work of a few sinks and, outside the Qwen heads
above $\rho=200$, not of one query head. At these values every whole-ball
bound that carries a factor $e^{\rho}$, this paper's included, prescribes a
coreset larger than the cache it compresses, and on real heads the construction of
this paper and the reimplemented baselines match uniform sampling at matched
size, the five query-oblivious rows of Table~\ref{tab:e2} lying within five
hundredths of one another at every size, and at matched error, where none
needs fewer keys in more than \vMeQMaxFrac\ of cells on Qwen and
\vMeLMaxFrac\ on Llama, while the heavy-hitter oracle has lower error in
every real-query cell of Table~\ref{tab:e2}, from \vEbQHhRealA\ against
\vEbQUnifRealA\ at $n/64$ on Qwen to \vEbLHhRealC\ against \vEbLUnifRealC\
at $n/4$ on Llama. Where
$\rho\le4$ is reached, on synthetic instances, the construction needs between
\vEdGainLo\ and \vEdGainHi\ times fewer pairs than uniform sampling at the
median configuration; the fitted slope of its size in $\rho$ lies within
$1\pm0.1$ in \vEdOhWithin\ of its \vEdConstructionCells\ configurations and
runs up to \vEdMaxSlope, against the bound's own $1.05$ to $1.12$, so the
sweep is consistent with the rate of Theorem~\ref{thm:main-alg}(a) and does
not determine it. In the closed loop the query-aware oracle degrades and the
query-oblivious methods do not. Every bound in this paper is stated in the
centred convention; the uncentred product is measured beside it
(Table~\ref{tab:e1}), and a bound stated in another convention has not been
evaluated on these models.

% =====================================================================
\section{Conclusion}
\label{sec:conc}

Two reductions to known theorems narrow the Liberty--Andoni--Kleiner gap to
$\sqrt{\log(1+\rho)}$ against their value-direction floor, and Tai's theorem
removes even that in fixed dimension.

The fixed-dimension branch is governed by the growth in $D$ of Tai's constant
(Open Problem~\ref{prob:g}), and is exponential in $\dv$ and, as far as is
known, in $\dk$; no lower bound of order $e^{\rho}/\eps$ is known against it,
since Theorem~\ref{thm:main-lb}(a) forces $e^{\rho}$ only once
$\dk\gtrsim\rho^{2}(\rho+\log\frac1{\eps\sqrt{\dv}})$. For scalar values no
lower bound carries a factor of $\dk$, and the sampling branch is unmatched
from below. The construction succeeds with constant probability at the stated
size, which a verifier for the signing objective
(Open Problem~\ref{prob:verify}) would remove.

Section~\ref{sec:exp} answers the two empirical questions of the
introduction: on the two models measured $\rho$ is near \vEaQrhoRound\ and
\vEaLrhoRound\ at the median head and no smaller than \vScQRhoMinHead\ on
any head, so every whole-ball bound with a factor $e^{\rho}/\eps$ is larger
than the cache it compresses, and on real heads the construction matches
uniform sampling; where $\rho\le4$ is reached, on synthetic instances, it
needs \vEdGainLo\ to \vEdGainHi\ times fewer pairs. The algorithmic
contribution is therefore asymptotic on the two models measured.
Two routes could change that: a guarantee over
the realised query distribution rather than the ball, which is a different
theorem and forfeits the property Proposition~\ref{prop:sep} is about; or an
architecture whose key geometry is tighter, of which models with normalised
or capped attention logits are candidates.

% =====================================================================
\appendix

\section{Tools}
\label{app:tools}

\begin{fact}[Banaszczyk~\cite{Ban98}]
\label{fact:ban}
Let $K\subseteq\R^{N}$ be a convex body (compact, convex, nonempty interior;
symmetry not required) with $\gam_N(K)\ge\tfrac12$, $\gam_N$ the standard
Gaussian measure on $\R^{N}$, and $x_1,\dots,x_m\in\R^{N}$
with $\norm{x_i}\le1$. Then some $\sigma\in\{\pm1\}^{m}$ has $\sum_i\sigma_ix_i\in5K$.
\end{fact}

\begin{fact}[Gram--Schmidt walk~\cite{BDGL18}, Thm.~1.4]
\label{fact:gsw}
There is a randomised algorithm which, given $x_1,\dots,x_m\in\R^{N}$ with
$\norm{x_i}\le1$, runs in time $O(m(m+N)^{\omega})$ and outputs $\sigma$ such
that $y=\sum_i\sigma_ix_i$ satisfies $\E e^{\inner\theta y}\le e^{\sigma_{\mathrm{GS}}^{2}\norm\theta^{2}/2}$
for all $\theta$, with $\sigma_{\mathrm{GS}}=\sqrt{40}$. Consequently $\E y=0$,
$\E\inner\theta y^{2}\le\sigma_{\mathrm{GS}}^{2}\norm\theta^{2}$, and
$\norm{\inner\theta y}_{\psi_2}\le C_\psi\sigma_{\mathrm{GS}}\norm\theta$ for an
absolute $C_\psi\ge1$ (\cite[Prop.~2.5.2]{Ver18}).
\end{fact}

\begin{fact}[Dudley's inequality, expectation and tail; \cite{Ver18}, Thms.~8.1.3 and~8.1.6]
\label{fact:dudley}
Let $(X_t)_{t\in T}$ have subgaussian increments $\norm{X_t-X_s}_{\psi_2}\le Kd(t,s)$.
Then for every $t_0$, and for every $u\ge0$ with probability at least
$1-2e^{-u^{2}}$,
\begin{gather*}
  \E\sup_{t}(X_t-X_{t_0})\le C_{\mathrm D}K\!\int_0^\infty\!\sqrt{\log N(T,d,\eta)}\,d\eta,\\
  \sup_{t,s}|X_t-X_s|\le C_{\mathrm D}K\Bigl[\int_0^\infty\!\sqrt{\log N(T,d,\eta)}\,d\eta+u\diam(T)\Bigr].
\end{gather*}
A centred Gaussian process $X_t=\inner gt$, $t\in T\subset\R^{N}$, has such
increments with $K=K_{\mathrm G}$ absolute in the Euclidean metric.
\end{fact}

\begin{fact}[Borell--TIS; \cite{AT07}, Thm.~2.1.1]
\label{fact:btis}
For a centred Gaussian process bounded a.s.\ on $T$ with
$\sigma^{2}=\sup_t\E X_t^{2}$, $Z=\sup_tX_t$ satisfies
$\Prob(Z>\E Z+r)\le e^{-r^{2}/(2\sigma^{2})}$.
\end{fact}

\begin{fact}[Covering numbers; \cite{Ver18}, Cor.~4.2.13]
\label{fact:cover}
For $0<\eta\le1$, $N(S^{d-1},\eta)\le(3/\eta)^{d}$ and $N(B^{d},\eta)\le(3/\eta)^{d}$.
\end{fact}

\begin{fact}[Spherical caps; {\cite[Lemma~2.2]{Ball97}}]
\label{fact:cap}
For $u$ uniform on $S^{d-1}$ and a fixed unit vector $v$,
$\Prob(|\inner uv|\ge t)\le2e^{-dt^{2}/2}$ for all $t\ge0$.
\end{fact}

\begin{lemma}[Distributional one-way INDEX]
\label{lem:index}
Let $x\in\{0,1\}^{N}$ be uniform, let $t\in[N]$ be uniform and independent, and
let public randomness $R$ be independent of both. If Alice sends
$M=M(x,R)$ of at most $b$ bits and Bob outputs $\hat x_t=\hat x_t(M,t,R)$ with
$\Prob(\hat x_t=x_t)\ge1-\delta$ for some $\delta\le\tfrac12$, then $b\ge(1-h(\delta))N$, $h$ the binary
entropy.
\end{lemma}

\begin{proof}
For each fixed $t$ let $\delta_t$ be Bob's error probability, so
$\tfrac1N\sum_t\delta_t\le\delta$. By Fano's inequality
$H(x_t\mid M,R)\le h(\delta_t)$. The bits of $x$ are independent, so
$b\ge H(M\mid R)\ge I(x;M\mid R)=N-H(x\mid M,R)\ge N-\sum_tH(x_t\mid M,R)\ge\sum_t(1-h(\delta_t))\ge N(1-h(\delta))$,
the last step by concavity of $h$ and its monotonicity on $[0,\tfrac12]$. (This is the average-case form of the
classical bound of~\cite{KNR99}.)
\end{proof}

% =====================================================================
\section{Proof of the first branch of Theorem~\ref{thm:main}}
\label{app:coreset}

Throughout, \eqref{eq:norm} holds. To keep every Gaussian process
finite-dimensional we truncate the feature map: for $M\in\mathbb N$ let
$\psi_M(z)=\bigoplus_{m\le M}z^{\otimes m}/\sqrt{m!}$ in
$E_{d,M}=\bigoplus_{m\le M}(\R^{d})^{\otimes m}$, so
$\inner{\psi_M(y)}{\psi_M(z)}=e_M(y^{\!\top}z)$ with $e_M(s)=\sum_{m\le M}s^{m}/m!$.

\begin{lemma}[Truncation]
\label{lem:trunc}
For all $y,z$, all $M$ and all $u,u'$ in a Hilbert space:
(i) $\norm{\psi_M(y)\otimes u-\psi_M(z)\otimes u'}\le\norm{\psi(y)\otimes u-\psi(z)\otimes u'}$
and $\norm{\psi_M(z)}\le\norm{\psi(z)}$;
(ii) if $|s|\le\rho$ then $|e^{s}-e_M(s)|\le\tau_M(\rho):=\sum_{m>M}\rho^{m}/m!$,
and $\tau_M(\rho)\le2^{-M}$ whenever $M+1\ge2e\rho$.
\end{lemma}

\begin{proof}
(i) The squared norm on the left is
$\sum_{m\le M}(\norm y^{2m}\norm u^{2}+\norm z^{2m}\norm{u'}^{2}-2(y^{\!\top}z)^{m}\inner u{u'})/m!$,
each summand nonnegative by $2|(y^{\!\top}z)^{m}\inner u{u'}|\le\norm y^{2m}\norm u^{2}+\norm z^{2m}\norm{u'}^{2}$;
the untruncated squared norm is the same sum over all $m$. (ii) The first
bound is termwise; for $m\ge M+1\ge2e\rho$, $\rho^{m}/m!\le(e\rho/m)^{m}\le2^{-m}$.
\end{proof}

\paragraph{Data and test vectors.}
Fix $P\subseteq[n]$, $|P|=m$, and $M$. Let $E=(E_{\dk,M}\otimes\R^{\dv})\oplus E_{\dk,M}\oplus\R$
and, for $i\in P$,
\begin{equation}
\label{eq:data}
  x_i:=\tfrac12\bigl(e^{-\rho/2}\psi_M(\sqrt\rho\,k_i)\otimes v_i,\;e^{-\rho/2}\psi_M(\sqrt\rho\,k_i),\;1\bigr),
  \qquad\norm{x_i}^{2}\le\tfrac34 ,
\end{equation}
by Lemma~\ref{lem:trunc}(i) and $\norm{k_i}\le1$. Test vectors, for $\norm q\le\rho$
and $u\in S^{\dv-1}$:
\[
t_A(q,u)=2(e^{\rho/2}\psi_M(q/\sqrt\rho)\otimes u,0,0),\quad
t_B(q)=2(0,e^{\rho/2}\psi_M(q/\sqrt\rho),0),\quad t_N=2(0,0,1),
\]
and $T_A=\{t_A(q,u)\}$, $T_B=\{\pm t_B(q)\}$, $T_N=\{\pm t_N\}$. For
$y_\sigma=\sum_{i\in P}\sigma_ix_i$,
\begin{equation}
\label{eq:readoff}
  \inner{y_\sigma}{t_A(q,u)}=\sum_{i\in P}\sigma_ie_M(q^{\!\top}k_i)\inner{v_i}u,\qquad
  \inner{y_\sigma}{t_B(q)}=\sum_{i\in P}\sigma_ie_M(q^{\!\top}k_i),\qquad
  \inner{y_\sigma}{t_N}=\sum_{i\in P}\sigma_i ,
\end{equation}
and $\norm{t_A},\norm{t_B}\le2e^{\rho/2}e^{\norm q^{2}/(2\rho)}\le2e^{\rho}$, $\norm{t_N}=2$.

\paragraph{Metric entropy.}
Write $\Phi(q)=e^{\rho/2}\psi(q/\sqrt\rho)$ (untruncated), $Q=\{\Phi(q):\norm q\le\rho\}$.

\begin{lemma}[Lipschitz bound]
\label{lem:lip}
For $\norm q,\norm{q'}\le\rho$, $\norm{\Phi(q)-\Phi(q')}\le\Lambda\norm{q-q'}/\rho$
with $\Lambda=e^{\rho}\sqrt{\rho(1+\rho)}$.
\end{lemma}

\begin{proof}
$\Phi$ is $C^{1}$ (the partial sums are polynomial and their derivatives
converge uniformly on bounded sets), and for $h\in\R^{\dk}$,
$\norm{D\Phi(q)[h]}^{2}=\partial_h\partial_{h'}\inner{\Phi(q)}{\Phi(q')}|_{q'=q,h'=h}
=e^{\rho}e^{\norm q^{2}/\rho}(\norm h^{2}/\rho+(q^{\!\top}h)^{2}/\rho^{2})\le e^{2\rho}\norm h^{2}(1/\rho+1)$.
Integrate along the segment.
\end{proof}

\begin{lemma}[Entropy integral of $Q$]
\label{lem:dudQ}
$I_Q:=\int_0^\infty\sqrt{\log N(Q,\eta)}\,d\eta\le5.2\,e^{\rho}\sqrt{\dk\log(1+\rho)}$ for every $\rho>0$.
\end{lemma}

\begin{proof}
Let $\Delta=\diam(Q)\le\min\{2e^{\rho},2\Lambda\}$. By Lemma~\ref{lem:lip} and
Fact~\ref{fact:cover}, $N(Q,\eta)\le(3\Lambda/\eta)^{\dk}$ for $\eta\le\Lambda$;
for $\Lambda\le\eta\le\Delta$ the point $\Phi(0)$ is an $\eta$-net, so the same
bound holds on $(0,\Delta]$, and $N=1$ beyond $\Delta$. With $a=\Delta/\Lambda\le2$
and $\eta=\Lambda s$,
$I_Q\le\Lambda\sqrt{\dk}\int_0^{a}\sqrt{\log(3/s)}\,ds\le\Lambda\sqrt{\dk}\,\varphi(a)$,
$\varphi(a):=a(\sqrt{\log(3/a)}+\sqrt\pi/2)$, which is increasing on $(0,2]$ (its
derivative at $a=2$ exceeds $0.73$ and decreases in $a$). If $\rho\ge1$, use
$a\le a_0=2/\sqrt{\rho(1+\rho)}$, $\Lambda a_0=2e^{\rho}$, and
$\log(3/a_0)\le\log(\tfrac32(1+\rho))\le1.6\log(1+\rho)$, $0.887\le1.07\sqrt{\log(1+\rho)}$,
giving $I_Q\le2e^{\rho}\sqrt{\dk}(1.27+1.07)\sqrt{\log(1+\rho)}<4.7e^{\rho}\sqrt{\dk\log(1+\rho)}$.
If $\rho<1$, use $a\le2$: $I_Q\le\Lambda\sqrt{\dk}\varphi(2)<3.047\Lambda\sqrt{\dk}\le4.31e^{\rho}\sqrt{\rho\dk}$
with $\Lambda\le e^{\rho}\sqrt{2\rho}$, and $\sqrt\rho\le1.2012\sqrt{\log(1+\rho)}$
by concavity, giving $I_Q<5.2e^{\rho}\sqrt{\dk\log(1+\rho)}$.
\end{proof}

\begin{lemma}[Entropy integrals of the test sets]
\label{lem:dudT}
With $W_k=e^{\rho}\sqrt{\dk\log(1+\rho)}$ and $W_v=e^{\rho}\sqrt{\dv}$,
$\Dud(T_A)\le21W_k+8W_v$, $\Dud(T_B)\le21W_k+4e^{\rho}$, and
$\diam(T_A),\diam(T_B)\le4e^{\rho}$. The same bounds hold for the untruncated
sets $\widetilde T_A=\{2\Phi(q)\otimes u\}$, $\widetilde T_B=\{\pm2\Phi(q)\}$.
\end{lemma}

\begin{proof}
By Lemma~\ref{lem:trunc}(i) the map $2\Phi(q)\otimes u\mapsto t_A(q,u)$ is a
contraction, so it suffices to treat $\widetilde T_A$. Since
$\norm{\Phi(q)\otimes u-\Phi(q')\otimes u'}\le\norm{\Phi(q)-\Phi(q')}+e^{\rho}\norm{u-u'}$,
the product of an $\eta/4$-net of $Q$ and an $\eta/(4e^{\rho})$-net of $S^{\dv-1}$
is an $\eta$-net of $\widetilde T_A$; the integral runs over $\eta\le4e^{\rho}$,
and $\sqrt{\log(ab)}\le\sqrt{\log a}+\sqrt{\log b}$ gives
$\Dud(T_A)\le4I_Q+4e^{\rho}\int_0^{1}\sqrt{\dv\log(3/\eta)}\,d\eta\le4I_Q+4e^{\rho}\sqrt{\dv}(\sqrt{\log3}+\sqrt\pi/2)\le21W_k+8W_v$.
For $T_B$ the sphere is replaced by $\{\pm1\}$, contributing $4e^{\rho}\sqrt{\log2}\le4e^{\rho}$.
\end{proof}

By Dudley's inequality $w(Q)\le C_{\mathrm D}K_{\mathrm G}I_Q$ and
$w(T)=\tfrac12\E\sup_{\widetilde T_A}\inner gt\le\tfrac12C_{\mathrm D}K_{\mathrm G}(21W_k+8W_v)$,
which is the upper bound quoted after Lemma~\ref{lem:chevet}.

\paragraph{A body of measure at least one half.}
Put
\begin{equation}
\label{eq:theta}
  \theta_A=C_{\mathrm D}K_{\mathrm G}(21W_k+8W_v)+3.6e^{\rho},\qquad
  \theta_B=C_{\mathrm D}K_{\mathrm G}(21W_k+4e^{\rho})+3.6e^{\rho},\qquad
  \theta_N=2\sqrt5 .
\end{equation}

\begin{lemma}[Product body]
\label{lem:body}
Let $g$ be standard Gaussian on $E$, $K_X=\{y:\sup_{T_X}\inner yt\le\theta_X\}$
and $K=K_A\cap K_B\cap K_N$. Then $K$ is closed, symmetric and convex, the
three events $\{g\in K_X\}$ are independent, each has probability $\ge0.8$,
so $\gam(K)\ge0.512$; and with $N=\dim E$, $R=10\sqrt N$, the set $K_R=K\cap RB_E$
is a convex body with $\gam(K_R)\ge0.502$ and $5K_R\subseteq5K$.
\end{lemma}

\begin{proof}
Each $K_X$ is an intersection of closed half-spaces and $T_X=-T_X$, so $K_X$ is
closed, convex, symmetric; the functionals depend on disjoint blocks of $y$,
whose Gaussian coordinates are independent. Block $N$:
$\Prob(2|g_N|>2\sqrt5)\le\tfrac15$ by Markov. Blocks $A,B$: $T_A$ is the
continuous image of a compact set, so $Z_A=\sup_{T_A}\inner gt$ is the supremum
of a centred Gaussian process with $\sigma^{2}\le4e^{2\rho}$; Fact~\ref{fact:dudley}
and Lemma~\ref{lem:dudT} give $\E Z_A\le C_{\mathrm D}K_{\mathrm G}(21W_k+8W_v)$,
and Fact~\ref{fact:btis} with $r=2e^{\rho}\sqrt{2\log5}<3.6e^{\rho}$ gives
$\Prob(Z_A>\theta_A)\le\tfrac15$; likewise for $B$. $K$ need not be bounded,
so intersect with the ball: $K_R$ is compact, symmetric, convex, and contains
the ball of radius $\min_X\theta_X/(2e^{\rho})$ since $|\inner yt|\le2e^{\rho}\norm y$
for every test vector, so it is a convex body; and
$\Prob(\norm g>R)\le N/R^{2}=0.01$.
\end{proof}

\begin{proof}[Proof of Lemma~\ref{lem:onelevel}]
Take $M=\lceil\max\{2e\rho,\log_2(2n)\}\rceil$. Fact~\ref{fact:ban} in $E$
with the vectors~\eqref{eq:data} and the body $K_R$ gives $\sigma$ with
$\sup_{T_X}\inner{y_\sigma}t\le5\theta_X$. By~\eqref{eq:readoff} and symmetry,
$|\sum\sigma_ie_M(q^{\!\top}k_i)\inner{v_i}u|\le5\theta_A$ for all $q,u$, so the
Euclidean norm of $\sum\sigma_ie_M(q^{\!\top}k_i)v_i$ is at most $5\theta_A$;
$|\sum\sigma_ie_M(q^{\!\top}k_i)|\le5\theta_B$; $|\sum\sigma_i|\le5\theta_N=10\sqrt5$.
By Lemma~\ref{lem:trunc}(ii) and the choice of $M$, replacing $e_M$ by $e^{(\cdot)}$
changes each kernel sum by at most $m\cdot2^{-M}\le\tfrac12\le e^{\rho}$. So
$D_A=5\theta_A+e^{\rho}$, $D_B=5\theta_B+e^{\rho}$, $D_N=10\sqrt5$ work, and
with $C_1'=105C_{\mathrm D}K_{\mathrm G}+23$ the stated bounds follow
from~\eqref{eq:theta}.
\end{proof}

\begin{proof}[Proof of the first branch of Theorem~\ref{thm:main}]
Set $D=D_A+D_B$ from Lemma~\ref{lem:onelevel}. If $2D>\eps n$ output $[n]$.
Otherwise let $\ell$ be maximal with $2^{\ell+1}D\le\eps n$, and obtain
$P_{j+1}$ from $P_j$ by Lemma~\ref{lem:onelevel} and Lemma~\ref{lem:step},
$P_0=[n]$, $n_j=|P_j|$. Sizes: $n/2^{j}\le n_j\le n/2^{j}+D_N$. Denominators:
$B_{P_0}\ge n$ by Jensen, and $B_{P_j}\ge n/2^{j}-D_B\sum_{i=1}^{j}2^{-i}\ge n/2^{j}-D_B$;
for $j\le\ell$, $n/2^{j}\ge2D/\eps\ge2D_B$, so $B_{P_j}\ge n/2^{j+1}$. Error:
Lemma~\ref{lem:step}(iii) at step $j$ has denominator $B_{P_{j+1}}\ge n/2^{j+2}$,
so $\norm{\Attn_{P_\ell}-\Attn}\le\sum_{j<\ell}2^{j+1}D/n\le2^{\ell+1}D/n\le\eps$
uniformly in $q$. Size: maximality gives $n/2^{\ell}<4D/\eps$, so
$n_\ell<4D/\eps+D_N$, and $4D\le8C_1'e^{\rho}(\sqrt{\dv}+\sqrt{\dk\log(1+\rho)})$.
\end{proof}

% =====================================================================
\section{Proofs for Section~\ref{sec:alg}}
\label{app:alg}

Fix $P$, $|P|=m$, and the untruncated vectors
$x_i=\tfrac12(\phi_0(k_i)\otimes v_i,\phi_0(k_i),k_i,1)$ in
$\mathcal H=(\mathcal H_{\dk}\otimes\R^{\dv})\oplus\mathcal H_{\dk}\oplus\R^{\dk}\oplus\R$,
$\norm{x_i}\le1$, with test vectors $t_A(q,u)=2(\Phi(q)\otimes u,0,0,0)$,
$t_B(q)=2(0,\Phi(q),0,0)$, $t_C(w)=2(0,0,w,0)$ for unit $w$, $t_N=2(0,0,0,1)$;
the read-offs~\eqref{eq:readoff} hold with $e^{(\cdot)}$, and
$\inner{y_\sigma}{t_C(w)}=\inner{\sum\sigma_ik_i}w$.

\begin{lemma}[Isometric coordinates]
\label{lem:gram}
$G_{ij}=\inner{x_i}{x_j}$ is the matrix of Definition~\ref{def:level}, positive
semidefinite; if $G=RR^{\!\top}$ then the rows $r_i$ of $R$ satisfy
$\inner{r_i}{r_j}=G_{ij}$, there is a linear isometry $\iota$ of
$V=\operatorname{span}\{x_i\}$ onto $\operatorname{span}\{r_i\}$ with $\iota x_i=r_i$,
and for every $t\in\mathcal H$,
$\inner{\sum\sigma_ix_i}t=\inner{\sum\sigma_ir_i}{\iota\Pi t}$ with $\Pi$ the
projection onto $V$ and $\norm{\iota\Pi t}\le\norm t$.
\end{lemma}

\begin{proof}
Two families with the same Gram matrix are related by the isometry
$\sum a_ix_i\mapsto\sum a_ir_i$, well defined since
$\norm{\sum a_ix_i}^{2}=a^{\!\top}Ga=\norm{\sum a_ir_i}^{2}$; and
$\inner yt=\inner y{\Pi t}$ for $y\in V$.
\end{proof}

We work in the real-RAM model; rows accurate to $\eta\le1/(2m)$ change every
read-off by at most $2m\eta e^{\rho}\le e^{\rho}$, absorbed by thresholds
exceeding $100e^{\rho}$, and multiply the subgaussian constant by at most $1+\eta$.

\begin{lemma}[Subgaussian process on the test sets]
\label{lem:process}
Let $\sigma$ be the walk's output on $r_1,\dots,r_m$ and $X_t=\inner{\sum\sigma_ix_i}t$.
Then $\E X_t=0$, $\norm{X_t-X_s}_{\psi_2}\le K_{\ast}\norm{t-s}$ with
$K_{\ast}=C_\psi\sigma_{\mathrm{GS}}$, and for $X\in\{A,B\}$, $Z_X=\sup_{T_X}X_t$
satisfies $\E Z_X\le C_{\mathrm D}K_{\ast}\Dud(T_X)$ and
$\Prob(Z_X>\tfrac12C_{\mathrm D}K_{\ast}[\Dud(T_X)+4ue^{\rho}])\le2e^{-u^{2}}$;
moreover $\E(\sum\sigma_i)^{2}\le4\sigma_{\mathrm{GS}}^{2}$ and
$\E\norm{\sum\sigma_ik_i}^{2}\le4\sigma_{\mathrm{GS}}^{2}\dk$.
\end{lemma}

\begin{proof}
$X_t=\inner y{\iota\Pi t}$ with $y=\sum\sigma_ir_i$ subgaussian by
Fact~\ref{fact:gsw}, so the increment bound follows with
$\norm{\Pi(t-s)}\le\norm{t-s}$. Fact~\ref{fact:dudley} on $(T_X,\norm\cdot)$
gives the expectation bound; for the tail, $T_X=-T_X$ and $X_{-t}=-X_t$ give
$Z_X\le\tfrac12\sup_{t,s}|X_t-X_s|$ with $\diam(T_X)\le4e^{\rho}$
(Lemma~\ref{lem:dudT}, untruncated sets). The last two claims are
$\E X_t^{2}\le\sigma_{\mathrm{GS}}^{2}\norm{\Pi t}^{2}$ at $t=t_N$ and at
$t=t_C(e_l)$, summed over $l\le\dk$.
\end{proof}

\begin{proof}[Proof of Lemma~\ref{lem:accept}]
By Markov, each check fails with probability $\le\tfrac14$, so acceptance has
probability $\ge\tfrac12$; $Z_X\ge0$ by symmetry, and dividing by
$\Prob(\mathrm{acc})$ gives the conditional statements.
\end{proof}

\begin{lemma}[Deterministic invariants]
\label{lem:detinv}
Let $D_N'=4\sigma_{\mathrm{GS}}$, $D_C'=4\sigma_{\mathrm{GS}}\sqrt{\dk}$. For all $j$,
$n/2^{j}\le n_j\le n/2^{j}+D_N'$ and $c_j:=\norm{\sum_{P_j}k_i}\le D_C'$; if
$n_j\ge\rho D_C'/\log2$ then $B_{P_j}(q)\ge n_j/2$ for all $\norm q\le\rho$.
\end{lemma}

\begin{proof}
Sizes are Lemma~\ref{lem:step}(i). $c_{j+1}\le\tfrac12(c_j+D_C')\le D_C'$ from
$c_0=0$. Jensen: $B_{P_j}\ge n_je^{q^{\!\top}\bar k_{P_j}}\ge n_je^{-\rho c_j/n_j}\ge n_j/2$.
\end{proof}

\begin{lemma}[Error per level]
\label{lem:errc}
If $n_{j+1}\ge\rho D_C'/\log2$ then
$\norm{\Attn_{P_{j+1}}(q)-\Attn_{P_j}(q)}\le2(a_j+b_j)/n_j$ for all $\norm q\le\rho$,
where $a_j,b_j$ are the realised values of $Z_A,Z_B$ at level $j$.
\end{lemma}

\begin{proof}
Lemma~\ref{lem:step}(iii) with $(a_j,b_j)$, and $B_{P_{j+1}}\ge n_{j+1}/2\ge n_j/4$.
\end{proof}

\begin{proof}[Proof of Theorem~\ref{thm:main-alg}(a), details]
Let $\bar D=2\bar C_{\mathrm D}\bar K_{\ast}(42W_k+8W_v+4e^{\rho})$ with explicit
numerical majorants $\bar C_{\mathrm D},\bar K_{\ast}$ of the absolute
constants, and $\ell$ maximal with $2^{\ell+2}\bar D\le\eps n$ ($\ell=0$ if
none). Running time: level $j$ costs $O(m^{2}(\dk+\dv))+O(m^{3})$ plus an
expected two walks at $O(m(2m)^{\omega})$, with $m\le n/2^{j}+D_N'$, summing
to $O(n^{\omega+1}+n^{2}(\dk+\dv))$. Size: $n_\ell\le n/2^{\ell}+D_N'<8\bar D/\eps+D_N'$.
Validity: for $j<\ell$, $n_{j+1}\ge n/2^{\ell}\ge4\bar D/\eps\ge4\bar D$, and
$\bar D\ge42\sqrt{40}\,e^{\rho}\sqrt{\dk\log(1+\rho)}\ge42\sqrt{40}\sqrt{\log2}\,\rho\sqrt{\dk}>221\rho\sqrt{\dk}$
using $\inf_{\rho>0}e^{\rho}\sqrt{\log(1+\rho)}/\rho\ge\sqrt{\log2}$, while
$\rho D_C'/\log2<36.5\rho\sqrt{\dk}$. Error: the accepted $\sigma^{(j)}$ is a
function of $P_j$ and fresh randomness with the law of one trial conditioned
on acceptance, so $\E[a_j+b_j\mid\text{history}]\le\bar D$ by
Lemma~\ref{lem:accept}, and by Lemma~\ref{lem:errc}, $n_j\ge n/2^{j}$ and the
tower property, $\E\sup_q\norm{\Attn_{P_\ell}-\Attn}\le\sum_{j<\ell}2^{j+1}\bar D/n\le\eps/2$.
\end{proof}

\begin{proof}[Proof of Theorem~\ref{thm:main-alg}(b), details]
With $u=\sqrt{\log(16L/\delta)}$, Lemma~\ref{lem:accept} gives at each level
$\Prob(a_j>\tfrac12C_{\mathrm D}K_{\ast}[\Dud(T_A)+4ue^{\rho}]\mid\text{history})\le\delta/(4L)$,
likewise for $b_j$; a union bound over $\le2L$ events leaves probability
$\ge1-\delta/2$ on which $a_j+b_j\le D^{\ast}$ for all $j<\ell$, with
$D^{\ast}=\tfrac12\bar C_{\mathrm D}\bar K_{\ast}(42W_k+8W_v+4e^{\rho}+8ue^{\rho})$,
with the majorants of~(a) so that the stopping rule is computable (the tail
bound of Lemma~\ref{lem:accept} then applies a fortiori). On that event
the error is $\le2^{\ell+1}D^{\ast}/n\le\eps$ deterministically with the
stopping rule $2^{\ell+1}D^{\ast}\le\eps n$, validity holds as before since
$D^{\ast}\ge\bar D/4$, and the size is $<4D^{\ast}/\eps+D_N'$.
\end{proof}

% =====================================================================
\section{Proof of Theorem~\ref{thm:main-lb}(a)}
\label{app:lb}

Let $L=e^{\rho}$, $\rho\ge2$, $\dk\ge2$, $\dv\ge1$, $\eps\le\eps_0/\sqrt{\dv}$ with
$\eps_0:=1/(40\sqrt{10})$, and put $c_1:=1/(40\sqrt{10}\,e)$. Define
\begin{equation}
\label{eq:m}
  m:=\max\Bigl\{1,\Bigl\lfloor\min\Bigl\{\frac{c_1L}{\eps\sqrt{\dv}},\;\frac{L^{2}}{1000e^{2}},\;\frac{e^{\dk/(4\rho^{2})}}{2}\Bigr\}\Bigr\rfloor\Bigr\},
\end{equation}
so that $m\ge\tfrac12\min\{\cdots\}$ in all cases.

\paragraph{Code.} Draw $u_1,\dots,u_m$ i.i.d.\ uniform on $S^{\dk-1}$. By
Fact~\ref{fact:cap} with $t=1/\rho$ and a union bound over $\binom m2$ pairs,
$\Prob(\exists i\ne h:|\inner{u_i}{u_h}|>1/\rho)\le m^{2}e^{-\dk/(2\rho^{2})}<1$
when $m<e^{\dk/(4\rho^{2})}$, which~\eqref{eq:m} ensures; fix such a code
(for $m=1$ there is nothing to check).

\paragraph{Instance.} Let $x=(x_{ij})\in\{0,1\}^{m\times\dv}$ and let
$b=(b_{ij})\in\{\pm1\}^{m\times\dv}$ be public signs. The instance has
$N':=2m\dv$ pairs: $(u_i,b_{ij}x_{ij}e_j)$ and $(-u_i,0)$ for $i\in[m]$,
$j\in[\dv]$. The key centroid is zero, all keys have norm one, all values
norm at most one, so~\eqref{eq:norm} holds with the given $\rho$.

\paragraph{Decoding.} For the query $q_i=\rho u_i$ put $w_{ih}=e^{\rho\inner{u_i}{u_h}}$;
then $w_{ii}=L$, and $e^{-1}\le w_{ih}\le e$ for $h\ne i$. The weight of the key
$-u_h$ is $w_{ih}^{-1}$. Hence
\[
  B_i:=B(q_i)=\dv\sum_{h}(w_{ih}+w_{ih}^{-1})\le\dv(2L+2em),
\]
a quantity depending only on the public keys. The $j$th numerator coordinate
is $A(q_i)_j=\sum_hw_{ih}b_{hj}x_{hj}$, so
$b_{ij}A(q_i)_j=Lx_{ij}+N_{ij}$ with $N_{ij}=\sum_{h\ne i}w_{ih}b_{ij}b_{hj}x_{hj}$.
Conditional on $x$ and the code, $\E_bN_{ij}=0$ and
$\E_bN_{ij}^{2}=\sum_{h\ne i}w_{ih}^{2}x_{hj}\le e^{2}m$, so by Chebyshev
$\Prob_b(|N_{ij}|>L/10)\le100e^{2}m/L^{2}\le\tfrac1{10}$ by~\eqref{eq:m}
(vacuous for $m=1$, where $N_{ij}=0$).

Let $S$ be an $\eps$-coreset. Bob, holding $(i,j)$, $S$, the data of $S$, the
public keys and $b$, computes $\hat y:=b_{ij}B_i\Attn_S(q_i)_j$. Since
$\sum_{j'}(\Attn_S(q_i)_{j'}-\Attn(q_i)_{j'})^{2}\le\eps^{2}$, for $j$ uniform in
$[\dv]$ Markov's inequality gives $|\Attn_S(q_i)_j-\Attn(q_i)_j|\le\eps\sqrt{10/\dv}$
with probability $\ge\tfrac9{10}$. On that event
$|\hat y-(Lx_{ij}+N_{ij})|\le B_i\eps\sqrt{10/\dv}\le2\eps\sqrt{10\dv}(L+em)\le L/20+L/20$,
using $\eps\sqrt{\dv}\le\eps_0$ for the first term and $m\le c_1L/(\eps\sqrt{\dv})$
for the second (for $m=1$ the latter holds since
$c_1L/(\eps\sqrt{\dv})\ge c_1L/\eps_0=e^{\rho-1}\ge e$). Together with $|N_{ij}|\le L/10$, Bob's threshold
$\hat x_{ij}:=\mathbf 1[\hat y>L/2]$ is correct. The two failure events have
total probability $\le\tfrac15$ by a union bound, so Bob succeeds with
probability $\ge\tfrac45$ over $x$, $b$ and $(i,j)$.

\paragraph{Counting.} Suppose that for every $x$ and every $b$ the instance
admitted an $\eps$-coreset of size at most $s$. Alice sends the positions of
$S$ among the $N'$ pairs and the bits $x_{ij}$ of the selected pairs with key
$+u_i$: at most $\log_2\binom{N'}{s}+s+O(\log N')\le s\log_2(eN'/s)+s+O(\log N')$
bits, and $N=m\dv=N'/2$ INDEX bits are being decoded. With $s=0.03N$ this is
$\le0.03N\log_2(2e/0.03)+0.03N+O(\log N)\le0.256N+O(\log N)$ bits, whereas Lemma~\ref{lem:index} with $\delta=\tfrac15$
requires at least $(1-h(\tfrac15))N>0.278N$ bits. For $N\ge N_0$ (an absolute
constant) this is a contradiction, so some $(x,b)$ yields an instance every
$\eps$-coreset of which has $|S|>0.03m\dv\ge0.015\dv\min\{\cdots\}$. For
$N<N_0$: since $c_1,\tfrac12>1/(1000e^{2})$, each term of the minimum in the
theorem is at most $1000e^{2}$ times the corresponding term in~\eqref{eq:m},
so $\dv\min\{L/(\eps\sqrt{\dv}),L^{2},e^{\dk/(4\rho^{2})}\}\le1000e^{2}\dv\cdot2m=2000e^{2}N$,
and the claimed bound is at most $2000e^{2}c\,N<2000e^{2}c\,N_0$, which is
below $1\le|S|$ for $c\le1/(2000e^{2}N_0)$. The same comparison of minima
gives, for $N\ge N_0$,
$|S|>0.015\dv\min\{\cdots\}\ge(0.015/(1000e^{2}))\,\dv\min\{L/(\eps\sqrt{\dv}),L^{2},e^{\dk/(4\rho^{2})}\}$,
so the first display holds with $c=\min\{0.015,\,1/(2N_0)\}/(1000e^{2})$ and $C=4$. The regime statement selects the first branch of the minimum:
$L/(\eps\sqrt{\dv})\lesssim L^{2}$ iff $\eps\sqrt{\dv}\gtrsim e^{-\rho}$, and
$L/(\eps\sqrt{\dv})\lesssim e^{\dk/(4\rho^{2})}$ iff
$\dk\gtrsim\rho^{2}(\rho+\log\frac1{\eps\sqrt{\dv}})$.
\qed

\paragraph{The counting step of~\cite[Thm.~7]{LAK26}.}
As written, it encodes the indices of $o(md)$ retained pairs, which cost
$\Theta(\log(md))$ bits each and so yield only $\Omega(md/\log(md))$; the
constant-fraction count above removes the logarithm.

\paragraph{Why the argument gives no key-dimension factor.}
The $\sqrt{\dv}$ comes from multiplexing independent bits over output
coordinates; the same one-bit-per-token architecture yields no key-dimension
factor, because at any single query the signal available from one token is at
most $2e^{\rho}/n$ against a required resolution $\eps/\sqrt{\dv}$, which caps
the argument at $e^{\rho}\sqrt{\dv}/\eps$. Breaking this requires side
information, as in~\cite{CFIKKP26}.

% =====================================================================
% Experimental appendices: details and additional tables and figures.
% Macros come from results_values.tex.
% =====================================================================
\section{Experimental details}
\label{app:expdetails}

\subsection{Corpus and capture}
\label{app:corpus}

The corpus is built once. Natural documents are
drawn from nine LongBench~\cite{Bai24} subsets (2wikimqa, gov\_report,
hotpotqa, multifieldqa\_en, musique, narrativeqa, qasper, qmsum, triviaqa) and
from the QASPER test split~\cite{Dasigi21}. Each document is tokenised with
the Qwen2.5 tokeniser and filed under the longest of the four lengths 4096,
8192, 16384 and 32768 tokens that it fills to at least nine tenths, fifty
natural documents per length; a document is assigned by the length it fills
rather than by its position in a file, and most of the documents in the
32768-token bucket are from narrativeqa. Fifty prompts of uniformly random tokens are added as a
control and the whole is interleaved so that any prefix of the corpus holds
every length and the controls in proportion. A model whose own tokeniser
counts fewer tokens than the 4096-token floor, or fewer than nine tenths of
the bucket's length, skips the prompt and records why: Qwen skips
\vEaQSkipped\ of the 250 prompts, the documents of the 4096-token bucket
that fall between nine tenths of it and the floor, and Llama \vEaLSkipped,
of which \vEaLSkippedShort\ fall short of the 16384-token bucket; they
measure \vEaQPromptsMeasured\ and \vEaLPromptsMeasured.

The census runs the prefill with the scaled-dot-product attention backend,
then generates 64 tokens greedily. Forward hooks read the
projected keys, values and queries of every layer; the rotary embedding is
applied to keys and queries at their positions and the model's scaling
$1/\sqrt{\dk}$ is folded into the queries, so that the captured
$q^{\!\top}k$ equals the logit the model computes, which the tests check
against the model's own attention output. Neither model uses sliding-window
attention at these lengths, so no layer is excluded on that ground. Nine
heads per model, stratified by terciles of the census's per-head $\rho$, are
dumped at the six leading prompts of the corpus, taken without the census's
skip rule (on Llama one of them, at \vEbLLengthsMin\ tokens, is below the
floor), each with 128 decoding queries, for E2 and the kernel census; E3
uses the same heads and documents with its own prefill
(Appendix~\ref{app:e3}).

\subsection{Definitions}
\label{app:defs}

For one head, $\kappa=\max_i\norm{k_i-\bar k}$ over the cache after the
prefill. $\vrho$ is computed over the decoding queries of the query heads
that share the key--value head, in the maximum convention (the supremum
of~\eqref{eq:coreset}) and in the 99th-percentile convention; the percentile
is left undefined when fewer than 100 decoding queries are captured, since
below that count a percentile interpolates between the two largest values.
$\rho=\vrho\kappa$. Because
$\vrho$ is a maximum, it depends on the query set: E1 takes the 64 decoding
queries of the census per query head, the E2 dumps carry 128, and E3 takes
the 256 queries of the reference generation, so the same head need not
carry the same $\rho$ in E1 and E2, since the prompt sets and the number of
queries differ. The uncentred control is $\max_t\norm{q_t}\max_i\norm{k_i}$ over the
same queries. The participation ratio of a query is $1/\sum_ip_i^{2}$ for
its softmax weights $p$, the effective number of keys it attends to; an
instance is near hard max when the mean participation ratio over its
decoding queries is at most 16. Errors are $\norm{\Attn_S(q)-\Attn(q)}_2$
divided by $\max_i\norm{v_i}$, computed in double precision from the
captured tensors, which are stored in half precision. A head is a (layer, key--value head) pair, written
L07H01; an instance is one head at one prompt; a cell is one (instance,
method, size).

\subsection{Methods and provenance}
\label{app:methods}

All methods receive the normalised instance of~\eqref{eq:norm}: keys centred
at their mean and divided by $\kappa$, queries multiplied by $\kappa$, values
divided by $\max_i\norm{v_i}$.

\emph{Uniform sampling} draws the target size without replacement.

\emph{Halving} (Theorem~\ref{thm:main-alg}(a)) forms the Gram matrix of
Definition~\ref{def:level} with all four blocks, factors it by a matrix-free
pivoted Cholesky decomposition with a relative stopping rule and a rank cap
of \vKnQCap, runs the Gram--Schmidt walk~\cite{BDGL18} on the rows of the
factor, checks the cardinality and key-sum constraints and redraws on failure
(at most four trials; on both E2 runs every level was accepted at its first
trial), keeps the larger sign class and
recurses; every intermediate level is a coreset, so all sizes
come from one recursion. The walk runs in double precision on the CPU; its
step solves a $k\times k$ system through the maintained Gram matrix, with the
residual of the cheap Sherman--Morrison update checked and a
ridge-regularised Cholesky solve substituted where it fails, which is what
keeps the walk correct at large $\rho$ where the Gram matrix is nearly
diagonal. The rank cap truncates the features the walk chooses a signing
from; every error reported is measured on the actual instance, so
truncation affects which signing is chosen and not how its error is
measured. The largest relative residual of any factorisation in the census
of Table~\ref{tab:kernel} was \vKnQResid\ on Qwen and \vKnLResid\ on Llama.

\emph{Halving, re-centred subsample} (Corollary~\ref{cor:sample-balance})
draws $s_0=\min\{n,\lceil2048e^{2\rho}/\eps^{2}\rceil\}$ pairs, redraws, at
most four times, until the sampled key mean is at most $2s_0^{-1/2}$ in norm,
re-centres and rescales the sample, and applies the halving above. On every
instance $s_0=n$ (see
Section~\ref{sec:exp-setup}), so the sampled set is the capped input.

\emph{BalanceKV}~\cite{KSHZK25} is a reimplementation: the self-balancing walk
of~\cite{ALS21} with the parameter $c=30\log(2m/\delta)$, $\delta=0.1$, inside
halving on the pivoted-Cholesky features of the numerator and denominator
blocks alone, without rejection and without the streaming merge structure of
the original; no public code was located. On \vKnQDegenHi\ of the
\vKnQDumpsHi\ Qwen instances above $\rho=\vKnQCollapse$ the census finds the
two blocks it uses diagonal to its $10^{-8}$ threshold
(Table~\ref{tab:kernel}), and the run's own per-level check flagged
degeneracy or an underflowed diagonal on every one of the twelve, so its
signing reduces to random signs there, which is a property of the method at
that $\rho$ and not of the implementation.

\emph{Kernel halving}~\cite{DM21,CGSDM25} is a reimplementation of the
pairwise halving of kernel thinning with the self-balancing probabilities,
on the same two blocks, without the refinement step and without the queries
that the thinning of~\cite{CGSDM25} uses; the published guarantee of that
method does not apply to this row. It runs on the whole cache, uncapped.

\emph{Heavy hitters} keeps the keys with the largest mean softmax mass under
the instance's decoding queries; it is fitted and evaluated on the same
queries, which is the most favourable reading of a SnapKV-style
oracle~\cite{Li24}. \emph{Per-query} runs the Beck--Fiala halving of
Corollary~\ref{cor:perquery} for the first decoding query, with vector
values projected on their first principal direction, and is evaluated at the
worst over all queries.

The set handed to the Gram--Schmidt or self-balancing walk is a uniform
subsample of 8192 pairs when the cache is larger; in E2 the cap binds on the
two longest of each head's six prompts, and in E4 on every instance. Sizes
are matched at the requested size, and a halving method returns the level
nearest it: the construction's achieved size exceeds the request by up to
\vEbQOursOversizeMax\ times on Qwen and \vEbLOursOversizeMax\ on Llama (the
latter on the instances whose cache length is not a power of two times the
request), and BalanceKV's by up to \vEbQBkvOversizeMax\ and
\vEbLBkvOversizeMax\ at the smallest sizes, since its levels are not
rejected to the cardinality check. The two variants of the construction
share a rank cap of \vKnQCap, a relative Cholesky tolerance of $10^{-12}$,
four trials per level and the 8192-pair cap; BalanceKV shares the rank cap,
the tolerance and the cap but takes one trial per level; kernel halving is
uncapped.

\subsection{Closed-loop protocol}
\label{app:e3}

For each prompt the model first decodes \vEcQSteps\ tokens from the
uncompressed cache; the tokens and the per-step queries of the nine selected
heads are the reference. A coreset is built per selected head from the
prefill cache with the walk settings of E2, at the requested sizes $n/16$
and $n/4$; the heavy-hitter oracle is fitted on the queries of the reference
generation. The open loop feeds the reference tokens back at every step
through the model with the selected heads masked to their coresets, so its
queries are those the masked model computes under the reference tokens, and
records, for each masked head at each step, the E2 error of the coreset's
attention against the full prefill cache's attention on the queries that
loop issued, taking the largest over the query heads of the group; the
closed loop decodes freely
with the same masks, so its tokens and queries depend on the compressed
cache, and records the same error on the queries it issued. Masking a position sets its logit to a large
negative value, which yields exactly the attention a shrunken cache would
produce. Token agreement is the
fraction of the \vEcQSteps\ generated tokens that equal the reference token
at that position, and the tracked prefix is the number of leading steps on
which they agree. Rows are (prompt, head, method, size); the per-step
trajectories of every row are recorded with the first divergent step marked.

\subsection{Synthetic instances}
\label{app:e4}

Keys are drawn uniformly on the unit sphere of $\R^{\dk}$, or from eight
clusters with centres uniform on the sphere and isotropic Gaussian noise of
standard deviation $0.15$ per coordinate, then centred and rescaled to unit
maximum norm. Values are uniform on the unit sphere of $\R^{\dv}$ for
the sphere family and, for the clustered family, a cluster value plus noise
of the same scale, rescaled to unit maximum norm; queries are 2000 points of
norm $\rho$. An instance is seeded from its seed alone and its query set from
(seed, $\dk$, $\dv$), so the sweep is independent of execution order and was
run as six shards, one per ($\dk$, key family), merged with a check that the
shards cover the declared grid exactly once. For each (family, $\dk$, $\dv$,
$\rho$, method, seed) the worst error is measured at the thirteen requested
sizes from $n/16384$ to $n/4$; the size reaching $\eps=0.1$ is the crossing
of a least-squares fit of $\log(\text{worst})$ on $\log(\text{requested
size})$ over those sizes, less the rows at $n/8$ and $n/4$ for the three
capped methods, which hold the 8192-pair cap subsample itself
(Section~\ref{sec:exp-setup}); the slope of $\log(\text{size})$ against
$\rho$ is fitted over $\rho\in\{1,2,4\}$ per seed, the reported slope being
the median over seeds and the bracketed range its extremes. A crossing at or
below the smallest fitted size, or above the largest, is off the fitted
sizes; a fit is excluded when at least half of its crossings are, or when
its fitted sizes vary by under five percent, and marked $\dagger$ when any
is. The ratio table of Section~\ref{sec:e4} divides the
construction's crossing by uniform sampling's within each seed, keeps a seed
only when both crossings lie inside the sizes their fits used, and reports a
cell only when at least half of its seeds are kept. The separation instance
is that of the proof of
Proposition~\ref{prop:sep}(ii); its witness query and value are computed in
closed form from the signing.

\subsection{Hardware, run times and seeds}
\label{app:hardware}

E1 and E3 ran on one RTX 5090 with 32\,GB, E2 and the kernel census on one
RTX 3090 with 24\,GB, and the six shards of E4 on the two cards, three
each; every walk ran on the CPU of the respective host.
Table~\ref{tab:runs} gives the wall-clock time of each run. The base seed is
$0$; E2 uses seeds $0$--$9$, E4 ten seeds per instance and the separation
eight.

\begin{table}[hbp]
\centering\small
\begin{tabular}{@{}llr@{}}
\toprule
Run & GPU & hours \\
\midrule
E1, Qwen & \vRunEaQGpu & \vRunEaQHours \\
E1, Llama & \vRunEaLGpu & \vRunEaLHours \\
E2, Qwen & \vRunEbQGpu & \vRunEbQHours \\
E2, Llama & \vRunEbLGpu & \vRunEbLHours \\
E3, Qwen & \vRunEcQGpu & \vRunEcQHours \\
E3, Llama & \vRunEcLGpu & \vRunEcLHours \\
E4 (six shards, summed) & \vRunEdGpu & \vRunEdHours \\
\bottomrule
\end{tabular}
\caption{Wall-clock time of the runs behind the tables.}
\label{tab:runs}
\end{table}

% =====================================================================
\section{Additional tables and figures}
\label{app:exptables}
\suppressfloats[t]

The tables and the figure below are referred to from Section~\ref{sec:exp}
and Appendix~\ref{app:expdetails}.

\begin{table}[htp]
\centering\small
\begin{tabular}{@{}lrrrrr@{}}
\toprule
& \multicolumn{3}{c}{$\vrho$ within one query-head group} & \multicolumn{2}{c}{$\kappa$ after trimming}\\
\cmidrule(lr){2-4}\cmidrule(lr){5-6}
Model & largest head & smallest head & ratio & drop $4$ & drop $32$ \\
\midrule
Qwen2.5-7B-Instruct & \vEaQgqaHi & \vEaQgqaLo & \vEaQgqaRatio & \vEaQtrimFour & \vEaQtrimThirtytwo \\
Llama-3.1-8B-Instruct & \vEaLgqaHi & \vEaLgqaLo & \vEaLgqaRatio & \vEaLtrimFour & \vEaLtrimThirtytwo \\
\bottomrule
\end{tabular}
\caption{E1. Left: $\vrho$ decomposed over the query heads that share one
key--value head; each column is a median over heads, so the ratio is the
median of per-head ratios. Right: $\kappa$ after removing the four and the
thirty-two keys furthest from the centroid, as a fraction of the untrimmed
value, median over heads.}
\label{tab:e1gqa}
\end{table}

\begin{table}[htp]
\centering\small
\begin{tabular}{@{}lrrrrrr@{}}
\toprule
Model & context $n$ & heads & prompts & median $\kappa$ & median $\vrho$ & median $\rho$ \\
\midrule
\vScaleBody
\bottomrule
\end{tabular}
\caption{E1. Centred $\rho$ and its two factors against context length. Each
row is a separate census; rows are bucketed by the requested length, and the
realised token count falls short of it by at most \vScQFillFrac\ on Qwen and
\vScLFillFrac\ on Llama: the corpus admits a document that fills nine tenths
of its bucket, and the Llama tokeniser counts fewer tokens than the Qwen
tokeniser the corpus was bucketed with. The documents differ between
lengths; the fitted exponents are in Section~\ref{sec:e1}.}
\label{tab:scale}
\end{table}

\begin{table}[htp]
\centering\small
\begin{tabular}{@{}lrrrr@{}}
\toprule
& \multicolumn{2}{c}{Qwen2.5-7B-Instruct} & \multicolumn{2}{c}{Llama-3.1-8B-Instruct}\\
\cmidrule(lr){2-3}\cmidrule(lr){4-5}
Method & size ratio & fewer keys & size ratio & fewer keys \\
\midrule
BalanceKV & \vMeQBkvRatio & \vMeQBkvFrac & \vMeLBkvRatio & \vMeLBkvFrac \\
Kernel halving & \vMeQThinRatio & \vMeQThinFrac & \vMeLThinRatio & \vMeLThinFrac \\
\textbf{Halving} & \vMeQOhRatio & \vMeQOhFrac & \vMeLOhRatio & \vMeLOhFrac \\
\textbf{Halving, re-centred subsample} & \vMeQOsbRatio & \vMeQOsbFrac & \vMeLOsbRatio & \vMeLOsbFrac \\
\addlinespace
Heavy hitters (oracle) & \vMeQHhRatio & \vMeQHhFrac & \vMeLHhRatio & \vMeLHhFrac \\
Per-query & \vMeQPqRatio & \vMeQPqFrac & \vMeLPqRatio & \vMeLPqFrac \\
\bottomrule
\end{tabular}
\caption{E2 read at matched error. For each instance and each swept size $T$,
the smallest swept size at which the method reaches uniform sampling's error
at $T$ on the real decoding queries, divided by $T$: ``size ratio'' is the
median over (instance, $T$) and ``fewer keys'' the fraction of those cells in
which the method needs strictly fewer keys. Cells at an instance's smallest
swept size are excluded, since no smaller size exists, and cells in which the
method never reaches uniform sampling's error are dropped. The size grid is
geometric, so the ratio is quantised to powers of two and a median of exactly
one means the same grid point. Population: \vMeQDumps\ instances of
\vMeQHeads\ heads on Qwen and \vMeLDumps\ of \vMeLHeads\ on Llama, the
near-hard-max instances excluded as in Table~\ref{tab:e2}.}
\label{tab:matched}
\end{table}

\begin{table}[htp]
\centering\footnotesize
\setlength{\tabcolsep}{4pt}
\begin{tabular}{@{}llrrrcrrr@{}}
\toprule
& & & & \multicolumn{2}{c}{pivoted Cholesky} & \multicolumn{2}{c}{off-diag.\ correlation} & \\
\cmidrule(lr){5-6}\cmidrule(lr){7-8}
Model & $\rho$ band & inst. & heads & rank & saturated & full & $A,B$ & degen. \\
\midrule
\vKernelBody
\bottomrule
\end{tabular}
\caption{The Gram matrix of Definition~\ref{def:level} on the E2 instances,
by band of the instance's $\rho$. ``Rank'' is the rank at which the
matrix-free pivoted Cholesky met its relative stopping rule, capped at
\vKnQCap; ``saturated'' counts the instances that reached the cap, for which
the recorded rank is only a lower bound on the true rank. The correlation columns give the largest
off-diagonal entry of the correlation matrix over a sample of pairs, for the
whole Gram matrix and for the numerator and denominator blocks alone; the
cardinality block adds a constant to every entry of the former, so only the
latter moves with $\rho$. ``Degenerate'' counts instances whose $A,B$ blocks
have no off-diagonal correlation above $10^{-8}$. Each band lists instances and distinct
heads, since one head contributes up to six instances.}
\label{tab:kernel}
\end{table}

\begin{table}[htp]
\centering\small
\begin{tabular}{@{}lrrrr@{}}
\toprule
& \multicolumn{2}{c}{Qwen2.5-7B-Instruct} & \multicolumn{2}{c}{Llama-3.1-8B-Instruct}\\
\cmidrule(lr){2-3}\cmidrule(lr){4-5}
Method & worst case & after divergence & worst case & after divergence \\
\midrule
Uniform sampling & \vEcQUnifRatio & \vEcQUnifPostDiv & \vEcLUnifRatio & \vEcLUnifPostDiv \\
BalanceKV & \vEcQBkvRatio & \vEcQBkvPostDiv & \vEcLBkvRatio & \vEcLBkvPostDiv \\
\textbf{Halving, re-centred subsample} & \vEcQOsbRatio & \vEcQOsbPostDiv & \vEcLOsbRatio & \vEcLOsbPostDiv \\
\addlinespace
Heavy hitters (oracle) & \vEcQHhRatio & \vEcQHhPostDiv & \vEcLHhRatio & \vEcLHhPostDiv \\
\bottomrule
\end{tabular}
\caption{E3, pooled over both sizes and over the diverged rows. The first
column of each pair is the ratio of worst-case errors of Table~\ref{tab:e3};
the second is the ratio of the mean error over the steps from the first
divergent token on; the step that produced that token still issues the same
query in both loops, so one shared step is included.}
\label{tab:e3post}
\end{table}

\begin{figure}[htp]
\centering
\includegraphics[width=\textwidth]{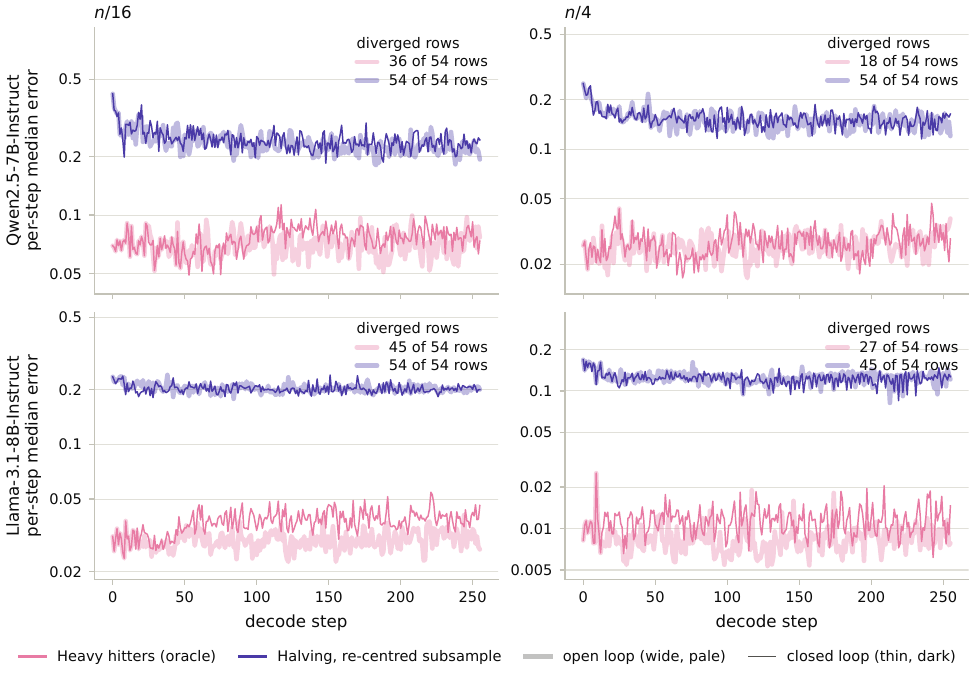}
\caption{E3. Median over the diverged rows of the attention error at each
decode step, the error at a step being the largest over the query heads of
the group; the open loop is the wide pale trace and the closed loop the thin
dark one, for the heavy-hitter oracle and the re-centred variant of the
construction, on both models at $n/16$ and $n/4$, each panel on its own
logarithmic scale. The legend of each panel gives the number of diverged rows
behind each method's pair of traces.}
\label{fig:e3}
\end{figure}

\begin{table}[htp]
\centering\footnotesize
\setlength{\tabcolsep}{4pt}
\begin{tabular}{@{}llrrr@{}}
\toprule
& & \multicolumn{3}{c}{fitted slope of $\log(\text{size})$ against $\rho$, and its range over seeds}\\
\cmidrule(lr){3-5}
Keys & Method & $\dk=2$ & $\dk=8$ & $\dk=64$ \\
\midrule
\vEdFullBody
\bottomrule
\end{tabular}
\caption{E4, all \vEdFits\ fits. Bold marks a slope within $1\pm0.1$
(\vEdFullWithin\ of \vEdFits); $\dagger$ marks a fit resting on at
least one crossing outside the sizes it was fitted over
(\vEdFullDaggered\ fits, Appendix~\ref{app:e4}); an excluded fit prints a
dash.}
\label{tab:e4slopefull}
\end{table}

\begin{table}[htp]
\centering\small
\begin{tabular}{@{}lrrrr@{}}
\toprule
& \multicolumn{4}{c}{worst-case error at a coreset of \vEfSize\ of \vEfN\ pairs}\\
\cmidrule(lr){2-5}
Method & $\rho=1$ & $\rho=2$ & $\rho=4$ & $\rho=6$ \\
\midrule
Uniform sampling & \vEfUnifRa & \vEfUnifRb & \vEfUnifRc & \vEfUnifRd \\
BalanceKV & \vEfBkvRa & \vEfBkvRb & \vEfBkvRc & \vEfBkvRd \\
Kernel halving & \vEfThinRa & \vEfThinRb & \vEfThinRc & \vEfThinRd \\
\textbf{Halving} & \vEfOhRa & \vEfOhRb & \vEfOhRc & \vEfOhRd \\
\textbf{Halving, re-centred subsample} & \vEfOsbRa & \vEfOsbRb & \vEfOsbRc & \vEfOsbRd \\
\bottomrule
\end{tabular}
\caption{E4 read at a fixed requested coreset size: median over the
\vEdSeeds\ seeds, both key families, all three $\dk$ and both $\dv$. Lower is
better.}
\label{tab:e4fixed}
\end{table}

% Every float of this appendix is placed before the bibliography.
\clearpage

% =====================================================================


\begin{thebibliography}{99}
\small

\bibitem{AT07}
R.~J.~Adler and J.~E.~Taylor.
\newblock \emph{Random Fields and Geometry}. Springer, 2007.

\bibitem{ALS21}
R.~Alweiss, Y.~P.~Liu, and M.~Sawhney.
\newblock Discrepancy minimization via a self-balancing walk.
\newblock In \emph{STOC}, 2021; arXiv:2006.14009.

\bibitem{Bai24}
Y.~Bai, X.~Lv, J.~Zhang, H.~Lyu, J.~Tang, Z.~Huang, Z.~Du, X.~Liu, A.~Zeng,
L.~Hou, Y.~Dong, J.~Tang, and J.~Li.
\newblock LongBench: A bilingual, multitask benchmark for long context understanding.
\newblock In \emph{ACL}, 2024; arXiv:2308.14508.

\bibitem{Ball97}
K.~Ball.
\newblock An elementary introduction to modern convex geometry.
\newblock In \emph{Flavors of Geometry}, MSRI Publ.~31, pp.~1--58, 1997.

\bibitem{Ban98}
W.~Banaszczyk.
\newblock Balancing vectors and Gaussian measures of $n$-dimensional convex bodies.
\newblock \emph{Random Structures \& Algorithms}, 12(4):351--360, 1998.

\bibitem{BDGL18}
N.~Bansal, D.~Dadush, S.~Garg, and S.~Lovett.
\newblock The Gram--Schmidt walk: a cure for the Banaszczyk blues.
\newblock In \emph{STOC}, 2018; \emph{Theory of Computing}, 15(21), 2019.

\bibitem{BF81}
J.~Beck and T.~Fiala.
\newblock ``Integer-making'' theorems.
\newblock \emph{Discrete Applied Mathematics}, 3(1):1--8, 1981.

\bibitem{BLPV25}
L.~Bohbot, C.~Letrouit, G.~Peyr\'e, and F.-X.~Vialard.
\newblock Token sample complexity of attention.
\newblock arXiv:2512.10656, 2025.

\bibitem{BR23}
R.~Bozzai and T.~Rothvoss.
\newblock Stronger coreset bounds for kernel density estimators via chaining.
\newblock arXiv:2310.08548, 2023.

\bibitem{CGSDM25}
A.~M.~Carrell, A.~Gong, A.~Shetty, R.~Dwivedi, and L.~Mackey.
\newblock Low-rank thinning.
\newblock arXiv:2502.12063, 2025.

\bibitem{CFIKKP26}
J.~Y.~Chen, Y.~Feng, P.~Indyk, M.~Kapralov, E.~Kochetkova, and B.~Prokhorov.
\newblock Towards tight bounds for streaming attention.
\newblock arXiv:2606.07205, 2026.

\bibitem{CKW24}
M.~Charikar, M.~Kapralov, and E.~Waingarten.
\newblock A quasi-Monte Carlo data structure for smooth kernel evaluations.
\newblock In \emph{SODA}, 2024; arXiv:2401.02562.

\bibitem{CM96}
B.~Chazelle and J.~Matou\v{s}ek.
\newblock On linear-time deterministic algorithms for optimization problems in fixed dimension.
\newblock \emph{Journal of Algorithms}, 21(3):579--597, 1996.

\bibitem{Che78}
S.~Chevet.
\newblock S\'eries de variables al\'eatoires gaussiennes \`a valeurs dans $E\widehat\otimes_\varepsilon F$.
\newblock \emph{S\'eminaire sur la G\'eom\'etrie des Espaces de Banach}, Expos\'e 19, 1978.

\bibitem{Dasigi21}
P.~Dasigi, K.~Lo, I.~Beltagy, A.~Cohan, N.~A.~Smith, and M.~Gardner.
\newblock A dataset of information-seeking questions and answers anchored in research papers.
\newblock In \emph{NAACL}, 2021; arXiv:2105.03011.

\bibitem{DM21}
R.~Dwivedi and L.~Mackey.
\newblock Kernel thinning.
\newblock In \emph{COLT}, 2021; \emph{Journal of Machine Learning Research}, 2024; arXiv:2105.05842.

\bibitem{GCDM26}
A.~Gong, A.~M.~Carrell, R.~Dwivedi, and L.~Mackey.
\newblock Express language modeling.
\newblock arXiv:2606.10944, 2026.

\bibitem{Llama3}
A.~Grattafiori, A.~Dubey, A.~Jauhri, et al.
\newblock The Llama 3 herd of models.
\newblock arXiv:2407.21783, 2024.

\bibitem{HO25}
T.~Haris and K.~Onak.
\newblock Compression barriers for autoregressive transformers.
\newblock In \emph{COLT}, 2025; arXiv:2502.15955.

\bibitem{HKV26}
J.~Haverbeck, C.~Amo~Alonso, A.~F.~Posada-Moreno, S.~Trimpe, and M.~Pavone.
\newblock The risk of KV cache compression: a kernel-discrepancy view of attention.
\newblock arXiv:2607.01520, 2026.

\bibitem{KL19}
Z.~Karnin and E.~Liberty.
\newblock Discrepancy, coresets, and sketches in machine learning.
\newblock In \emph{COLT}, 2019.

\bibitem{KSHZK25}
I.~Kochetkova, A.~Sheth, I.~Han, A.~Zandieh, and M.~Kapralov.
\newblock BalanceKV: streaming attention approximation via discrepancy theory.
\newblock In \emph{NeurIPS}, 2025; arXiv:2502.07861.

\bibitem{KNR99}
I.~Kremer, N.~Nisan, and D.~Ron.
\newblock On randomized one-round communication complexity.
\newblock \emph{Computational Complexity}, 8(1):21--49, 1999.

\bibitem{Li24}
Y.~Li, Y.~Huang, B.~Yang, B.~Venkitesh, A.~Locatelli, H.~Ye, T.~Cai,
P.~Lewis, and D.~Chen.
\newblock SnapKV: LLM knows what you are looking for before generation.
\newblock arXiv:2404.14469, 2024.

\bibitem{LAK26}
E.~Liberty, A.~Andoni, and E.~Kleiner.
\newblock Nearly optimal attention coresets.
\newblock arXiv:2605.05602, 2026.

\bibitem{LPMST15}
D.~Lopez-Paz, K.~Muandet, B.~Sch\"olkopf, and I.~Tolstikhin.
\newblock Towards a learning theory of cause-effect inference.
\newblock In \emph{ICML}, 2015.

\bibitem{Phi13}
J.~M.~Phillips.
\newblock $\varepsilon$-samples for kernels.
\newblock In \emph{SODA}, 2013.

\bibitem{PT18}
J.~M.~Phillips and W.~M.~Tai.
\newblock Near-optimal coresets of kernel density estimates.
\newblock In \emph{SoCG}, 2018; \emph{Discrete \& Computational Geometry}, 63:867--887, 2020.

\bibitem{Qwen25}
Qwen Team (A.~Yang et al.).
\newblock Qwen2.5 technical report.
\newblock arXiv:2412.15115, 2024.

\bibitem{SM26}
T.~Schr\"oder and L.~Mackey.
\newblock WildCat: near-linear attention in theory and practice.
\newblock arXiv:2602.10056, 2026.

\bibitem{SZ26}
Y.~Sui and J.~Zhang.
\newblock The approximation rank of softmax attention: sharp geometric laws and robust interaction dimension.
\newblock arXiv:2608.28150, 2026.

\bibitem{Tai22}
W.~M.~Tai.
\newblock Optimal coreset for Gaussian kernel density estimation.
\newblock In \emph{SoCG}, LIPIcs~224, 63:1--63:15, 2022; arXiv:2007.08031.

\bibitem{TS14}
R.~Tomioka and T.~Suzuki.
\newblock Spectral norm of random tensors.
\newblock arXiv:1407.1870, 2014.

\bibitem{Ver18}
R.~Vershynin.
\newblock \emph{High-Dimensional Probability}. Cambridge University Press, 2018.

\bibitem{Wang26}
X.~Wang.
\newblock How much cache does reasoning need? Depth--cache tradeoffs in KV-compressed transformers.
\newblock arXiv:2604.17935, 2026.

\bibitem{ZHMK24}
A.~Zandieh, I.~Han, V.~Mirrokni, and A.~Karbasi.
\newblock SubGen: token generation in sublinear time and memory.
\newblock arXiv:2402.06082, 2024.

\end{thebibliography}
\end{document}